\documentclass[11pt]{amsart}
\usepackage{geometry}                
\usepackage{url}
\usepackage{graphicx}
\usepackage{amssymb}
\usepackage{epstopdf}
\usepackage{color}
\usepackage{amsmath}
\usepackage{amsthm}
\usepackage{url}
\usepackage[square,numbers]{natbib}
\usepackage{subeqnarray}
\usepackage{graphicx,bbm,setspace,amssymb,amsfonts,amsmath,mathtools,mathrsfs,stmaryrd,amsthm,bm,etoolbox,comment,hyperref,url,caption,subcaption,color,enumitem,booktabs,microtype,nicefrac,lingmacros,nccmath,tree-dvips,tikz-cd}
\usepackage{mathrsfs}
\usepackage[mathscr]{euscript}
 \usepackage{caption}
\usepackage[colorinlistoftodos]{todonotes}
\usepackage[utf8]{inputenc} 
\usepackage[T1]{fontenc}    

\usepackage[noend]{algpseudocode}

\newtheorem{theorem}{Theorem}[section]

\newtheorem{thm-defn}[theorem]{Theorem/Definition}

\newtheorem{prop}[theorem]{Proposition}

\theoremstyle{definition}

\theoremstyle{remark}

\newcommand{\ignore}[1]{}{}

\definecolor{darkred}{rgb}{.7,0,0}

\definecolor{darkgreen}{rgb}{.15,.55,0}

\definecolor{darkblue}{rgb}{0,0,0.7}

\title[{A modified Hegselmann--Krause model with political parties}]{A modified Hegselmann--Krause model for interacting voters and political parties}
\author{Patrick H.~Cahill}
\address{School of Mathematics and Statistics, University of Sydney}
\email{patrick.cahill@sydney.edu.au}
\author{Georg A.~Gottwald}
\address{School of Mathematics and Statistics, University of Sydney}
\email{georg.gottwald@sydney.edu.au}

\begin{document}

\maketitle

\begin{abstract}
The Hegselmann--Krause model is a prototypical model for opinion dynamics. It models the stochastic time evolution of an agent's or voter's opinion in response to the opinion of other like-minded agents. The Hegselmann--Krause model only considers the opinions of voters; we extend it here by incorporating the dynamics of political parties which influence and are influenced by the voters. We show in numerical simulations for $1$- and $2$-dimensional opinion spaces that, as for the original Hegselmann--Krause model, the modified model exhibits opinion cluster formation as well as a phase transition from disagreement to consensus. We provide an analytical sufficient condition for the formation of unanimous consensus in which voters and parties collapse to the same point in opinion space in the deterministic case. Using mean-field theory, we further derive an approximation for the critical noise strength delineating consensus from non-consensus in the stochastically driven modified Hegselmann--Krause model. We compare our analytical findings with simulations of the modified Hegselmann--Krause model.
\end{abstract}

\smallskip


\smallskip



\section{Introduction}
\label{Introduction}
Voting and elections are an essential part of modern democracies. Typically, elections consist of many voters and a few political parties. While the individual behaviour of voters and parties is difficult to predict, simple behavioural rules can produce complex behaviour that resembles known political dynamics. State-based models, such as the voter model, have long been used to model the propagation of opinions across a population \cite{AlbiEtAl17,redner2019reality,lanchier2022consensus}. In reality, voters rarely consider the parties for which they vote to share their opinions exactly. More complex considerations are made about which party's view is closer to their own. In addition, political parties respond to voters through political campaigns and by supporting or opposing policies.  In this paper, we propose an interacting particle system, which treats parties and voters as agents that evolve and interact within an opinion space.

The collective behaviour of interacting particle systems has been extensively studied \cite{AxelrodHamilton81,DeffuantEtAl00,CastellanEtAl09,BlondelEtAl09,motsch2014heterophilious,Chazelle15,Chazelle15b,EasleyKleinbergEtAl,CarrilloEtAl19}. The celebrated Hegselmann--Krause model describes the temporal evolution of opinions of interacting agents in a continuous opinion space \cite{krause1997soziale,HegselmannKrause02}. An agent's opinion evolves towards the average opinion of agents with similar opinions. The Hegselmann--Krause model has been used to model how opinions of experts evolve in small committees \cite{Hegselmann23}. We are interested here in the evolution of opinions of thousands of voters in a political landscape. Such systems are often modelled by agent-based systems such as the Deffuant--Weisbuch model which assumes that each voter interacts at any point in time only with one other randomly chosen voter \cite{DeffuantEtAl00}. The number of voters a single voter interacts with at once has an important impact on the dynamics \cite{UrbigEtAl08}. Both the Hegselmann--Krause model and the Deffuant--Weisbuch model work on the premise of {\em{bounded confidence}}, whereby voters only update their opinion by voters with opinions sufficiently close to theirs. \citet{NugentEtAl24b} recently showed that agent based models such as the Deffuant--Weisbuch model converge to the continuous differential Hegselmann--Krause model, under the assumption that voters interact frequently with only small changes in their opinions with each interaction. Moreover, social media allows for the interaction of many voters through personalized feeds which are a manifestation of bounded confidence, justifying the usage of the Hegselmann--Krause model to represent an electorate of voters in a political landscape. Differential equation models such as the Hegselmann--Krause model offer the advantage of being amenable to mathematical analysis. We hence adopt here the framework of the Hegselmann--Krause model. To account for any, possibly irrational, individual behaviour of voters additive noise is added to the dynamics of the Hegselmann--Krause model. Such dynamics may give rise to clustering dynamics and exhibits a phase transition from disagreement to consensus formation with decreasing noise strength. In recent years, several modifications to the original Hegselmann--Krause model were proposed to include grouped populations, self-belief and heterogeneity \cite{FuEtAl15,ChenEtAl17,HanEtAl19,ZhuEtAl23}, the effect of leadership voters \cite{WongkaewEtAl15,ZhuXie17,AtasEtAl21}, inertial effects \cite{ChazelleWang17} and underlying social network structure \cite{NugentEtAl24}. For recent reviews of opinion spreading models see \cite{CastellanEtAl09,LiuChen15}. These models only describe the mutual interaction of voters and neglect the dynamics of parties. Parties shape the opinion space of voters with a strong impact on voters' behaviour \cite{clive1989,ideology2004,kim2016}, and similarly voters correct and influence parties who are competing for votes \cite{zechmeister2003sheep,adams2012causes}. The competition for votes also leads to parties delineating themselves from other parties and carving out opinion space for themselves \cite{laver2005policy,adams2012causes}. We introduce an extended noisy Hegselmann--Krause model in which the voter dynamics are augmented by a dynamical model for parties which takes into account these interactions. Voters and parties are represented in a $d$-dimensional opinion space, representing their respective political orientations towards $d$ separate political issues \citep{davis1970expository,ray2002new,falck2020measuring,arian2017election,kim2016}. Being able to describe complex voter and party behaviour as dynamics in an opinion space requires the availability of information across the whole electorate, which has been made more effective with the advent of social media \cite{welch2016sheep}.  While opinion dynamics models such as the voter model \cite{redner2019reality} describe the behaviour of voters in a two-party system, the study of multi-party systems is rarer and typically restricted to a static situation. Such static spatial models have been used for data-driven predictions of election outcomes \cite{quinn1999voter,kedar2005moderate}. From statistical physics, the Potts model \cite{wu1982potts} has also been used to model elections with multiple candidates \cite{nicolao2019potts}, but these models represent parties as states that voters may or may not support, instead of agents that evolve in time. Our proposed modified noisy Hegselmann--Krause model allows for the study of the temporal evolution of a multi-party system, giving rise to rich interactions and distinct behaviours. Numerous countries, including Australia and Germany, which traditionally were dominated by two main parties have evolved into a multi-party system \cite{paun2011after}. Other countries, such as the Netherlands, traditionally operate under a multi-party system \cite{blondel1968party}.\\ 

The modified noisy Hegselmann--Krause model will be shown to exhibit known voter and party dynamics such as party-bases, swing voters, disaffected voters as well as a state of voter consensus in which parties and voters decouple and voters assemble around a single opinion state whereas parties arrange themselves in separate areas in opinion space. We show that, depending on the strength of the interactions between voters and parties, there exist metastable states where voters are attracted to different parties which occupy different regions in opinion space, before eventually unanimous consensus occurs and voters and parties collapse to a single localized area in opinion space. Such a state of unanimous consensus is, of course, not realistic. Using linear stability analysis and mean-field theory we analytically provide sufficient conditions for consensus and find an expression for the critical noise strength above which no consensus is possible. We find that the existence of parties is conducive to the formation of clusters. Moreover, the stabilising effect of additional party dynamics increases with the underlying dimension of the opinion space, i.e. with the number of topics that dominate the political debate.\\

The paper is organized as follows: In Section~\ref{sec:model} we propose our new model combining consensus dynamics of voters with party dynamics. In Section~\ref{sec:examples} we show how the inclusion of parties recovers to known political scenarios such as party-base formation, disaffected voters and swing voters. In Section~\ref{sec:consensus} we numerically explore the evolution to consensus. For the deterministic case we provide a sufficient condition for consensus formation in Section~\ref{sec:app1}. In Section~\ref{sec:meanfield} we provide analytical results on a noise-induced phase transition from random voter dynamics to consensus. We conclude in Section~\ref{sec:disc} with a discussion.


\section{A modified Hegselmann--Krause model incorporating party dynamics}
\label{sec:model}

The original Hegselmann--Krause model is concerned with the interaction of $N_v$ voters $v_i$ \cite{HegselmannKrause02,PinedaEtAl09,motsch2014heterophilious}. Voters are represented by their position in some $d$-dimensional opinion space. A point in the opinion space can be thought of as representing a set of views. For $d=2$, the opinion space is referred to as the political compass \citep{davis1970expository,ray2002new,falck2020measuring,arian2017election} with the two coordinates of $v_i$ representing, for example, a voter's political leaning on the spectrum from left-wing to right-wing social views and their economic preferences on the spectrum from libertarian to socialist. For $d>2,$ each dimension might represent an opinion on a specific issue \cite{stoetzer2015multidimensional}. In the noisy Hegselmann--Krause model, the dynamics of the voters $v_i(t)\in \mathbb{R}^d$ in this opinion space is governed by the following weakly interacting particle system
\begin{align}
\label{eq:HK0}
    dv_i = \frac{1}{N_v}\sum_j^{N_v} \phi(v_i,v_j)(v_j - v_i)dt + \sigma dW^i_t.
\end{align}
Here, $W^{i}$ are $d$-dimensional independent standard Brownian motions representing uncertainty with strength $\sigma$. The interaction kernel $\phi$ depends on the Euclidean distance in opinion space between the voters and encodes how voters $v_j$ with different opinions to voter $v_i$ can influence voter $v_i$. We choose here an isotropic kernel which is compactly supported on $[0,R_{vv}]$ with 
\begin{align}
    \phi(v_i,v_j) = \phi\left( x =\frac{||v_i-v_j||}{R_{vv}} \right) = \begin{cases}1 & x< 1 \\ 0 & \text{else}\end{cases},
     \label{eq:phi}
\end{align}
where $R_{vv}$ is the interaction radius of voters. The interaction kernel \eqref{eq:phi} lets a voter $i$ interact equally with all voters that are within a radius of $R_{vv}$ in opinion space and ignores interactions between voters that are too distant in opinion space (often called bounded confidence). Other choices of interaction kernels are used, for example, allowing for heterophilious dynamics \cite{motsch2014heterophilious}. Noting that only differences in opinion enter the dynamics, we restrict the opinion space without loss of generality to the $d$-dimensional unit square $[0,1]^d$. Note that any domain size of the opinion space can be absorbed by scaling the interaction radius $R_{vv}$.

\noindent
One of the key features of Hegselmann--Krause models is the emergence of consensus, which occurs when voters collapse in opinion space to a single opinion cluster, the size of which depends on the strength of the noise $\sigma$ \citep{motsch2014heterophilious, garnier2017consensus, wang2017noisy}. We remark that in the noiseless case $\sigma=0$ unanimous consensus occurs with all voters converging to the same point in opinion space. In particular, the system exhibits a phase transition \citep{Fortunato05,zagli2023dimension}; there exists a critical noise strength $\sigma_c$ such that for $\sigma < \sigma_c$ the noisy Hegselmann--Krause model \eqref{eq:HK0} asymptotically approaches consensus, whereas for $\sigma > \sigma_c$ no consensus is reached and instead voters behave as $N_v$ independent stochastic processes \cite{wang2017noisy}.  

\noindent
The noisy Hegselmann--Krause model \eqref{eq:HK0} only models the dynamics of individual voters and is not designed to model the dynamics of a political landscape involving political parties. Therefore, it cannot be used to serve as a model for elections. We will extend the classical Hegselmann--Krause model to couple the preferences and dynamics in the $d$-dimensional opinion space of $N_v$ individual voters $v_i\in \mathbb{R}^{d}$ and a finite number $N_p$ of parties $p_\alpha\in \mathbb{R}^{d}$. We make the reasonable assumption that voters will vote for political parties that share similar opinions to them and that both voters and parties, share the same $d$-dimensional opinion space \cite{carson2016wedge}. We make the following assumptions about the interactions between voters and parties: voters are affected by their interactions with other voters and also with parties. In particular, voters attract one another according to the right-hand side drift term of the Hegselmann--Krause model \eqref{eq:HK0} which we denote as $F_{vv}$. The effect is that voters move towards the mean of the voter opinion within a $d$-dimensional sphere with radius $R_{vv}$ centred on the voter. Voters are further attracted by parties which are sufficiently close to them in opinion space \cite{kim2016,adams2012causes}. The associated force can be thought of as a leadership effect and will be denoted as $F_{pv}$ \cite{laver2005policy,clive1989}. Voters who have a similar opinion to a political party will shape their beliefs to conform to the leaders with whom they are similar. Conversely, parties are affected by the voters and by other parties \cite{zechmeister2003sheep,welch2016sheep}. It is reasonable to assume that political parties aim to maximise the number of votes they receive. A natural way of achieving this is to minimise the distance between them and their potential voters. However, parties also have ``identities'' - for example, some are considered left-wing and others centrist - which means that the voters they seek to be close to (or best represent) are within a certain distance $R_{vp}$ in opinion space. The smaller the value of $R_{vp}$ the less a party's view is influenced by far-away voters. We denote the associated attracting force between parties and voters by $F_{vp}$. Finally, it is reasonable to assume that political parties want to differentiate themselves from each other \cite{elias2015position}. Why vote for one party if there is another party that is exactly the same? Hence, we assume a repulsive interactive force $F_{pp}$ between political parties. Figure \ref{fig:forcesdiagram} is a schematic of these four interactions.\\

\begin{figure}[!htb]
    \centering
    \includegraphics[width=0.8\textwidth]{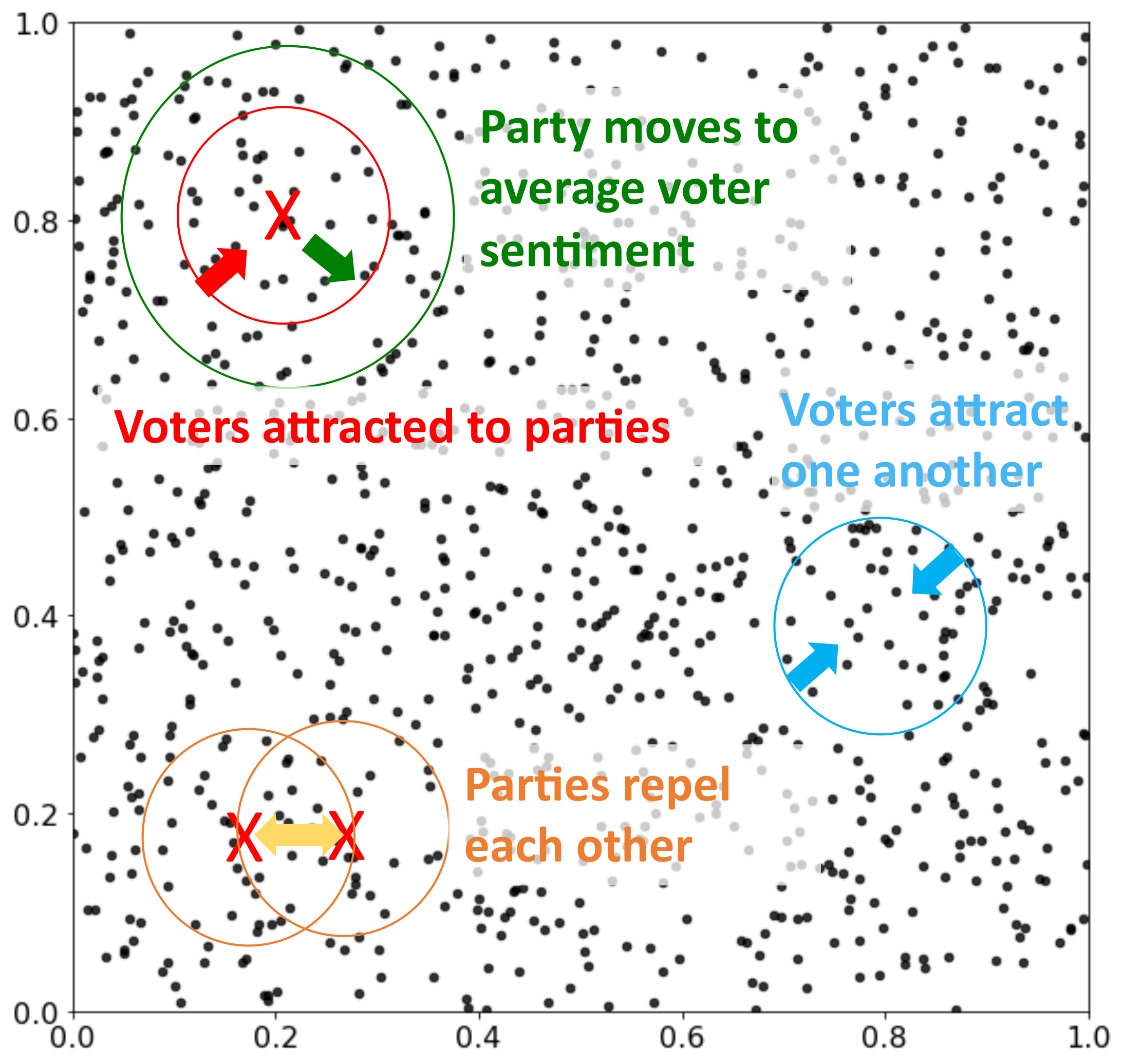}
    \caption{Sketch of the different interaction forces between voters (black dots) and parties (red crosses) in a $2$-dimensional opinion space. The blue force corresponds to $F_{vv}$ with \eqref{eq:F_vv}; the red force corresponds to $F_{pv}$ with \eqref{eq:F_pv}; the green force corresponds to $F_{vp}$ with \eqref{eq:F_vp}; the orange force corresponds to $F_{pp}$ with \eqref{eq:F_pp}.}
    \label{fig:forcesdiagram}
\end{figure}

\noindent
Summarizing the assumed interactions outlined above we propose the following coupled voter-party Hegselmann--Krause model 
\begin{align}
\label{eq:HKa}
    dv_i &= \mu_{vv}F_{vv}(v_i,v) \, dt + \mu_{pv}F_{pv}(v_i,p) \, dt + \sigma_v dW_t ^i\\
    dp_\alpha &= \mu_{vp}F_{vp}(p_\alpha,v)\, dt - \mu_{pp}F_{pp}(p_\alpha,p)\, dt + \sigma_p dW_t ^\alpha,
\label{eq:HKb}
\end{align}
with the interaction forces
\begin{align}
    F_{vv}(v_i,v) &= \frac{1}{N_v}\sum_{j=1}^{N_v}\phi\left(\frac{||v_j - v_i||}{R_{vv}}\right)(v_j-v_i),
    \label{eq:F_vv}\\
    F_{pv}(v_i,p) &= \frac{1}{N_p}\sum_{\beta=1}^{N_p}\phi\left(\frac{||p_\beta - v_i||}{R_{pv}}\right)(p_\beta-v_i),
    \label{eq:F_pv}\\
    F_{vp}(p_\alpha,v) &= \frac{1}{N_v}\sum_{j=1}^{N_v}\phi\left(\frac{||v_j - p_\alpha||}{R_{vp}}\right)(v_j-p_\alpha),
    \label{eq:F_vp}\\
    F_{pp}(p_\alpha,p) &= \, \frac{1}{N_p}\sum_{\beta=1}^{N_p}\phi\left(\frac{||p_\beta - p_\alpha||}{R_{pp}}\right)(p_\beta-p_\alpha).
    \label{eq:F_pp}
\end{align}
Here $W^i(t)$ for $i=1,\hdots, N_v$ and $W^\alpha(t)$ for $\alpha=1,\hdots, N_p$ are independent standard Brownian motion processes which represent unpredictable changes in a particular agent's opinion with the strength given by the diffusion coefficients $\sigma_v$ and $\sigma_p$. In the following, Latin alphabet sub- and superscripts refer to voters and Greek alphabet sub- and superscripts refer to political parties. The forces are assumed to have compact support in a $d$-dimensional sphere with positive radii $R_{vv}$, $R_{pv}$, $R_{vp}$ and $R_{pp}$, respectively (cf. \eqref{eq:phi} for $R_{vv}$). It is reasonable to assume that $R_{vp}>R_{vv}$ as parties typically consider a much larger contingency of voters than individual voters do. Further, we assume that $R_{pp}$ is small as parties only repel each other when they become too close in opinion space to delineate themselves from each other. The strength of the respective forces is controlled by the strength parameters $\mu_{vv},\mu_{pv},\mu_{vp},\mu_{pp}\geq 0$ which are assumed to be constant in time. Note that for $\mu_{vv}=1, \mu_{pv}=0$ the evolution of the voters \eqref{eq:HKa} reduces to the original noisy Hegselmann--Krause model \eqref{eq:HK0}. 

We choose a normalization with $1/N_v$ and $1/N_p$ for the interaction forces \eqref{eq:F_vv}-\eqref{eq:F_pp}. We have done so to allow for a well-defined mean-field limit (see Section~\ref{sec:meanfield}). Other normalisations have been used in the literature and may lead to different behaviour \cite{NugentEtAl24b}.

Here we have chosen force strengths and interaction radii to be equal for all agents. In reality, some parties might be more ``charismatic'' to certain voters, and similarly, some voters might possess especially strong attractiveness to others (such as between family members). Moreover, some voters might be more open to different opinions and so operate with larger interaction radii than others. To include such agent-specific parameters in our model \eqref{eq:HKa}--\eqref{eq:HKb} one can make the force strengths and interaction radii voter and party dependent. For example, we could allow for a force strength $\mu^\alpha_{{pv}_i}$ modelling the leadership effect of party-$\alpha$ on voter-$i$. For simplicity, we do not consider such effects in this work.

We employ here periodic boundary conditions, except in Section \ref{sec:consensus}, where we use reflective boundary conditions. This significantly simplifies the mean-field analysis in Section~\ref{sec:meanfield}. At first sight, periodic boundary conditions do not seem realistic as they imply that a voter can switch from one end of the opinion spectrum to the other. However, this scenario is known in political theory as the horseshoe theory \cite{Faye} which claims that the extreme left and the extreme right share similar opinions. Nevertheless, this theory remains contested, see for example \cite{mayer2011extremes}. We remark that our code, available on the GitHub repository \url{https://github.com/PatrickhCahill/ModifiedHegselmannKrauseModel} allows for reflective and periodic boundary conditions.

Both the periodic boundary conditions and the reflective boundary conditions ensure that the dynamics remains in a bounded region in opinion space. A bounded domain corresponds to the realistic assumption that there is some maximal extremeness that an agent can hold on a given topic. Bounding the dynamics is particularly important in our model, which includes repulsive forces between parties. A strong repulsive force between parties can lead to parties and their attached voter cluster leaving a prescribed domain in opinion space, evolving in ever more extreme regions in opinion space.


\section{Political scenarios in the modified Hegselmann--Krause model}
\label{sec:examples}

We now show that the modified Hegselmann--Krause model \eqref{eq:HKa}-\eqref{eq:HKb} is able to reproduce several types of voter behaviour in a political landscape with competing parties. We restrict here to a $1$-dimensional opinion space.

Let us start from a standard scenario where each party has their own loyal party-base, sometimes referred to as "core voters" \cite{cox200913}. This politically stable situation, which we term \textit{party-base}, occurs when each party occupies a particular region in opinion space and all parties are sufficiently far apart in opinion space from each other with 
\begin{align}
||p_\alpha-p_\beta||\ge 2 {\rm{max}}(R_{pv},R_{vv})
\label{eq:condpartybase}
\end{align}
for all parties $\alpha\neq \beta$ so that parties do not compete for voters and interactions between voters attached to different parties is suppressed. We further require $\sigma_{v} \gg \sigma_p$ to ensure the voter dynamics occurs much faster than the party dynamics. In such situations, clusters may form around each party. The number of voters $N^{(c)}_v$ attached to each cluster $c$ depends on the initial political opinions of the voters. 

The size of a cluster $\delta_{\rm{cl}}$ is defined as the average diameter of a collection of agents that are not interacting with any other agents, allowing (at least temporarily) for a coherent structure. Heuristically, it is clear that increasing the noise will increase the cluster size, while increasing the force strengths will result in a more tightly packed cluster. We will now find an approximation for the cluster size $\delta_{\rm{cl}}$, in the case where the interaction radii $R_{vv}$ and $R_{pv}$ cover the cluster. Suppose that a voter cluster of $N_v^{(c)}$ voters is centred around $w \in [0,1]$. Further, assume that there are $N_{p}^{(c)}$ parties interacting with the voter cluster and that they are also centred around the same $w$. As the dynamics depends only on the differences between the positions of the relative agents, we shift, for simplicity, the domain such that the voter cluster is centred at $w=0$. Finally, assume that no other agents are interacting with the cluster. With these assumptions the forces acting on the voters become linear with
\begin{align}
    F_{vv}(v_i;v) &= \mu_{vv}\frac{1}{N_{v}}\sum_{j=1}^{N_v}\phi\left(\frac{||v_j-v_i||}{R_{vv}}\right)(v_j - v_i) \nonumber \\
    &= -\mu_{vv}\frac{N_{v}^{(c)}}{N_{v}}v_i
\end{align}
and
\begin{align}
    F_{pv}(v_i;p) &= \mu_{pv}\frac{1}{N_{p}}\sum_{\beta=1}^{N_p}\phi\left(\frac{||p_\beta-v_i||}{R_{pv}}\right)(p_\beta - v_i) \nonumber \\
    &= -\mu_{pv}\frac{N_{p}^{(c)}}{N_{p}}v_i.
\end{align}
The voter dynamics \eqref{eq:HKa} reduces to an Ornstein-Uhlenbeck process with
\begin{align} 
dv_i = - \left( \frac{N^{(c)}_v}{N_v}\mu_{vv} + \frac{N^{(c)}_p}{N_p} \mu_{pv} \right) v_i + \sigma_v dW^i_t.
\label{eq:vOU}
\end{align}
The asymptotic mean is zero and the standard deviation is given by
\begin{align}
{\rm{std}}_{\rm{OU}} = \frac{\sigma_v}{\sqrt{2}}\left(\frac{N^{(c)}_v}{N_v}\mu_{vv} + \frac{N^{(c)}_p}{N_p}\mu_{pv}\right)^{-\frac{1}{2}}.
\label{eq:stdOU}
\end{align}
Defining a cluster as the region which covers approximately $95\%$ of the voters, we obtain the heuristic estimate for the cluster size 
\begin{align}
\delta_{\rm{cl}} = 4 \, {\rm{std}}_{\rm{OU}}.
\label{eq:deltacluster}
\end{align}
It is seen from \eqref{eq:deltacluster} that the cluster size decreases for decreasing noise strength. In the limit $\sigma_v\to 0$, the cluster size is $\delta_{\rm{cl}}=0$ corresponding to a Dirac measure in opinion space. In the case that $R_{vv}$ and $R_{pv}$ do not cover the cluster, one may still observe transient short-lived groupings of voters (see Figure~\ref{fig:3party-disaffected_voters} for an example in a three-party system).

In the case of a party-base cluster centred around a single party, we set $N^{(c)}_p=1$. Figure~\ref{fig:2party-political_base} shows an example for such party-base formation with $N_v = 1,000$ voters which initially are distributed uniformly across $[0,1]$ and with two parties which are initially located in opinion space at $p_1(0)=0.25$ and $p_2(0)=0.75$. The strengths of the voter forces were chosen as $\mu_{vv}=\mu_{pv}=1$ and those of the party forces as $\mu_{vp}=0.03$ and $\mu_{pp}=0.01$. The interaction radii are $R_{vv}=0.15$, $R_{pv}=0.1$, $R_{pp}=0.1$ and $R_{vp}=0.3$. The noise strengths are $\sigma_v = 0.04$ and $\sigma_p = 0.005$, which ensures that the voter dynamics occurs much faster than the party dynamics. Figure~\ref{fig:2party-deltacluster_check} shows the standard deviations of voters in the two-party-base clusters corresponding to the party-base case depicted in Figure~\ref{fig:2party-political_base}. We show results of the numerical simulation of the modified Hegselmann--Krause model \eqref{eq:HKa}-\eqref{eq:HKb} as well as the prediction of the cluster size $\delta_{\rm{cl}}$ given by \eqref{eq:deltacluster}. Since the voters form two approximately equal clusters we set $N_{v}^{(c)} \approx \frac{N_v}{2}$ and $ N_{p}^{(c)}=1$ implying $\delta_{\rm{cl}}=4\, \rm{std}_{\rm OU}=0.113$. Note that here the cluster is covered by the interaction radii with $\delta_{\rm{cl}}<2R_{pv}<2R_{pv}=0.3$. To numerically estimate $\delta_{\rm cl}$ for the full modified Hegselmann--Krause model, we set $N_v^{(c)}$ to be the number of voters which are within a distance of $0.3$ away from the respective parties $p_\alpha$. The threshold $0.3$ is chosen because the region $[p_\alpha-0.3, p_\alpha+0.3]$ for $\alpha = 1,2$ contains the cluster around $p_\alpha$ since the interaction radii are sufficiently small with $R_{vv}=0.15, R_{pv}=0.1$. The region is also sufficiently small that it does not contain any other clusters. Figure~\ref{fig:2party-deltacluster_check} shows that our expression \eqref{eq:deltacluster} well approximates the observed cluster size as measured by the standard deviations with ${\rm{std}}_{\rm{OU}}=0.113/4\approx 0.028$.

Such stable party-bases break down if parties increase their interaction radius $R_{pv}$ and hence are competing with each other over voters which are affected by two or more parties, or equivalently by parties moving towards each other in opinion space, induced by their nonzero diffusion $\sigma_p$. In particular, consider two parties $p_\alpha$ and $p_\beta$ that satisfy 
\begin{align}
||p_\alpha-p_\beta|| < {\rm{max}}(R_{pv},R_{vv}),
\label{eq:condswing}
\end{align}
and where any other parties are sufficiently far away in opinion space such that they do not interact with any voters attracted to parties $p_\alpha$ and $p_\beta$. In this case, clusters of swing voters which are confined between two parties will emerge. In elections such swing voters vote for either of the two close parties, depending on small hard to predict preferences \cite{cox200913}. The state of swing voters is illustrated in Figure~\ref{fig:2party-swing_voters}. Here the parameters are as for the party-base scenario but for a larger party interaction radius $R_{pv}=0.35$. The size of swing voter clusters $\delta_{\rm{cl}}$ is also given by \eqref{eq:deltacluster} with $N^{(c)}_p=2$ because the party above and below are both within the interaction radius of the voter cluster. We find $\delta_{\rm{cl}}=0.094$ which approximates well the observed cluster size with ${\rm{std}}_{\rm{OU}}=0.094/4$. 

When the party dynamics is faster - or alternatively the voter dynamics is slower - we may observe more competitive behaviour between the parties. Figure~\ref{fig:2party-political_comp_1} shows clusters of swing voters that are slowly entrained by one of the two parties. We further note the coexistence of party-bases and swing voter clusters, although neither is stable. Here $\sigma_v = 0.017$ ensures the voter behaviour is slower than in Figure~\ref{fig:2party-political_base}. Swing voters decide to join a particular party-base either by their individual stochastic slow exploration of the opinion space or by parties moving towards them on a faster time scale. The latter scenario may be viewed as a form of party competition to attract more voters. 

Political competition can lead to intricate complex transitions between the scenarios depicted in Figures~\ref{fig:2party-political_base}-\ref{fig:2party-political_comp_1}. 
In particular, if two parties $p_\alpha$ and $p_\beta$ satisfy $\max(R_{pv},R_{vv}) < ||p_\alpha - p_\beta || < 2 \max(R_{pv},R_{vv})$, neither the criterion for the existence of party-base clusters \eqref{eq:condpartybase} nor the criterion of the existence of swing voter clusters \eqref{eq:condswing} are met and it is possible for both states to coexist and interact. For example, Figure~\ref{fig:2party-political_comp_2} 
shows a party-base forming around each party while two swing voter clusters form between the two parties around $0.5$ and around $0$ in opinion space. The cluster around $0$ exists due to the periodic boundary conditions. 

Yet another important political scenario occurs when parties leave a region in opinion space unoccupied and this region is larger than the interaction radius $R_{pv}$ of any party. In this case voters can occupy this region, attracted by voter-voter interactions only. The voter dynamics degenerates to the original noisy Hegselmann--Krause model \eqref{eq:HK0} and a voter-only cluster forms of size $\delta_{\rm{cl}}\approx 4\sigma_v/\sqrt{2\mu_{vv}N_{v}^{(c)}/N_v}$ provided that $R_{vv}>\delta_{\rm{cl}}/2$ \cite{wang2017noisy} (cf \eqref{eq:deltacluster} and \eqref{eq:stdOU} with $N_p^{(c)}=N_p=0$). 
In this scenario voters may be viewed as disaffected voters who feel unrepresented by political parties, and whose support the parties are not interested in earning. In trying to attract those politically far away voters parties may risk losing voters that are more aligned with a party's current views. Figure~\ref{fig:2party-disaffected_voters} shows an example with disaffected voters. Here the same parameters are used as for the party-base clusters in Figure~\ref{fig:2party-political_base} but with different initial conditions for the parties allowing for unoccupied political space. The cluster of disaffected voters depicted in Figure~\ref{fig:2party-disaffected_voters} contains $N_v^{{c}}\approx 370 $ voters. The estimated cluster size is hence $\delta_{\rm{cl}}\gtrapprox 4\sigma_v/\sqrt{2\mu_{vv} 0.37}\approx 0.19$ which is consistent with the observed cluster size seen in Figure~\ref{fig:2party-disaffected_voters} of $0.21$.

Lastly, we report on a scenario in which voters abandon their party-base and aggregate to a single cluster. This occurs when $R_{vv} > R_{pv}$ and $\mu_{vv} \gtrapprox \mu_{pv}$. In this case the attraction between voters is more prominent and the political leadership effect of the parties becomes irrelevant to the behaviour of the voters. Voters collapse to a single state, typically around a single party. We coin this collapsed state {\em{voter consensus}}. Note that this is achieved with all parties having the same force strength $\mu_{pv}$ and without one party having a significantly higher attraction force than the other parties. This may be viewed as voter consensus because the effect occurs when the force between voters becomes more prominent and the political leadership effect of the parties becomes irrelevant to the behaviour of the voters, which collapse to a single state, typically around a single party. An example of such a voter consensus behaviour is seen in Figure~\ref{fig:2party-voter consensus} where a single party-base cluster forms around party $p_1$. We remark that if the condition for swing voters \eqref{eq:condswing} is satisfied, the final state will be a swing voter cluster. We have presented here results for a traditional  two-party system. The same scenarios also occur in multi-party systems. We present in the Appendix an example for a three-party system.

We remark that typically the scenarios described above are transitory and the stochasticity and/or party dynamics lead to a break up of these structures. The inclusion of heterogeneous and possible time-varying equation parameters, in particular of the force strengths, would allow for an even richer set of scenarios and transitions between them. We further note that neither the party-base behaviour, nor the swing voter behaviour are realistic. It is highly unlikely that all voters will be loyal to a single party or will not be captured by a single party. Rather a mix of behaviours -- such as in the political competition scenario -- seems more realistic. In the next Section we will see how and when our model allows for a final state of unanimous consensus, which however, unlike the transient scenarios described above, is not a likely political scenario. 

\begin{figure}[htpb!]
    \centering
        \begin{subfigure}[htbp]{0.48\textwidth}
            \includegraphics[width=\linewidth]{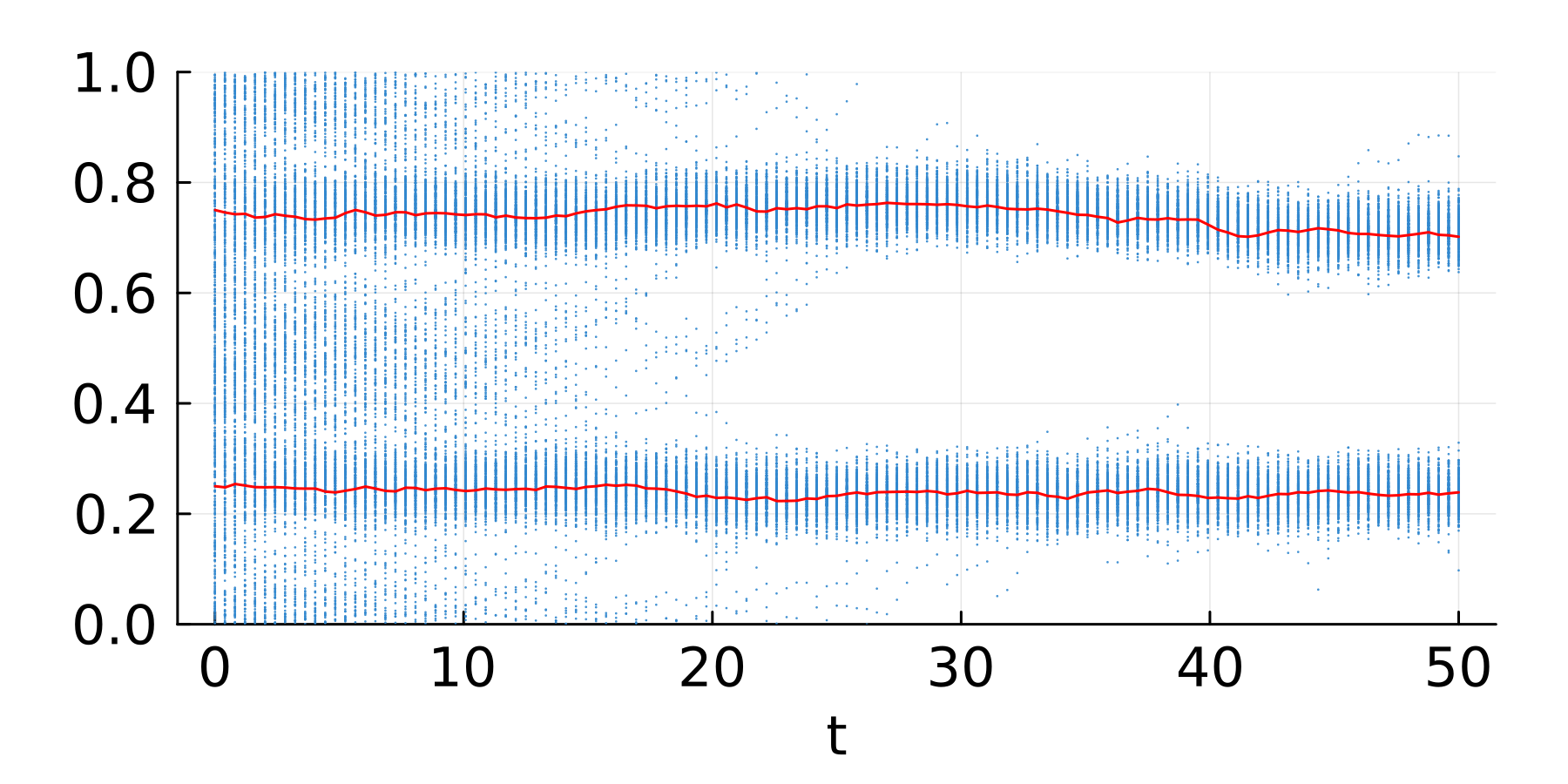}
            \caption{\textit{Party-base}}
            \label{fig:2party-political_base}
    \end{subfigure}
    \hfill
    \begin{subfigure}[htbp]{0.48\textwidth}  
        \centering
        \includegraphics[width=\textwidth]{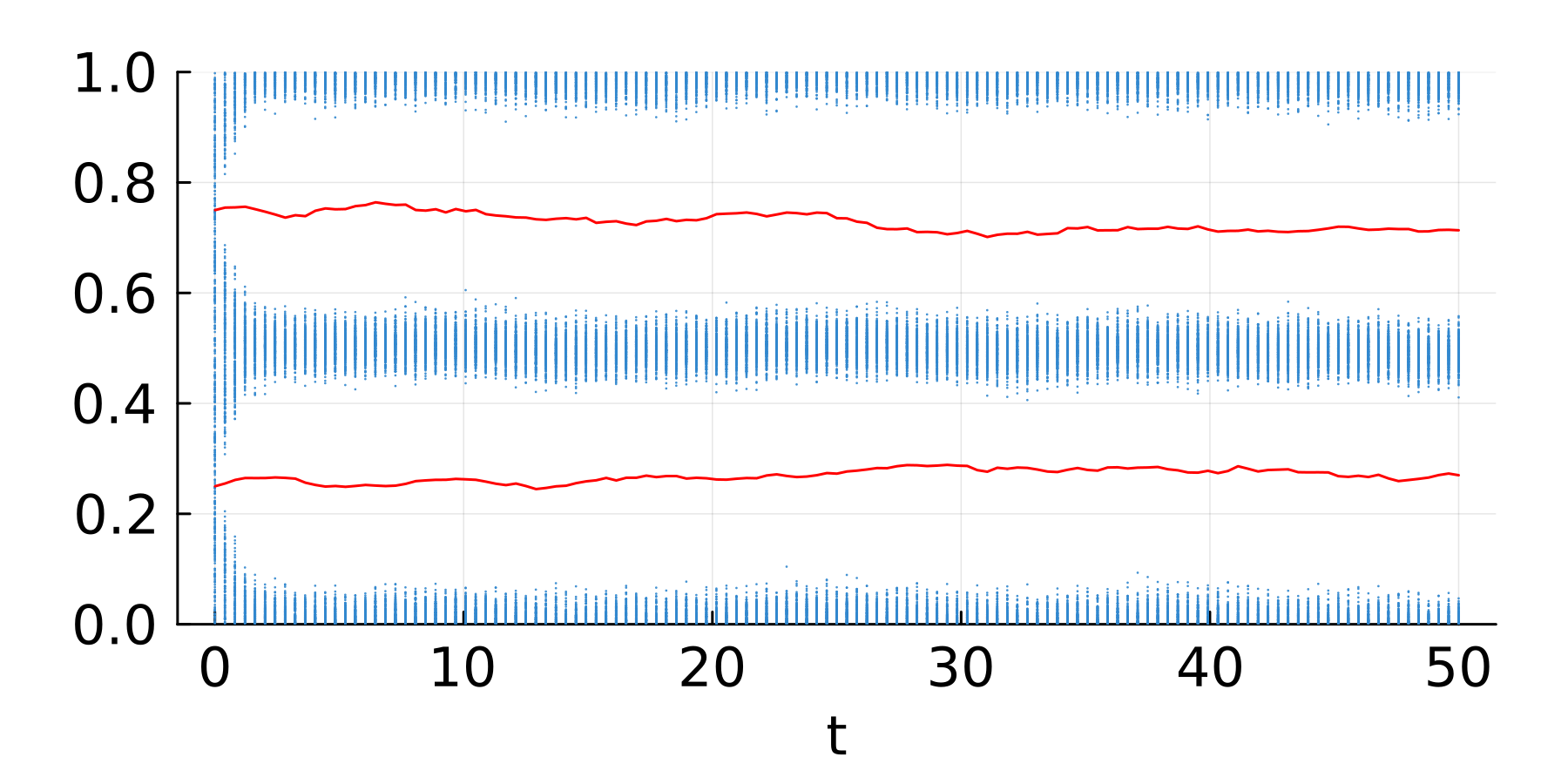}
        \caption{\textit{Swing voters}. Here $R_{pv} = 0.5$.}
        \label{fig:2party-swing_voters}
    \end{subfigure}
    \vskip\baselineskip
    \begin{subfigure}[htbp]{0.48\textwidth}
        \centering
        \includegraphics[width=\textwidth]{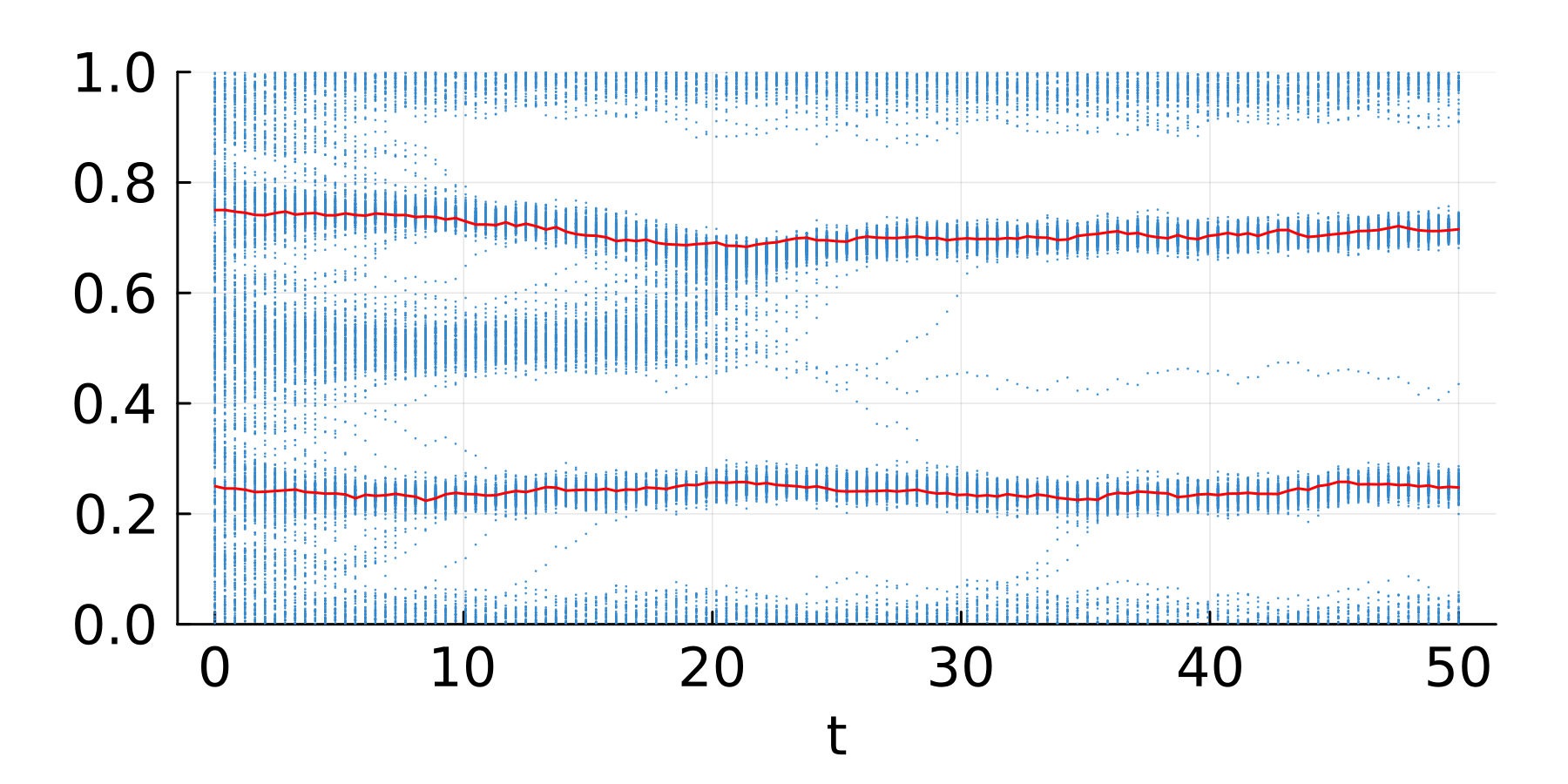}
        \caption{\textit{Political competition}. Here $\sigma_v = 0.017.$}
        \label{fig:2party-political_comp_1}
    \end{subfigure}
    \hfill
    \begin{subfigure}[htbp]{0.48\textwidth}  
        \centering
        \includegraphics[width=\textwidth]{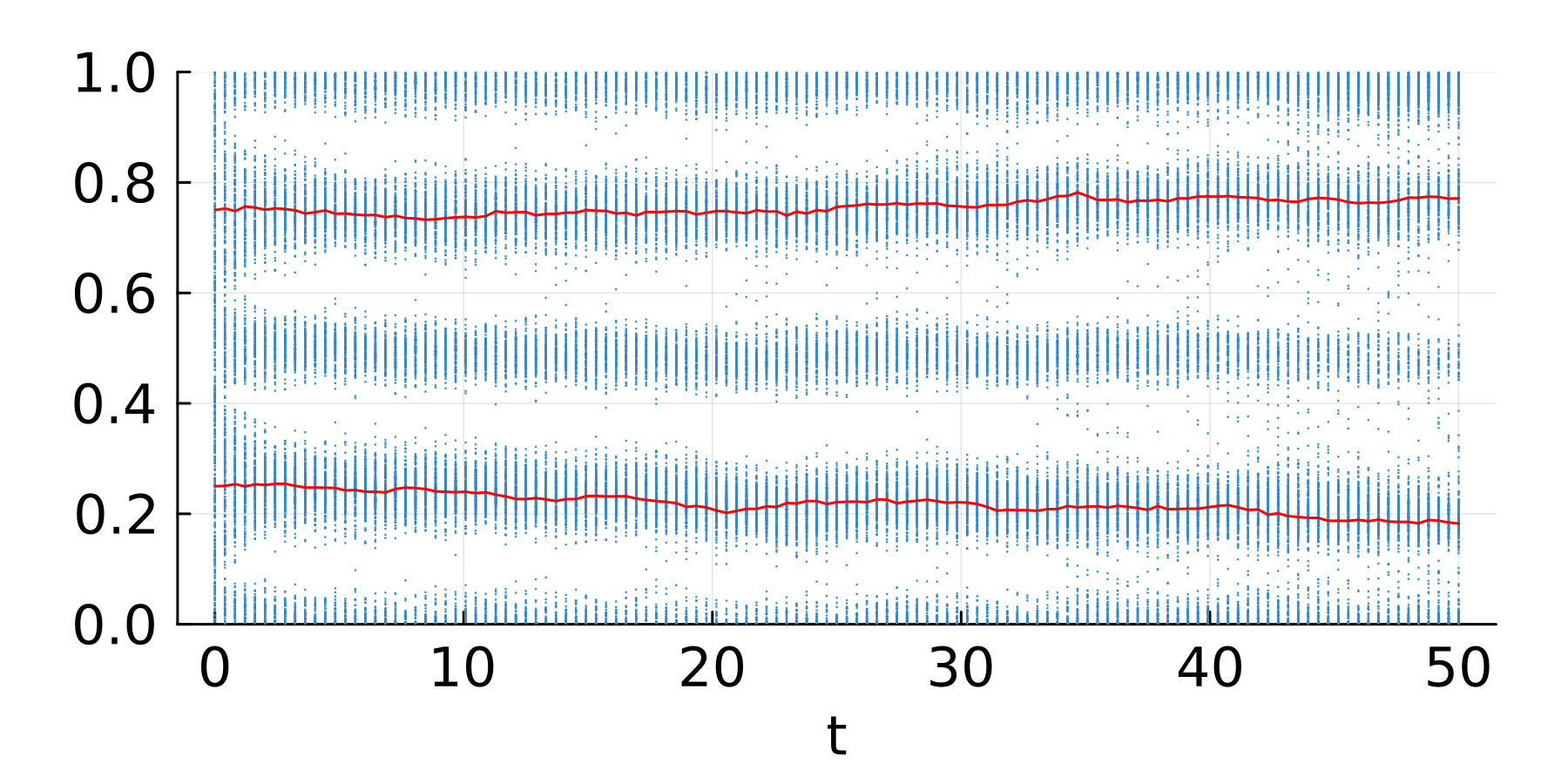}
        \caption{\textit{Political competition}: Here $R_{pv}= 0.35$.}
        \label{fig:2party-political_comp_2}
    \end{subfigure}
    \vskip\baselineskip
    \begin{subfigure}[htbp]{0.48\textwidth}   
        \centering 
        \includegraphics[width=\textwidth]{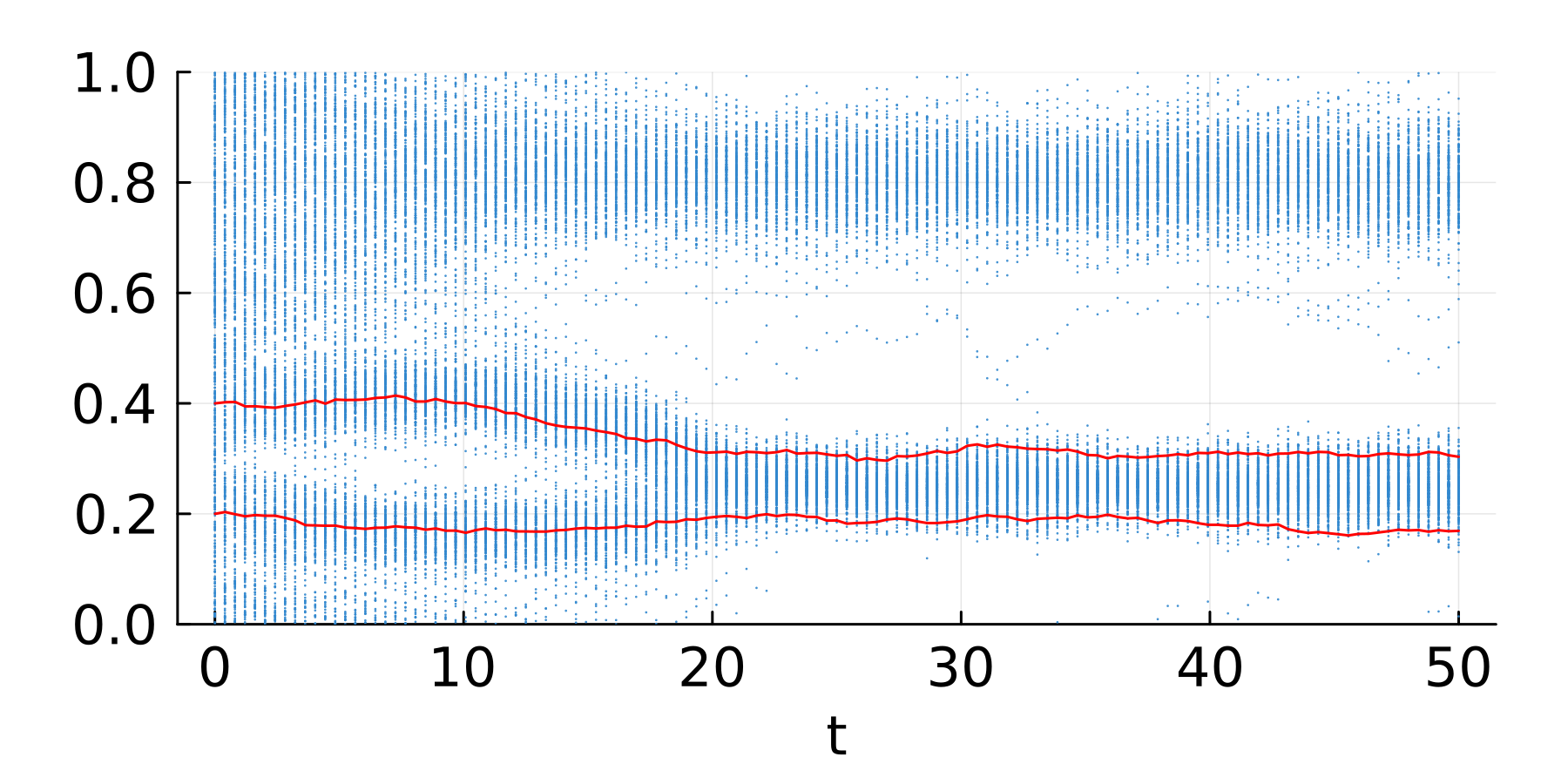}
        \caption{\textit{Disaffected voters}: Here $p_1(0)=0.2$ and $p_2(0)=0.4$.} 
        \label{fig:2party-disaffected_voters}
    \end{subfigure}
    \hfill
    \begin{subfigure}[htbp]{0.48\textwidth}   
        \centering 
        \includegraphics[width=\textwidth]{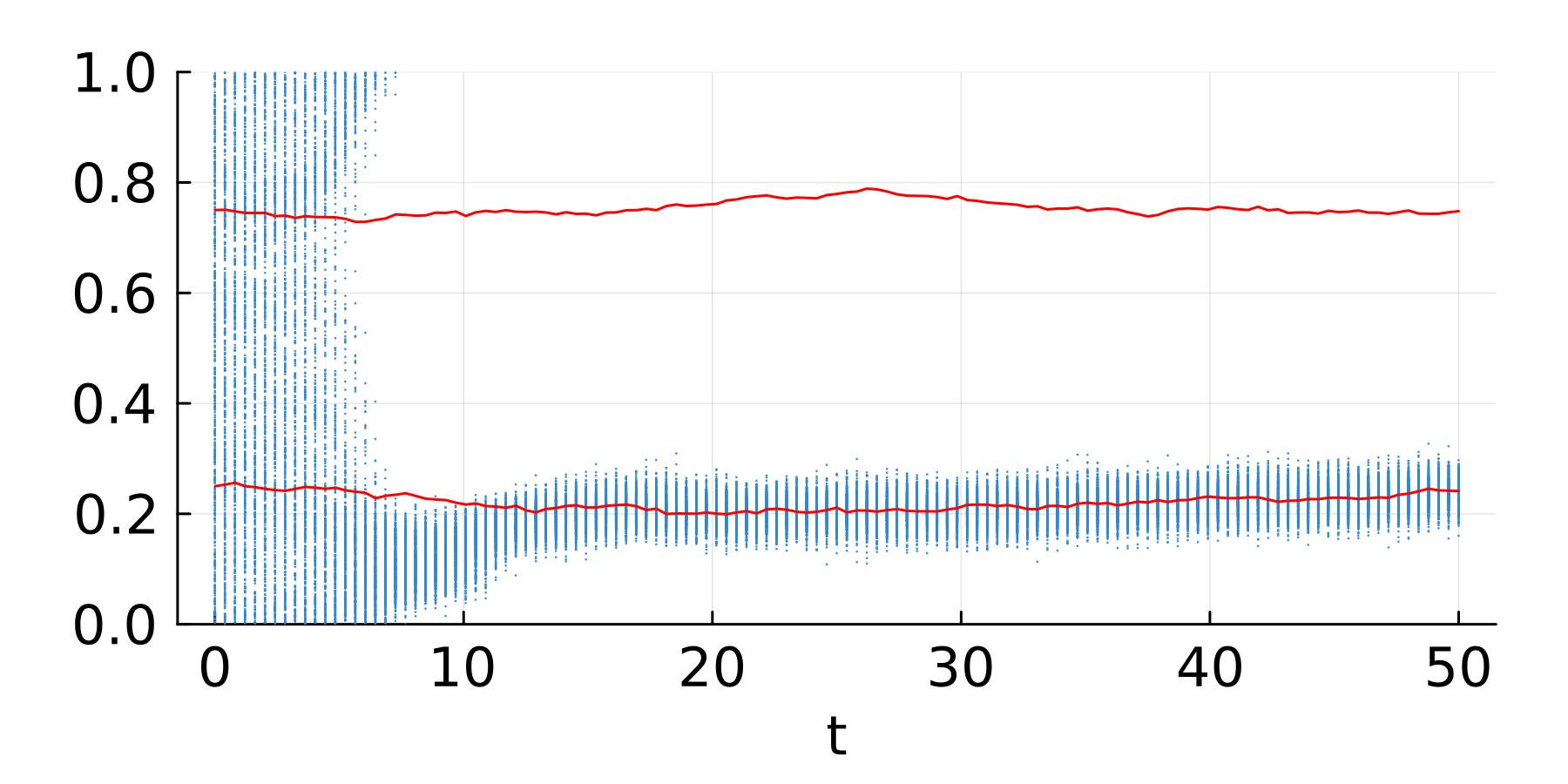}
        \caption{\textit{Voter consensus}. Here $R_{vv} = 0.5$} 
        \label{fig:2party-voter consensus}
    \end{subfigure}
    \caption{Prototypical political scenarios for a two-party system. At time $t=0$, $N_v = 1,000$ voters are distributed uniformly across $[0,1]$ and parties are initially at $p_1(0)=0.25$ and $p_2(0)=0.75$, except for Figure~\ref{fig:2party-disaffected_voters}. The strengths of the voter forces are $\mu_{vv}=\mu_{pv}=1$. The strengths of the party forces are $\mu_{vp}=0.03$ and $\mu_{pp}=0.01$. The interaction radii are $R_{vv}=0.15$, $R_{pv}=0.1$, $R_{vp}=0.3$ and $ R_{pp}=0.1$. The noise strengths are $\sigma_v = 0.04$ and $\sigma_p = 0.005$.}
    \label{fig:2party-example_behaviours}
\end{figure}
\begin{figure}[htbp]
    \centering
    \includegraphics[width=0.5\linewidth]{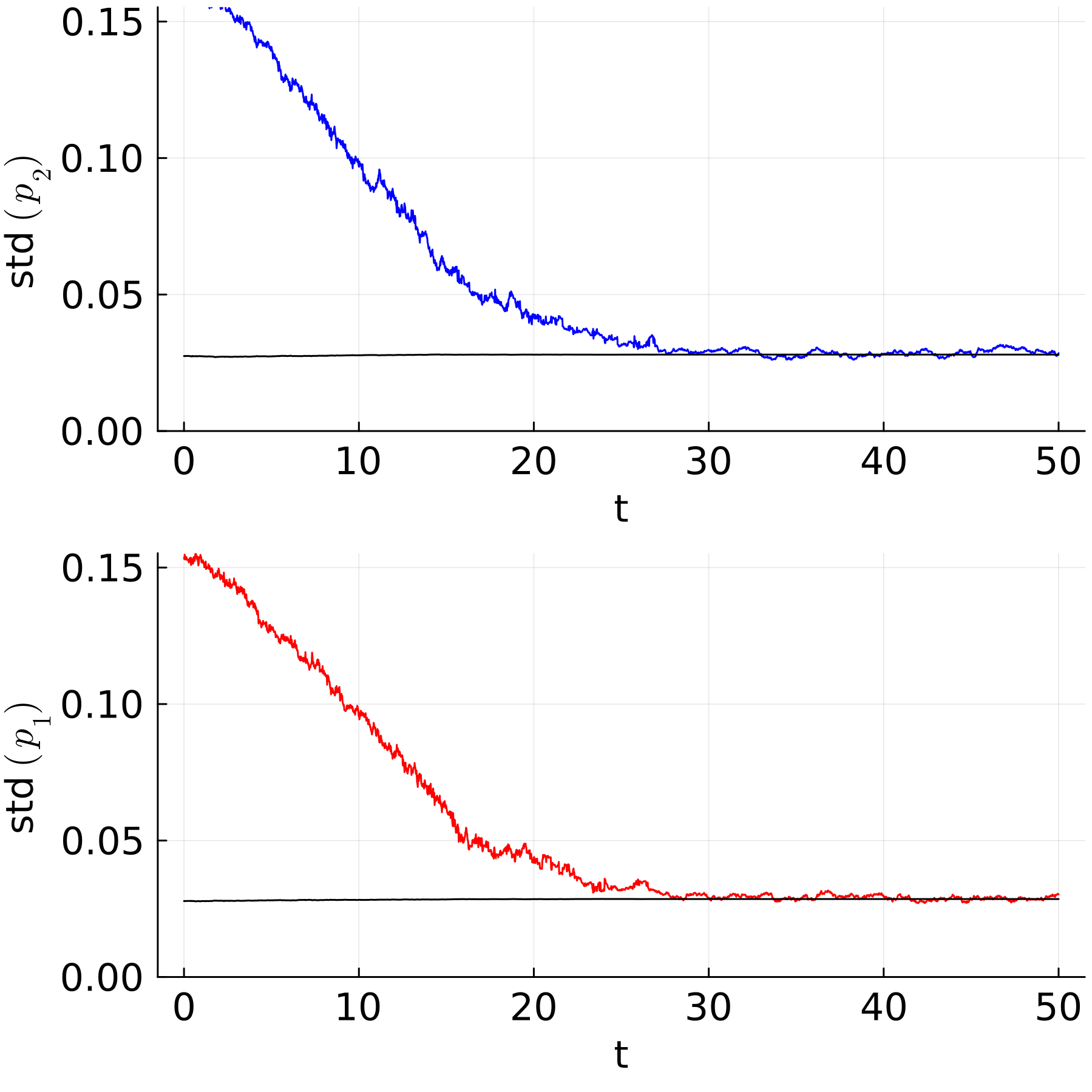}
    \caption{Standard deviations of the voter cluster centred around each party in Figure~\ref{fig:2party-political_base}. The black line denotes the analytical approximation \eqref{eq:stdOU}. Top: party $p_2$ with $p_2(0)=0.75$. Bottom: party $p_1$ with $p_1(0)=0.25$.}
    \label{fig:2party-deltacluster_check}
\end{figure}

\section{Consensus in the modified Hegselmann--Krause model}
\label{sec:consensus}
As we have seen in the previous Section the modified Hegselmann--Krause model (\ref{eq:HKa})--(\ref{eq:HKb}) allows for the formation of opinion clusters. Consensus in our model refers to all voters and parties collapsing to a cluster the size of which is determined by the noise strengths $\sigma_{v}$ and $\sigma_p$. We will see that if $\sigma_{v},\sigma_p$ are sufficiently small and $\mu_{pp} < \mu_{vp}$ the system approaches consensus. The route from a disordered state of uniformly distributed voters and parties to a state of consensus in which voters and parties collapse to a localized region in opinion space is typically via partially ordered states in which voters and parties form several clusters. We show in Figure~\ref{fig:noisyconvergence} an example for a $d=1$-dimensional opinion space and in Figure~\ref{fig:2devolution} an example for a $d=2$-dimensional opinion space. The model recreates some of the core features of the classical noisy Hegselmann--Krause model. In particular, in both cases, the voters and parties form smaller clusters before eventually reaching a consensus state. Figure~\ref{fig:noisyconvergence} shows that initially uniformly distributed voters begin to form clusters centred around parties. These clusters constitute what we encountered in Section~\ref{sec:examples} as the party-base of a party. The top two parties $p_3$ and $p_4$ develop well separated party-base clusters since the distance between them is larger than the interaction-radius $R_{vv}$, prohibiting interactions between voters of different party-bases. The bottom two parties $p_1$ and $p_2$ located near $0.25$ and $0.35$ are closer to one another and a cluster of swing voters forms. By around $t\approx 140$ all voter clusters have merged to form a single cluster with mean $v=0.5$. The outer two parties $p_1$ and $p_4$ slowly move into the cluster as the attraction of the voters overcomes their repulsion of the other parties and they enter the cluster of voters. 

Similarly, in the two-dimensional case depicted in Figure~\ref{fig:2devolution}, we observe clusters of voters forming around parties. Due to the increased size of the opinion space, many more voters are not part of a cluster loyal to a single party. Instead as seen at $t> 120$ there is a cluster of disaffected voters in the bottom right of the opinion space which are not centred around any party. Note that in contrast to the disaffected voters depicted in Figure~\ref{fig:2party-disaffected_voters} here $R_{vv}$ is sufficiently large to allow for a cluster of disaffected voters of size $\delta_{\rm{cl}}$. In the model, over time these disaffected voters randomly evolve into one of the attractive areas of the parties. Parties evolve to minimise their distance from voters within their interaction radius, and so move towards the single largest cluster. This process of consensus formation can be slowed down by decreasing the interaction radius $R_{vp}$. In this case, parties are not affected by far away voters, requiring slow diffusion to form consensus. Depending on the strength and interaction radii of the relative forces, parties may be attracted to voters when $\mu_{vp}>\mu_{pv}$, whereas in the case that $\mu_{pv}>\mu_{vp}$, illustrated in Figure~\ref{fig:2devolution}, voters will cluster around parties before the voter-voter interactions lead to global consensus. Figure~\ref{fig:t1000} shows that at $t=1,000$ most agents are located in a single cluster with only a few voters not yet entrained. We remark that in bounded opinion spaces, Brownian motion is recurrent even in high dimensions, and hence voters, which are not yet entrained by the main cluster at the final time in Figures~\ref{fig:noisyconvergence} and \ref{fig:2devolution}, will eventually join the cluster. The final asymptotic state as well as the dynamical behaviours encountered in the dynamics evolving towards this state crucially depend on the initial conditions of the voters and parties.\\

\begin{figure}[htbp]
    \centering
     \includegraphics[width=\textwidth]{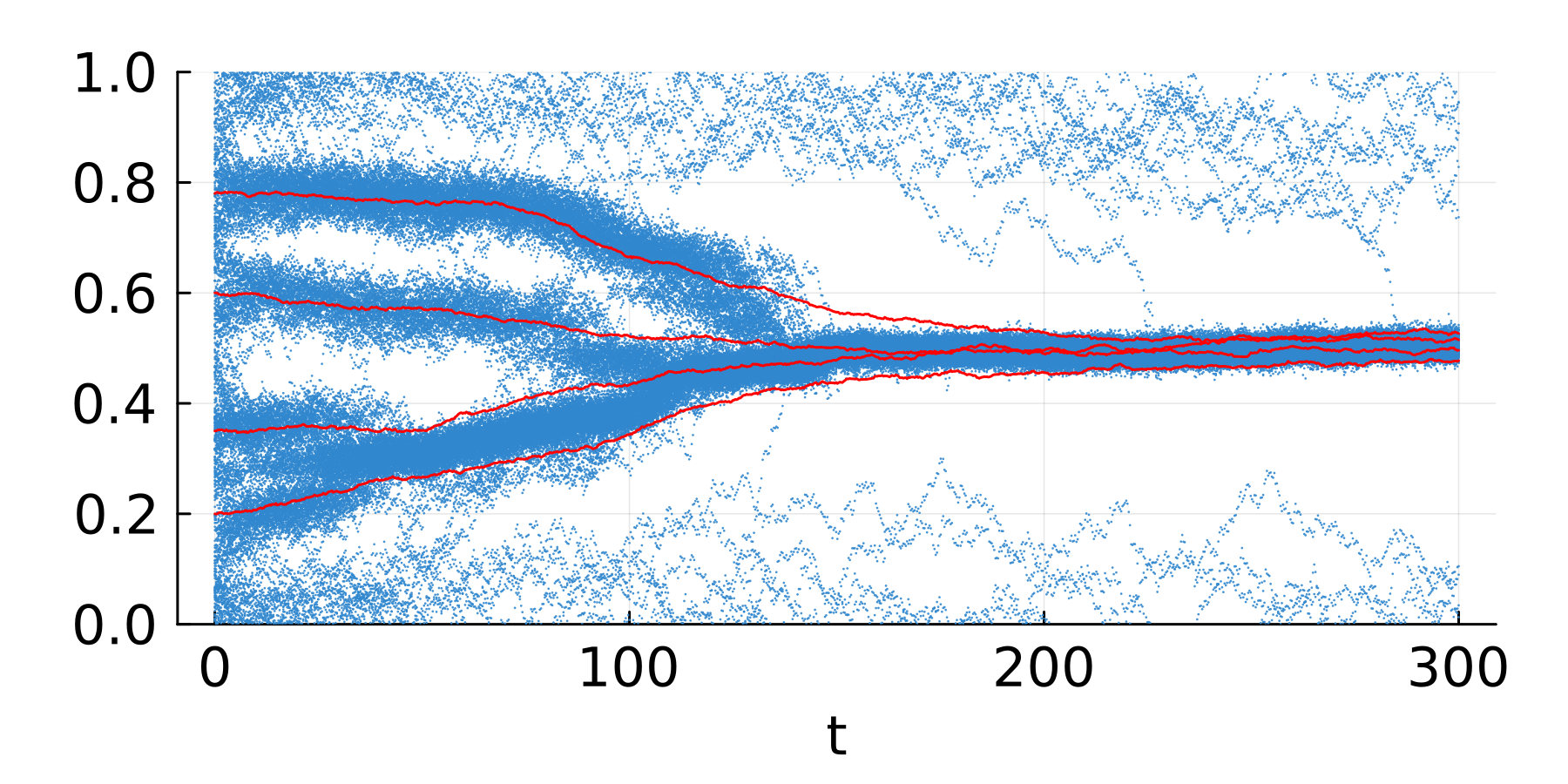}
    \caption{Evolution of the modified Hegselmann--Krause model (\ref{eq:HKa})--(\ref{eq:HKb}) in a $d=1$-dimensional opinion space for $500$ voters (blue) and $4$ parties (red). Initially voters are distributed uniformly across $[0,1]$ and parties are initially at $p_1(0)=0.2$, $p_2(0)=0.35$, $p_3(0)=0.6$ and $p_4(0)=0.78$.} Parameters are $R_{vv}=0.05$, $R_{pv}=0.1$, $R_{vp}=0.4$, $R_{pp}=0.05$ and $\mu_{vv}=0.5$, $\mu_{pv}=0.8$, $\mu_{vp}=0.02$, $\mu_{pp}=0.02$. The noise strengths are $\sigma_v = 0.02$ and $\sigma_{p} = 0.002$.
    \label{fig:noisyconvergence}
\end{figure}
\begin{figure}[htpb]
    \centering
    \begin{subfigure}[htbp]{0.42\textwidth}
        \centering
        \includegraphics[width=\textwidth]{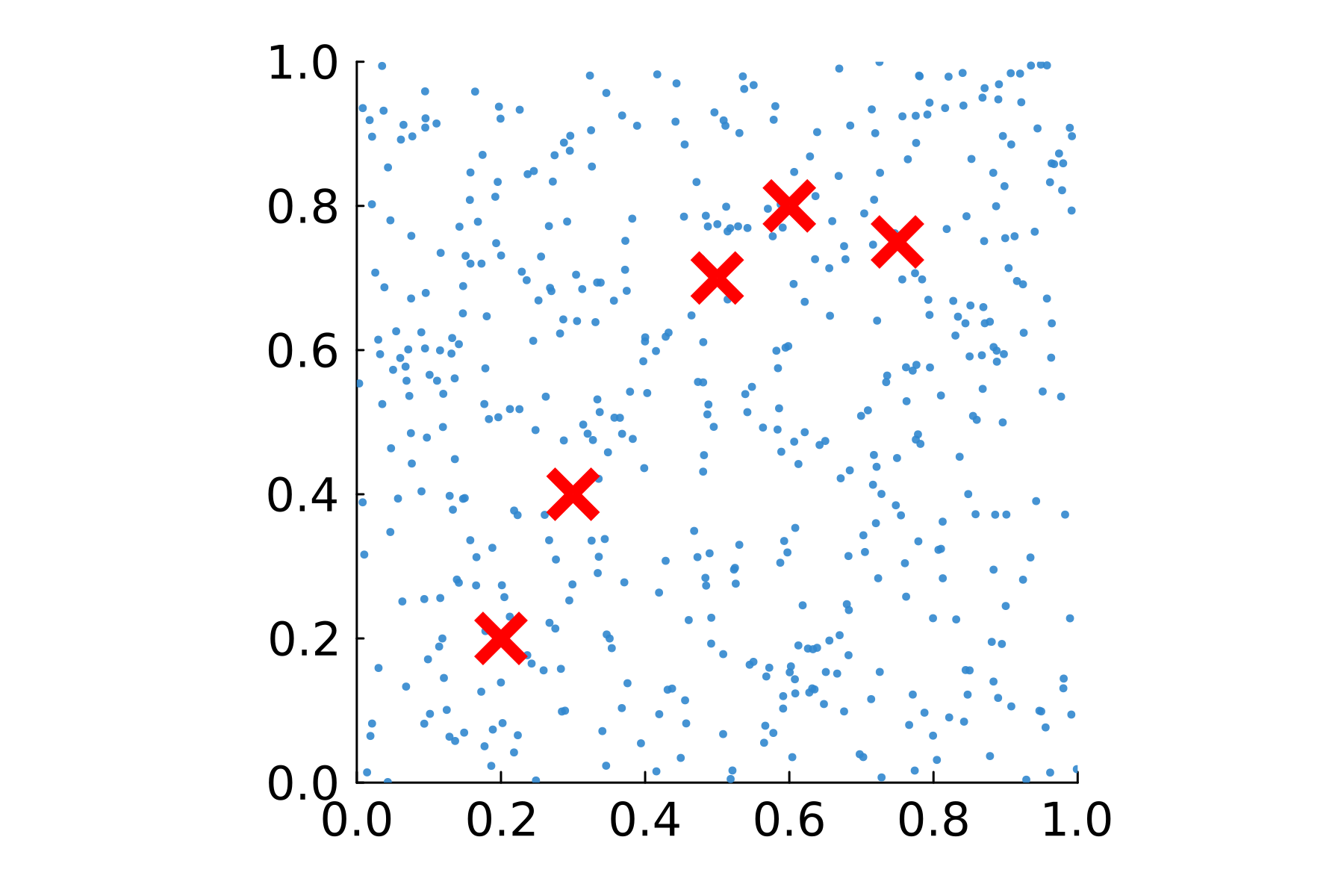}
        \caption{t=0} 
    \end{subfigure}
    \hfill
    \begin{subfigure}[htbp]{0.42\textwidth}  
        \centering 
        \includegraphics[width=\textwidth]{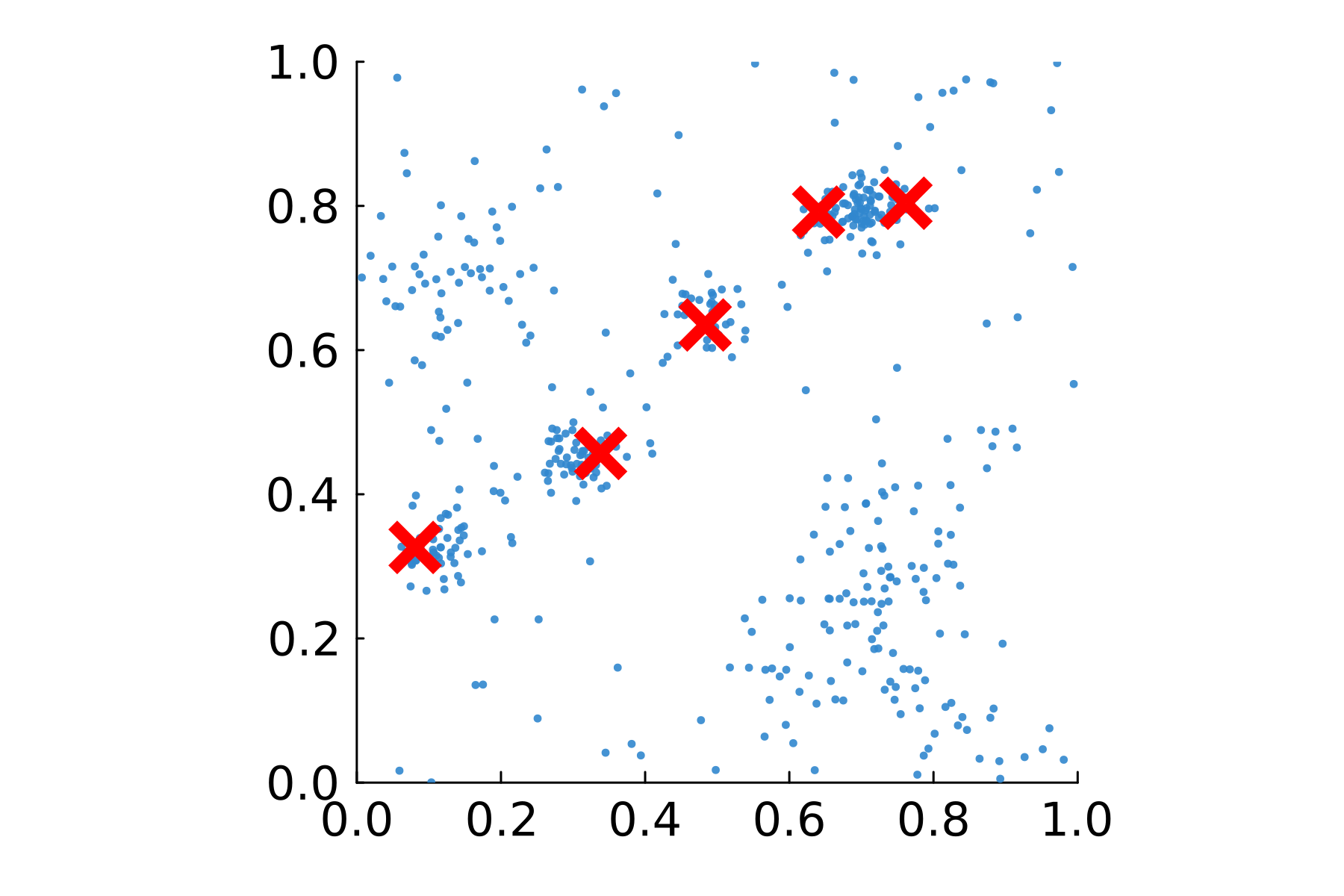}
        \caption{t=80} 
    \end{subfigure}
    \vskip\baselineskip
    \begin{subfigure}[htbp]{0.42\textwidth}   
        \centering 
        \includegraphics[width=\textwidth]{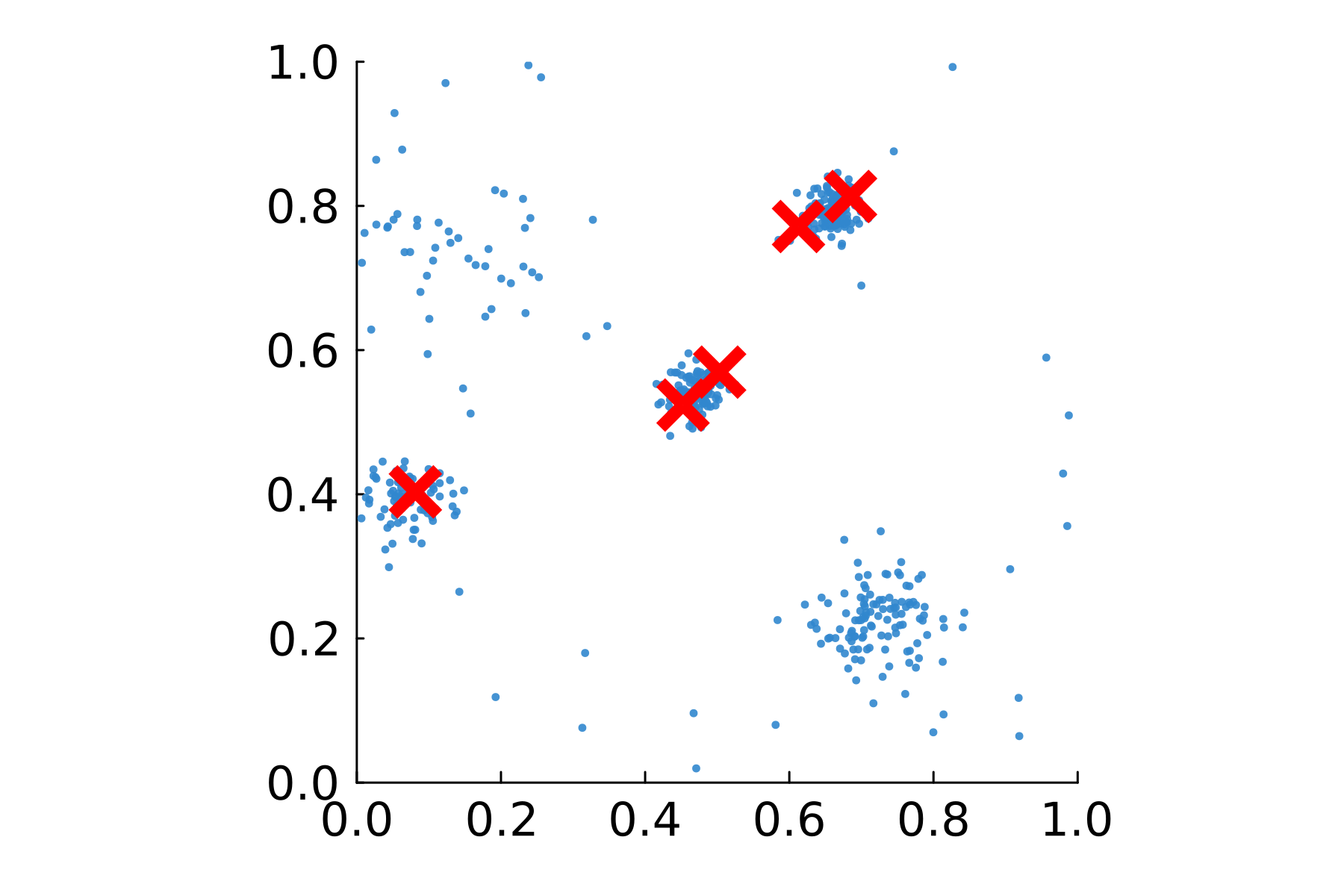}
        \caption{t=120} 
    \end{subfigure}
    \hfill
    \begin{subfigure}[htbp]{0.42\textwidth}   
        \centering 
        \includegraphics[width=\textwidth]{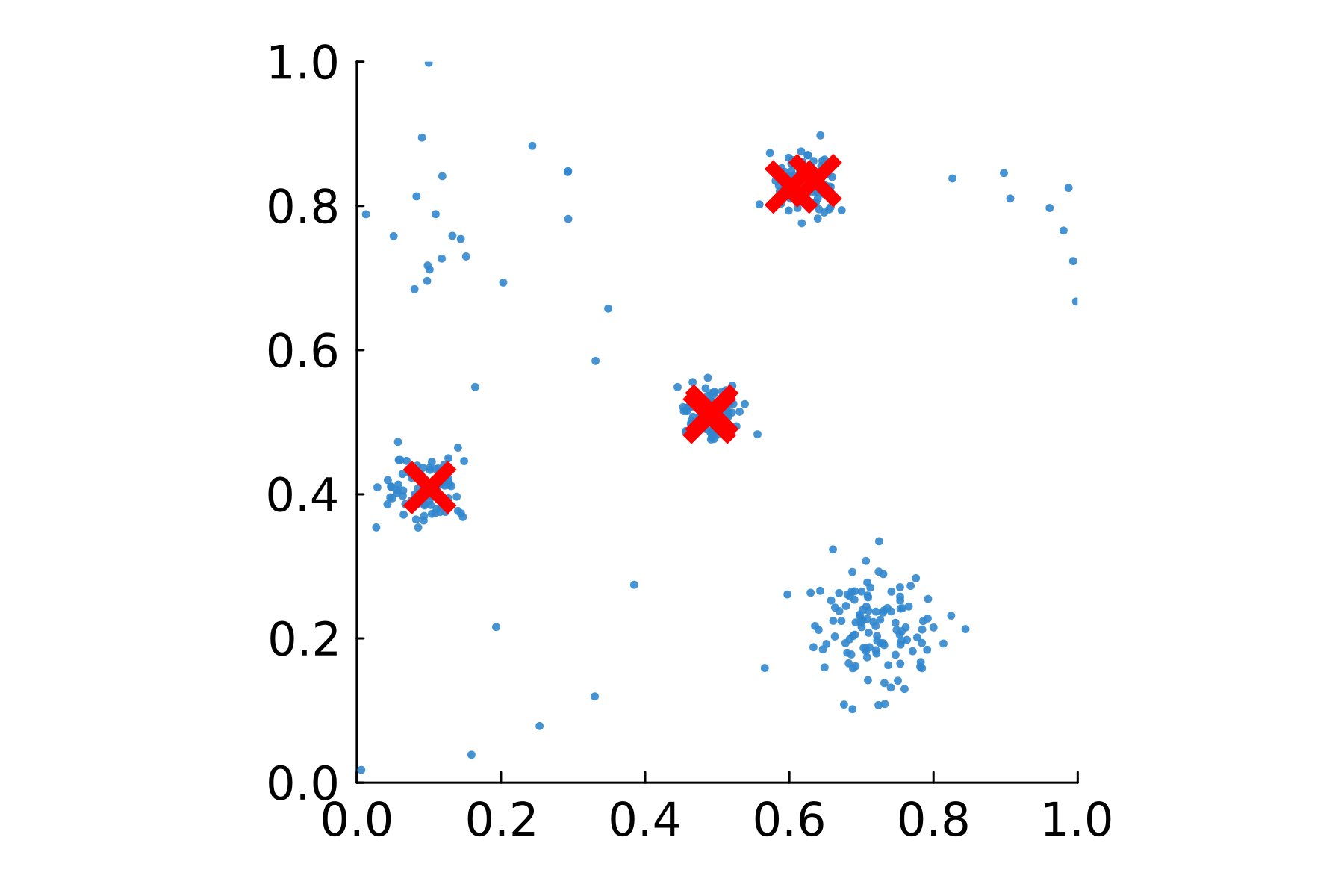}
        \caption{t=160} 
    \end{subfigure}
    \begin{subfigure}[htbp]{0.42\textwidth}   
        \centering 
        \includegraphics[width=\textwidth]{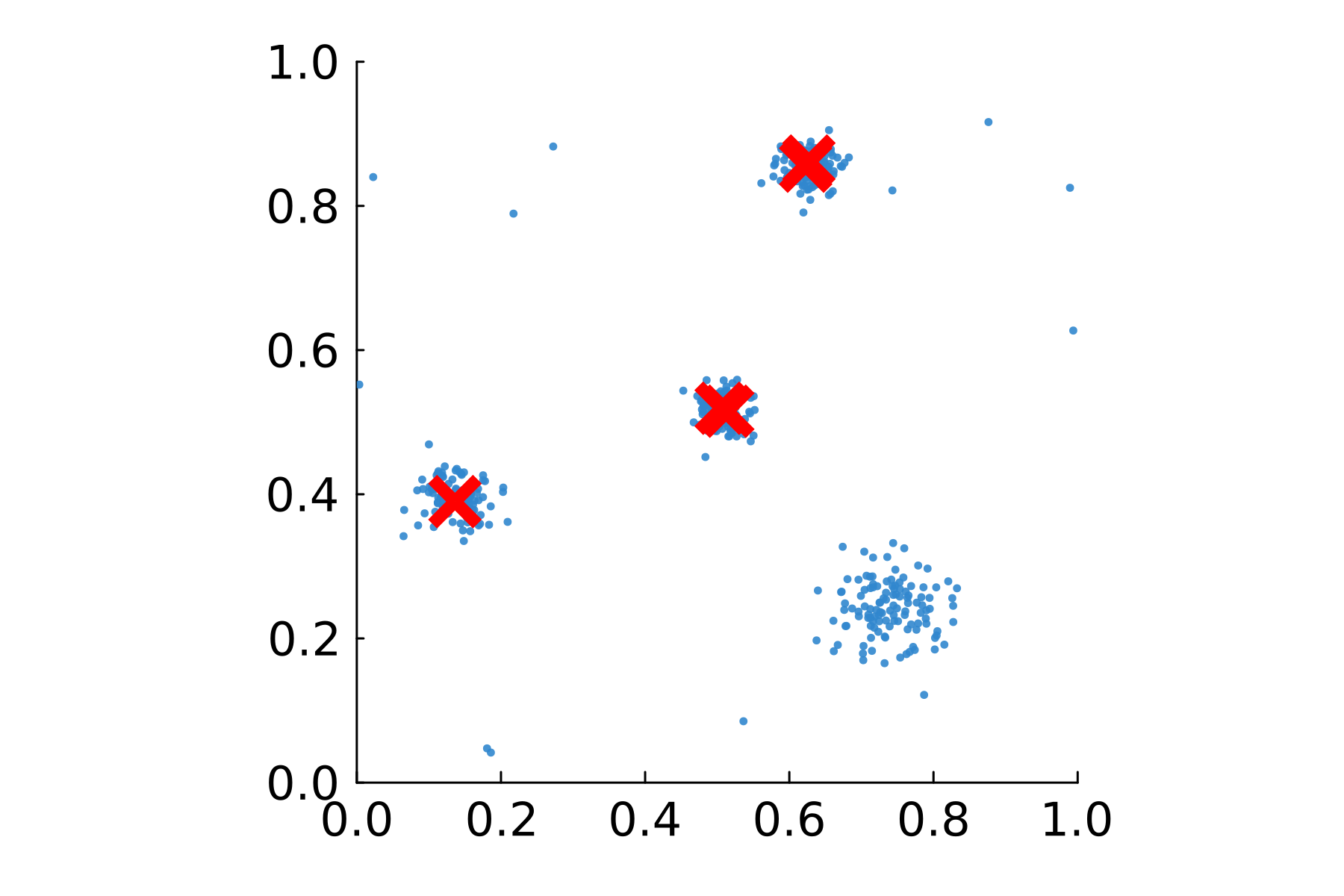}
        \caption{t=300} 
    \end{subfigure}
    \hfill
    \begin{subfigure}[htbp]{0.42\textwidth}   
        \centering 
        \includegraphics[width=\textwidth]{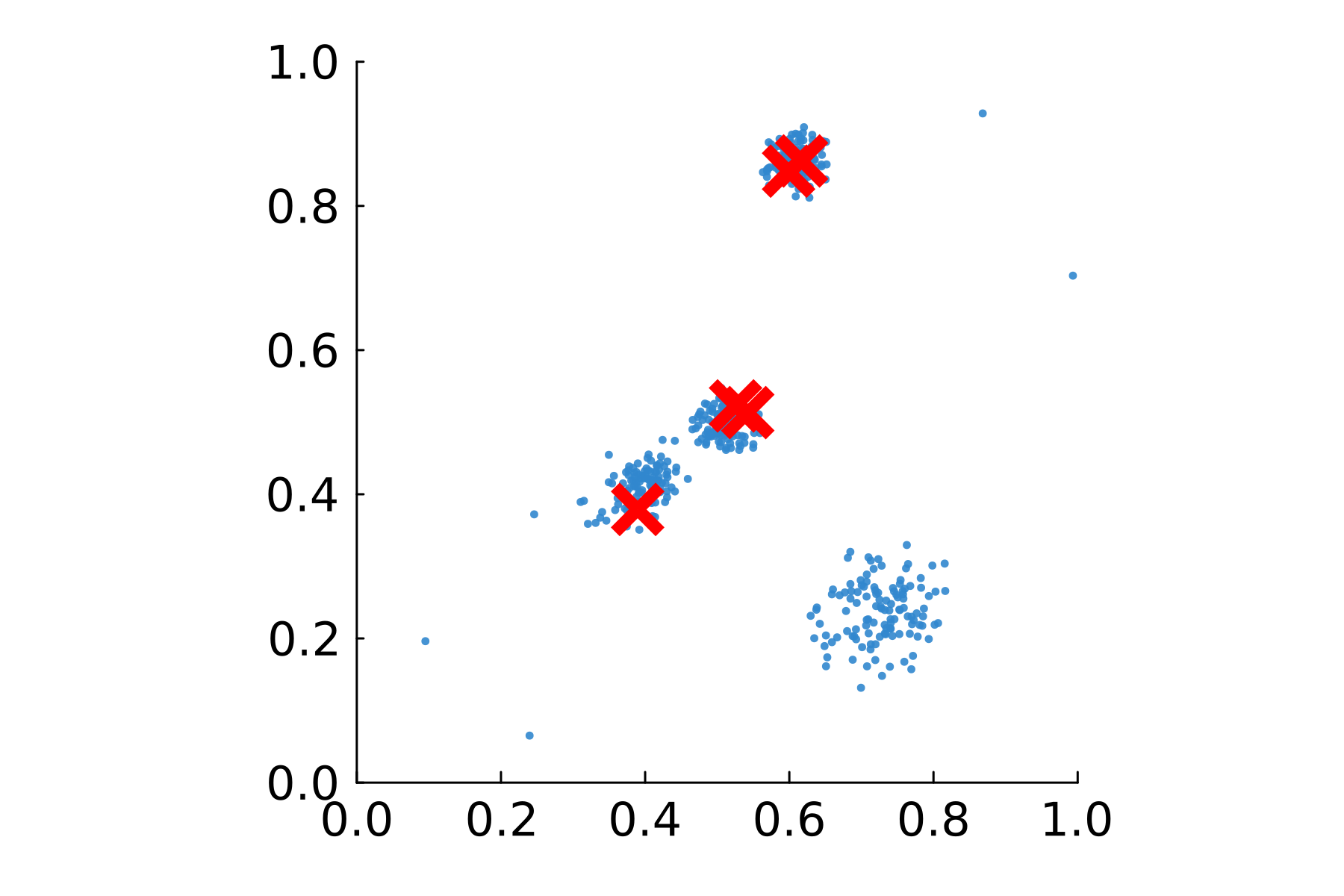}
        \caption{t=410} 
    \end{subfigure}
    \begin{subfigure}[htbp]{0.42\textwidth}   
    \centering 
    \includegraphics[width=\textwidth]{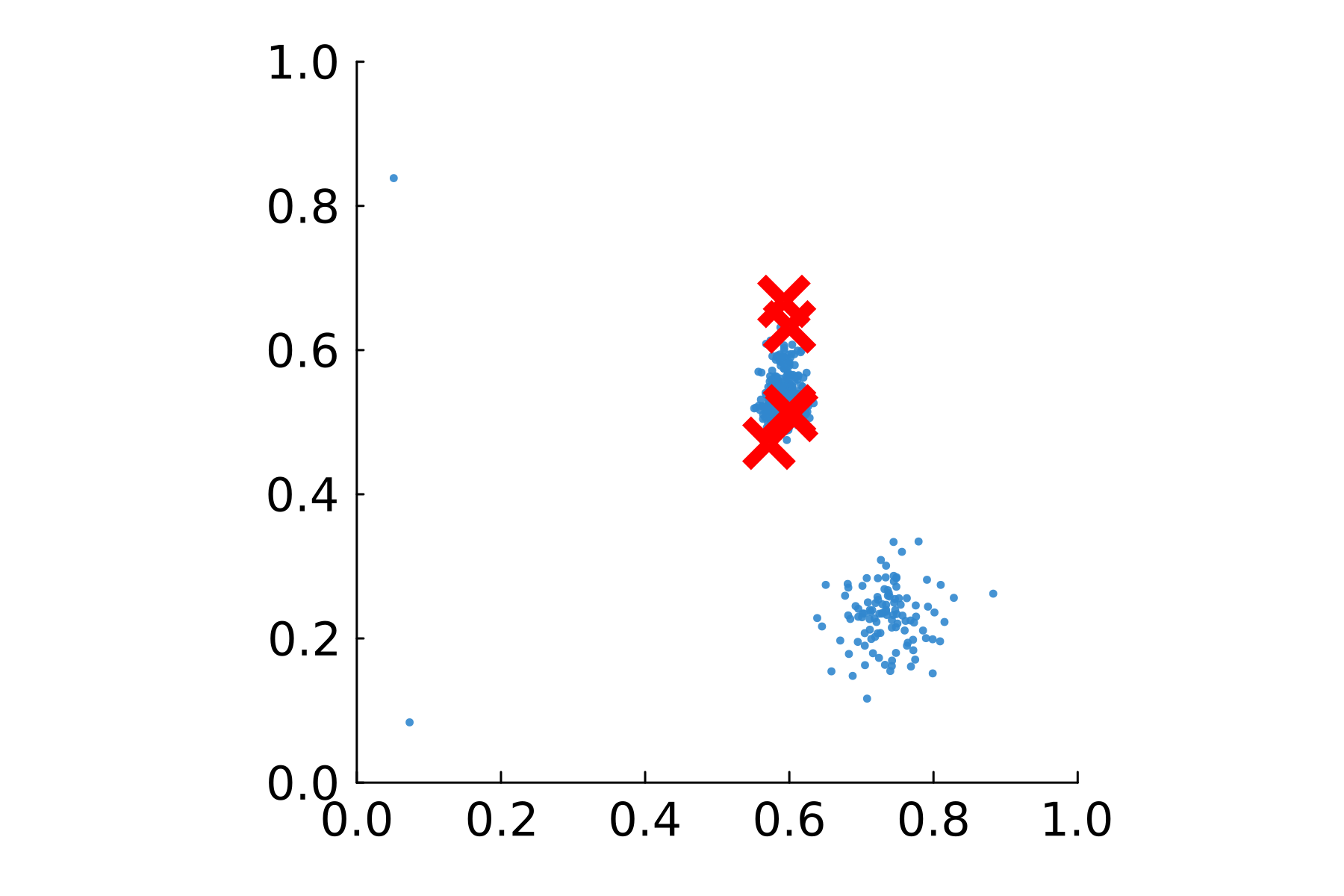}
    \caption{t=450} 
\end{subfigure}
\hfill
\begin{subfigure}[htbp]{0.42\textwidth}   
    \centering 
    \includegraphics[width=\textwidth]{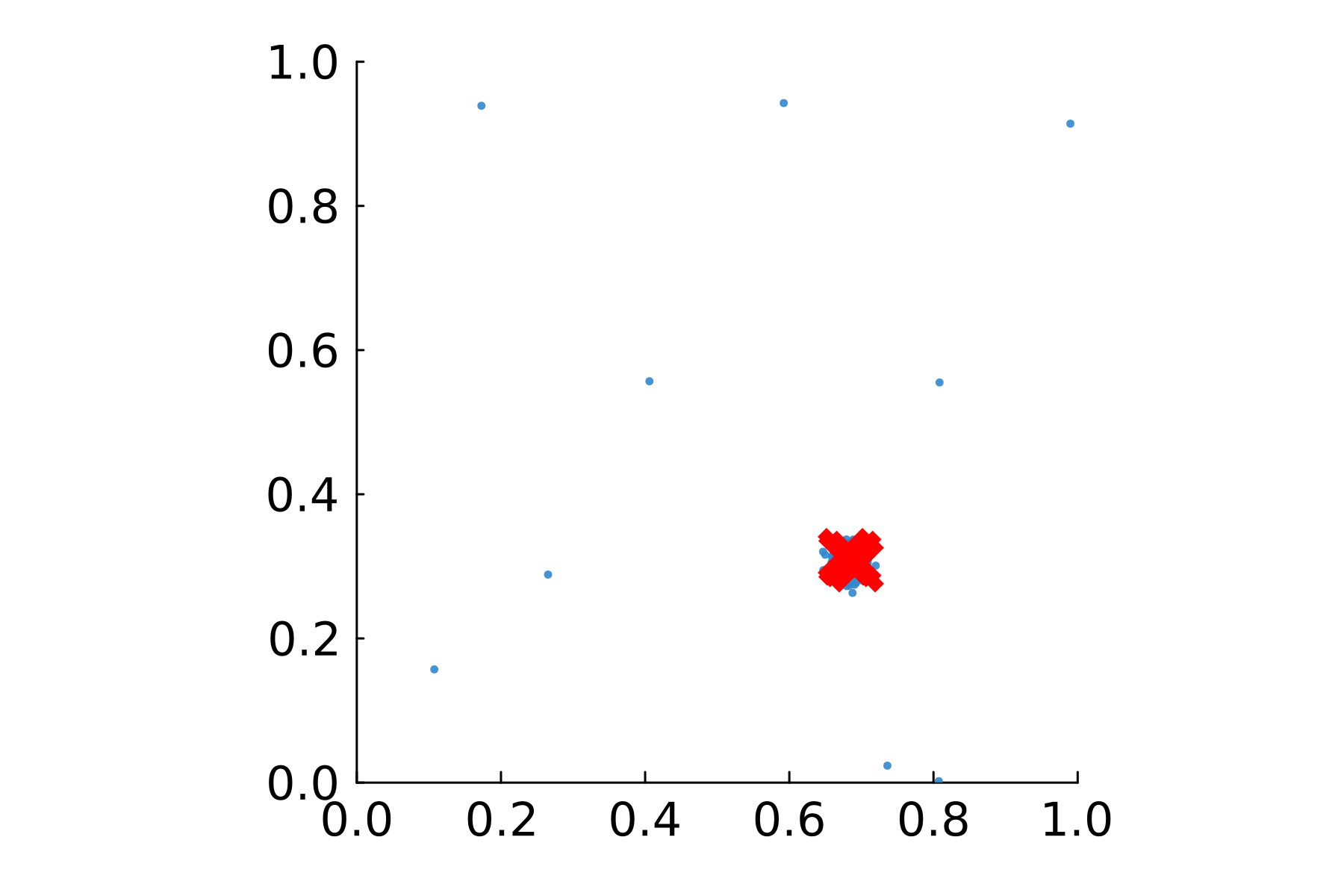}
    \caption{1,000} 
    \label{fig:t1000}
\end{subfigure}
    \caption{Evolution of the modified Hegselmann--Krause model (\ref{eq:HKa})--(\ref{eq:HKb}) in a $d=2$-dimensional opinion space for $500$ voters (blue dots)  and $5$ parties (red crosses). At time $t=0$ voters are distributed uniformly across $[0,1]^2$ and parties are initially at $p_1(0)=(0.2,0.2)$, $p_2(0)=(0.3,0.4)$, $p_3(0)=(0.5,0.7)$, $p_4(0)=(0.6,0.8)$, and $p_5(0)=(0.75,0.75)$.   Parameters are $R_{vv}=0.2$, $R_{pv}=0.1$, $R_{vp}=0.5$, $R_{pp}=0.05$ and $\mu_{vv}=0.5$, $\mu_{pv}=1.0$, $\mu_{vp}=0.05$, $\mu_{pp}=0.02$. The noise strengths are $\sigma_v = 0.02$ and $\sigma_{p} = 0.002$.}
    \label{fig:2devolution}
\end{figure}

To quantify the degree of consensus of voters in the original Hegselmann--Krause model \eqref{eq:HK0} \citet{wang2017noisy} introduced the order parameter 
\begin{align}
Q_{vv} = \frac{1}{N_{v}^2}\sum_{i=1}^{N_v}\sum_{j=1}^{N_v}\phi\left(\frac{||v_j-v_i||}{R_{vv}}\right).
\end{align}
This voter-centric parameter is, however, not well suited to quantify consensus in the presence of parties. One may envisage a situation of no consensus in which voters form a tight cluster in opinion space due to the voter-voter force proportional to $\mu_{vv}$ but the parties evolve unaffected by the voters if their distance is larger than the respective interaction radii. This, admittedly, unrealistic situation would be classified as consensus with $Q_{vv}=1$. A more common situation is the ordered state when voters are strongly attracted to a party by the party-voter force with strength $\mu_{pv}$, forming a cluster around the party. This cluster forms even if the voter-voter force $F_{vv}$ is weak due to a small interaction radius $R_{vv}$. The smallness of the interaction radius $R_{vv}$ implies $Q_{vv}\ll 1$ misclassifying the ordered state as disordered. Depending on the parameter space and initial conditions, it is non-trivial to determine a single order parameter that may be used to quantify the order. Indeed, in the scenarios outlined in Section~\ref{sec:examples}, $Q_{vv}$ can decrease in the party-base or the swing voter scenarios, compared to a uniform distribution, even though there is clear clustering occurring.

To account for the presence of parties, we introduce an alternative consensus diagnostic which is given by the weighted average distance between voter-voter pairs, voter-party pairs and party-party pairs, which measures a deviation from the disordered state in which voters and parties are uniformly distributed over the opinion space. In particular, we define
\begin{equation}
    \hat{D} = 1 - \frac{1}{3 k_d}\left(\frac{1}{N_{v}^2}\sum_{v_i,v_j}||v_i - v_j|| + \frac{1}{N_{v}N_p}\sum_{v_i,p_\alpha}||v_i - p_\alpha|| + \frac{1}{N_{p}^2}\sum_{p_\alpha,p_\beta}||p_\alpha - p_\beta||\right),
\end{equation}
where
\begin{align}
k_d = 2^d \int_{\textbf{0}}^{\textbf{1/2}} \sqrt{x_1^2 + x_2^2 + ... + x_d^2}\; d\textbf{x}
\end{align}
is the expected distance between two random draws from a uniform distribution on the hypercube $[0,1]^d$ with periodic boundary conditions. The parameter $\hat D$ is not a strict order parameter in the sense of statistical mechanics but instead a diagnostics measuring how far the distribution of voters and parties differs from the maximally disordered state of uniformly distributed voters. The proposed consensus diagnostic satisfies $0 \le \hat{D} \le 1$ with equality given only by uniformity and unanimous consensus, respectively. 

The two diagnostics $Q_{vv}$ and $\hat D$ capture different aspects of consensus formation. Whereas $Q_{vv}$ counts the number of interacting voter pairs (within $R_{vv})$, $\hat{D}$ measures the distance of the voters and parties from a uniform distribution. This makes $Q_{vv}$ more sensitive to changes in clusters (mergers and formations), while $\hat{D}$ is less sensitive to clusters but more sensitive to departures from a uniform state.\\

Figure~\ref{fig:noisyconvergenceorders} (left) shows the evolution of $Q_{vv}$ and $\hat{D}$ for the case of the $1$-dimensional opinion space depicted in Figure~\ref{fig:noisyconvergence}. Near consensus is achieved at $t=300$ with $Q_{vv}>\hat{D}\approx 0.9$. The consensus diagnostic $\hat{D}$ exhibits a monotonic increase. It initially grows slowly during the observed clustering of voters into party-bases and then increases more strongly when clusters coalesce. The growth in $\hat{D}$ is fastest at $t\approx 100$ when two voter clusters merge into a single cluster. By contrast, $Q_{vv}$ better captures the periods of coexisting clusters which are characterized by plateaus in $Q_{vv}$. The period when there are $4$ distinct party-base clusters for $10 \lessapprox t \lessapprox 30$ and the period  $50 \lessapprox t\lessapprox 100$ when there are two party-base clusters and one cluster of swing clusters are well captured by plateaus in $Q_{vv}$. The existence of nonentrained voters at $t=300$ implies that both $\hat D$ and $Q_{vv}$ are not equal to $1$. The nonentrained voters will eventually diffuse into the main cluster.

Figure~\ref{fig:noisyconvergenceorders} (right) shows the evolution of $Q_{vv}$ and $\hat{D}$ corresponding to the dynamics in the $2$-dimensional opinion space shown in Figure~\ref{fig:2devolution}. The consensus diagnostic $\hat D$ exhibits a plateau for $120 \lessapprox t \lessapprox 380$ corresponding the state of four weakly interacting cluster (three party-base clusters and one disaffected voter cluster). When these clusters begin to interact more strongly $\hat D$ increases monotonically. The order parameter $Q_{vv}$ captures the ordered state of four weakly interacting clusters less clearly. However, it better captures the merger of two groups (three parties) at $t\approx 420$ and a further merger of all party-base clusters around $t\approx 450$. The slow increase of $\hat D$ for $t\gtrapprox 450$ is due to the slowed down merger of parties caused by their mutual repulsion. All voters and parties coalesce in a single cluster around $t\approx 1,000$.

\begin{figure}[htbp]
    \centering
     \includegraphics[width=0.49\linewidth]{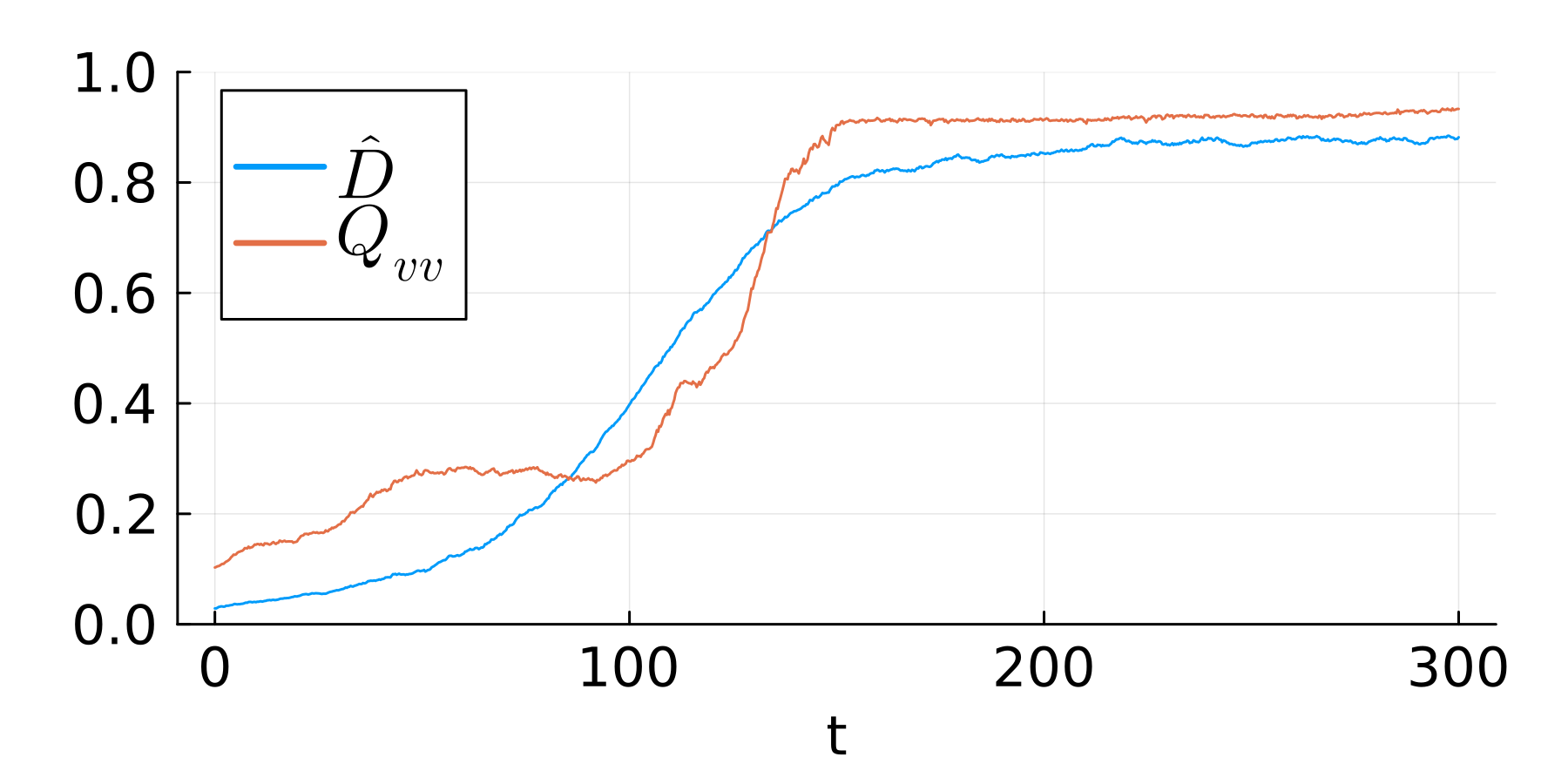}
     \includegraphics[width=0.49\linewidth]{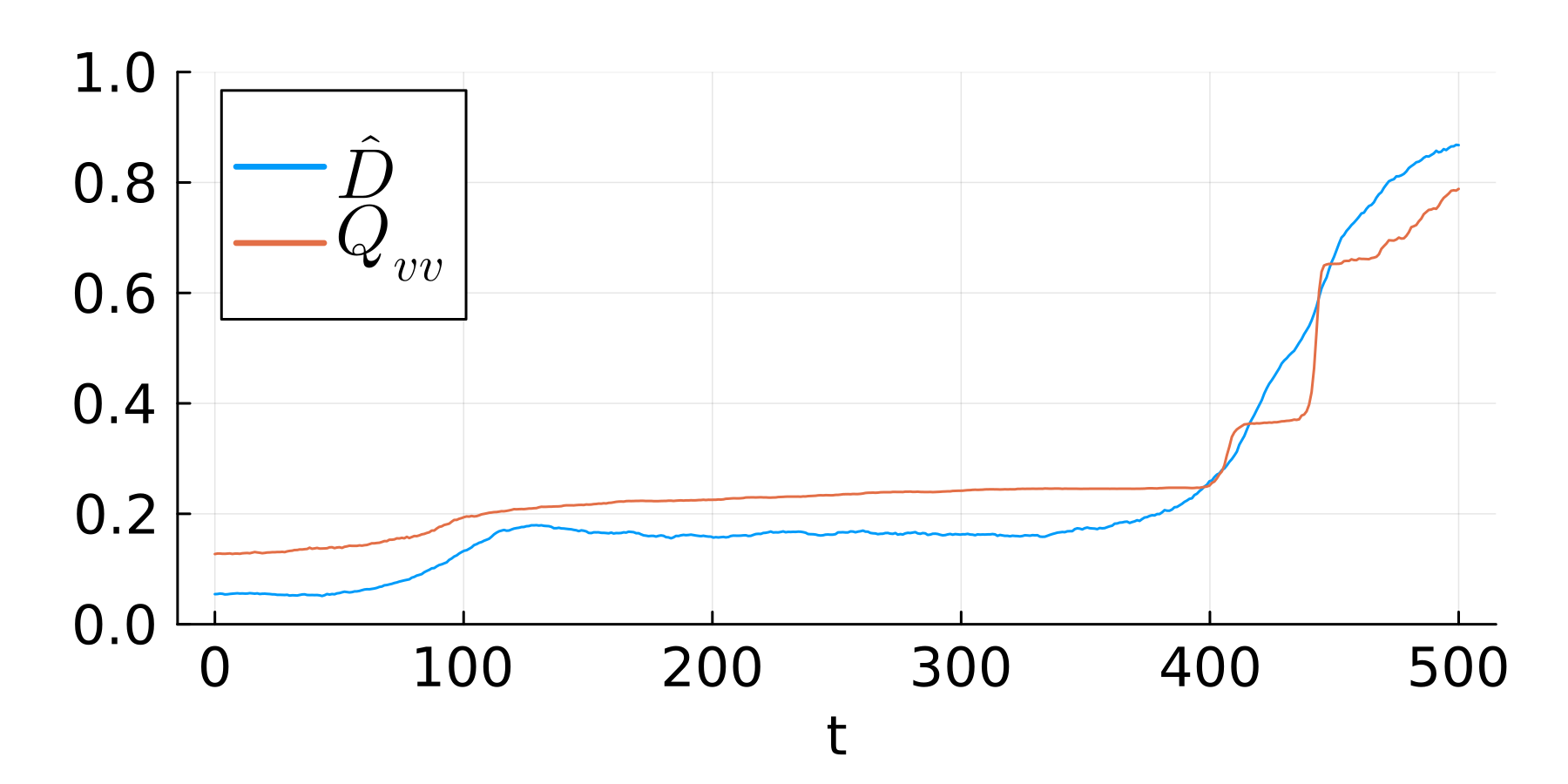}
     \caption{Order parameter $Q_{vv}$ and consensus diagnostic $\hat{D}$ corresponding to the simulations in the $1$-dimensional opinion space shown in Figure~\ref{fig:noisyconvergence} (left) and in the $2$-dimensional opinion space shown in Figure~\ref{fig:2devolution} (right).}
    \label{fig:noisyconvergenceorders}
\end{figure}
A remarkable feature of the Hegselmann--Krause model (\ref{eq:HK0}) is the existence of a phase transition: there exists a critical noise strength $\sigma_c$ such that for noise strength $\sigma>\sigma_c$ voters are not able to evolve towards consensus but instead diffuse as independent noisy agents \citep{Fortunato05,garnier2017consensus, wang2017noisy, delgadino2023phase, lucarini2020response,GoddardEtAl21}. It is intuitive that the modified Hegselmann--Krause model \eqref{eq:HKa}-\eqref{eq:HKb} exhibits a similar phase transition. In Figure~\ref{fig:noisyconvergence} we have seen an example when the system organizes in clusters before eventually reaching consensus with $Q_{vv}>0.95$ around $t\approx 150$ when the voters collapse to a single cluster. The remaining voters which have not yet been entrained by the cluster at $t=200$ will eventually diffuse into the consensus cluster. If we keep all parameters the same but increase the noise strength from $\sigma_v=0.002$ to $\sigma_v=0.1$ final consensus cannot be reached. As shown in Figure~\ref{fig:statecomparison} the noise term dominates over the attractive interaction forces, and the system effectively behaves like a set of $N_v$ independent Brownian motions. In this case, voters do not form clusters and both $Q_{vv} \approx 2R_{vv}=0.1$ and $\hat{D}\approx 0.028$ correspond to values when the voters are drawn from a uniform distribution. We note that with periodic boundary conditions, a set of $N_v$ independent Brownian processes are uniformly distributed.

\begin{figure}[htbp]
    \centering
     \includegraphics[width=0.49 \textwidth]{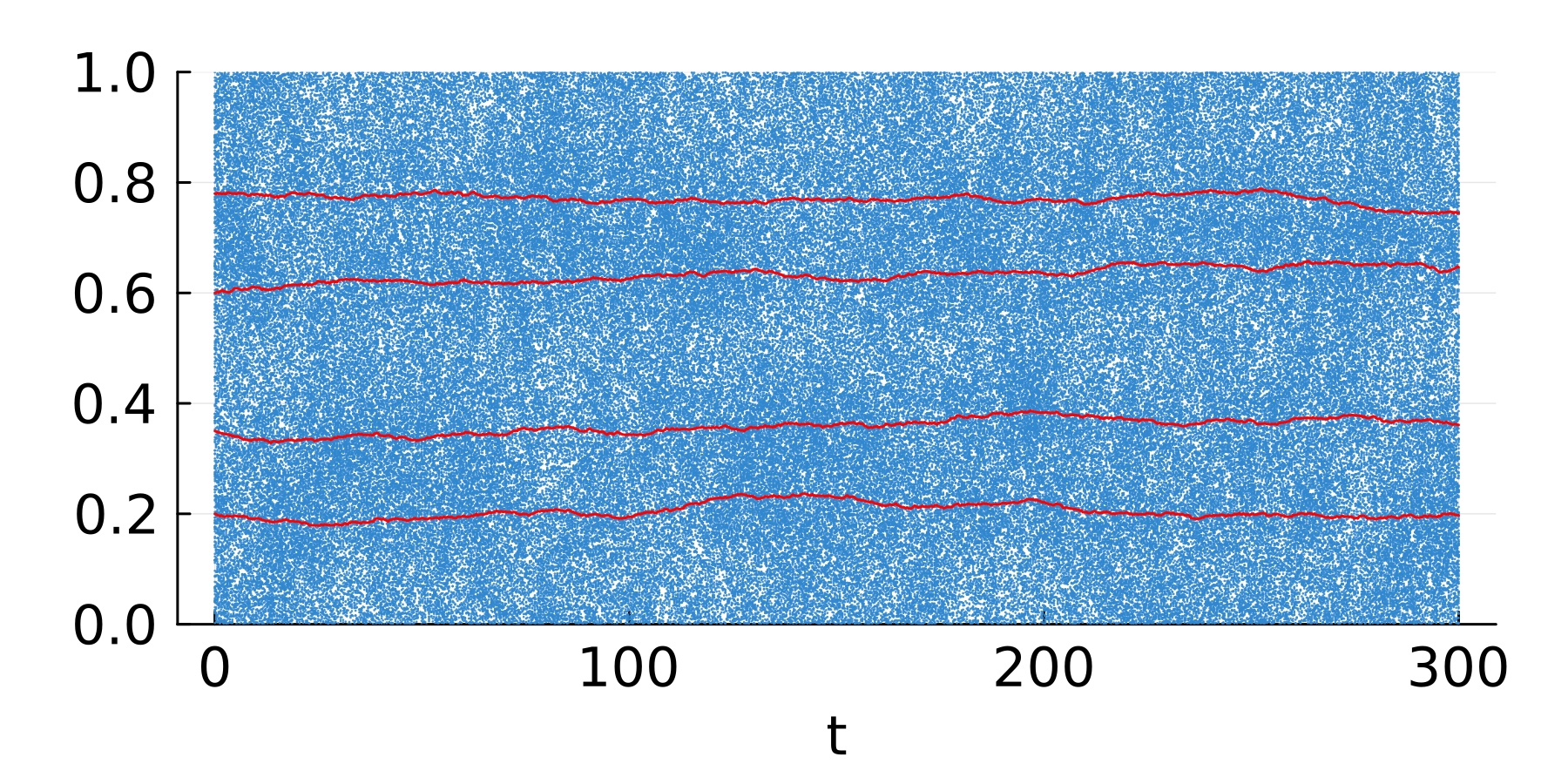}
     \includegraphics[width=0.49 \textwidth]{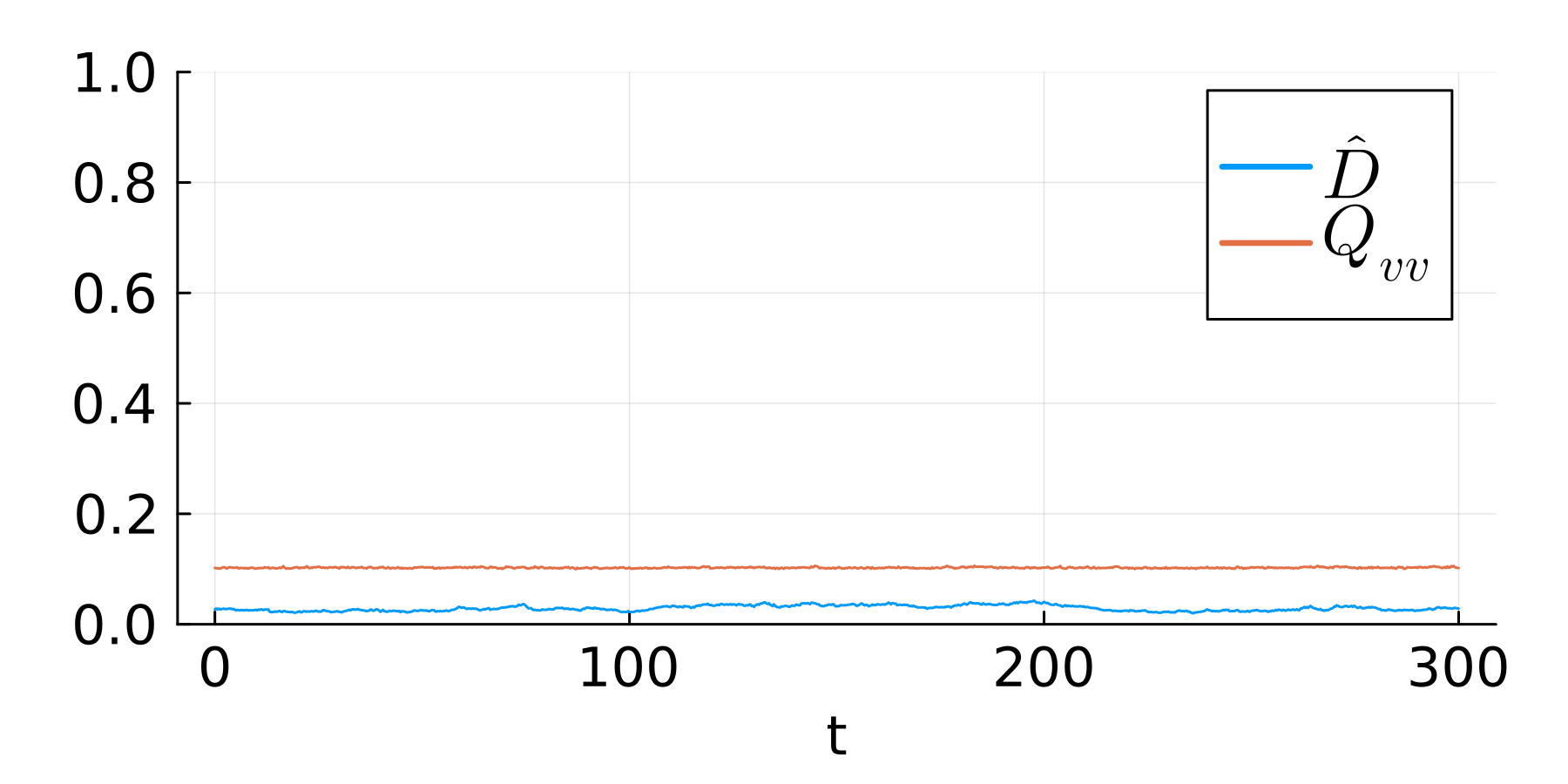}
    \caption{Evolution of the modified Hegselmann--Krause model (\ref{eq:HKa})--(\ref{eq:HKb}) in a $d=1$-dimensional opinion space for $500$ voters and $4$ parties for $\sigma_v=0.1$ leading to a stable uniform distribution of voters. The initial conditions and the rest of the parameters is as in Figure~\ref{fig:noisyconvergence}. Left: Actual voter (blue dots) and party (red crosses) dynamics. Right: Corresponding evolution of the order parameter $Q_{vv}$ and the consensus diagnostic $\hat D$.}
    \label{fig:statecomparison}
\end{figure}

In the following two Sections, we investigate the phase transition of the modified Hegselmann--Krause model \eqref{eq:HKa}--\eqref{eq:HKb}. In Section~\ref{sec:app1} we consider the noiseless case and provide a sufficient condition for the occurrence of consensus. In Section~\ref{sec:meanfield} we determine the conditions for the phase transition in the noisy case, employing the mean-field limit of the model.
  

\section{Criterion for consensus in the noiseless modified Heg\-sel\-mann--Krause model}
\label{sec:app1}
We consider here the deterministic case with $\sigma_v = \sigma_p = 0$. In the deterministic case, unanimous consensus is defined as the situation when all voters and all parties occupy the same position in opinion space. For notational convenience, we introduce the state of the system at time $t$ as $\varphi (t) = (v(t), p(t))$ with the voter opinion profile $v(t) := (v_1(t), \hdots, v_{N_v}(t)) \in \mathbb{R}^{dN_v}$ and the party opinion profile $p(t)=(p_1(t), \hdots,p_{N_p}(t))  \in \mathbb{R}^{dN_p}$. If there is no repulsive force between the parties with $\mu_{pp} =0$, global unanimous consensus occurs if all interaction radii are sufficiently large to cover the convex hull $\Omega(\varphi(0))$ of the initial distribution of voters and parties. In this case, all agents, voters and parties, are mutually and attractively interacting.  In the case that the smallest support of the forces does not cover the convex hull, consensus is not guaranteed and typically non-interacting clusters form. The size and number of these clusters will depend on the interaction radii and force strengths.

If the repulsive party-party interaction force is included with $\mu_{pp}>0$, more complex interactions are possible. Unanimous consensus can still occur provided the strength of the forces exerted by voters on parties dominates over the repulsion of the party-party interaction. This is formulated in the following Proposition which provides a sufficient condition for consensus.
\begin{prop}
\label{prop:E}
     The state $\varphi(t)$ approaches unanimous consensus if 
     \begin{align}
     \mu_{pp} < \mu_{vp}
     \label{eq:Prop1}
     \end{align}
     and if 
     \begin{align}
        R>\frac{1}{2}\rm{diam}\left(E(0)\right),
         \label{eq:Prop2}
     \end{align}
     where $R={\rm{min}}(R_{vv},R_{vp},R_{pv},R_{pp})$ is the smallest of all interaction radii and $\rm{diam}(E(0))$ is the diameter of the set
    \begin{equation}
        E(t):=\Omega\left(\varphi(t)\cup \{\psi(t)\}\right)
\label{eq:consensusrepulsion}
\end{equation}
at initial time $t=0$. The set $E(t)$ is the convex hull covering the voters and parties with 
\begin{equation}
 \psi(t) = \frac{\mu_{vp}\langle v(t) \rangle - \mu_{pp}\langle p(t) \rangle}{\mu_{vp} - \mu_{pp}}.
\end{equation} 
Here the angular brackets denote averages over voters and parties. Additionally, for $t_2 > t_1 \ge 0$, we have
\begin{equation}
E(t_2) \subset E(t_1).    
\end{equation}
\end{prop}
\begin{proof}
    To see that \eqref{eq:Prop1} and \eqref{eq:Prop2} are sufficient conditions to guarantee convergence, we present an adapted argument from \cite{motsch2014heterophilious} for the classical Hegselmann--Krause model. Note that all interaction forces are positive, in particular $\mu_{pp}>0$. By construction, all attractive interactions are initially nonzero with
\begin{align}
\phi\left(\frac{||v_i - v_j||}{R_{vv}}\right) = \phi\left(\frac{||v_j-p_\alpha||}{R_{vp}}\right) = \phi\left(\frac{||p_\alpha - v_i||}{R_{pv}}\right) = 1
\end{align}
for all $i,j \le N_v$ and for all $\alpha \le N_p$. Note that there is no condition on the interaction radius for the repulsive party-party interaction kernel. Hence, at $t=0$ \eqref{eq:HKa} may be written as,
\begin{align}
    \frac{dv_i}{dt} &= \frac{\mu_{vv}}{N_v}\sum_{j}(v_j-v_i) + \frac{\mu_{pv}}{N_p}\sum_{\beta}(p_\beta-v_i)\notag\\
    &= \mu_{vv} \left(\langle v \rangle - v_i \right) + \mu_{pv}\left(\langle p\rangle - v_i \rangle \right)\notag\\
    &= (\mu_{vv} + \mu_{pv})\left(\frac{\mu_{vv}\langle v \rangle + \mu_{pv}\langle p\rangle}{\mu_{vv} + \mu_{pv}} - v_i\right).\label{eq:dvisufficient}
\end{align}
Thus, $v_i$ exponentially decays to the weighted voter and party average, where the weights are determined by the strengths of the forces acting on voters. It is clear that $\frac{\mu_{vv}\langle v \rangle + \mu_{pv}\langle p\rangle}{\mu_{vv} + \mu_{pv}} \in \Omega(\varphi(t)) \subseteq E(t)$. The evolution of the parties, \eqref{eq:HKb}, can be written for $t=0$ as
\begin{align}
    \frac{dp_\alpha}{dt} = \left(\mu_{vp} - \mu_{pp}\right)\left(\frac{\mu_{vp}\langle v \rangle - \mu_{pp}\langle p \rangle}{\mu_{vp} - \mu_{pp}} - p_\alpha\right) .
\label{eq:dpsufficient}
\end{align}
Hence, if $\mu_{vp} > \mu_{pp}$, $p_\alpha$ decays exponentially to $\psi(t)=\frac{\mu_{vp}\langle v \rangle - \mu_{pp}\langle p \rangle}{\mu_{vp} - \mu_{pp}} \in E(t)$ by construction. Hence, $E(t)$ is a boundary on the party dynamics, as well as the voter dynamics for all $t>0$. Next, observe that $\langle v(t) \rangle, \langle p(t)\rangle \in E(t)$. The evolution of the mean positions of the voters and parties are given by
\begin{align}
    \frac{d\langle v \rangle}{dt} &= \mu_{pv}\left(\langle p \rangle - \langle v \rangle \right)\notag,\\
    \frac{d\langle p \rangle}{dt} &= \mu_{vp}\left(\langle v \rangle - \langle p \rangle \right)\notag,
\end{align}    
implying
\begin{align}    
   \frac{d \left(\langle v \rangle - \langle p \rangle\right)}{dt} &= -(\mu_{pv} + \mu_{vp})\left(\langle v \rangle - \langle p \rangle\right),
\end{align}
which is readily solved to yield
\begin{align}
    \langle p(t) \rangle = \langle v(t) \rangle + e^{-(\mu_{pv} + \mu_{vp})t}A,
\end{align}
where $A=\langle v(0) \rangle - \langle p(0)\rangle$ is constant. This implies that $\langle p(t) \rangle \to \langle v(t) \rangle$ for $t\to \infty$ with exponential rate of convergence $\mu_{pv} + \mu_{vp}$. As $\langle p(t) \rangle \to \langle v(t) \rangle$ this implies
\begin{align}
    \frac{\mu_{vp}\langle v(t) \rangle - \mu_{pp}\langle p(t) \rangle}{\mu_{vp} - \mu_{pp}} \to  \langle v(t) \rangle 
    \qquad {\rm{and}} \qquad  
    \frac{\mu_{vv}\langle v(t) \rangle + \mu_{pv}\langle p(t)\rangle}{\mu_{vv} + \mu_{pv}} \to \langle v(t)\rangle,
\end{align}
which implies according to (\ref{eq:dvisufficient}) and (\ref{eq:dpsufficient}) that the system approaches consensus, as desired.

Additionally, \eqref{eq:dvisufficient} and \eqref{eq:dpsufficient} imply that $\Omega(\varphi({t_2})) \subset E(t_1)$ for some $t_2 > t_1 \ge 0$. We may write $v(t)$ explicitly as
\[v(t) = \langle v(t)\rangle - \frac{\mu_{pp}}{\mu_{vp}-\mu_{pp}}e^{-(\mu_{pv}+\mu_{vp})t}A,\]
which implies that $v(t)$ monotonically converges to $\langle v(t)\rangle$ and hence we have,
\begin{align}
    E(t_2) \subset E(t_1).
\end{align}
\end{proof}

The result of Proposition~\ref{prop:E} may be interpreted in an intuitive way. For a given set of initial conditions with parameters that lead to consensus, increasing the repulsive force strength $\mu_{pp}$ or decreasing the interaction force strength $\mu_{pv}$ will lead to a lack of consensus. We illustrate the result in Figure~\ref{fig:repulsiveconvergence} where we show how the two conditions \eqref{eq:Prop1} and \eqref{eq:Prop2} affect the formation of unanimous consensus. The top subfigure shows $100$ voters uniformly distributed across $[0,1]$ with $4$ parties placed at $p_1(0)=0.1$, $p_2(0)=0.3$, $p_3(0)=0.5$ and $p_4(0)=0.7$. The parameters are chosen such that they conform with both conditions of the Proposition. We observe that voters and parties reach unanimous consensus, with the parties converging at a slower rate due to their repulsive forces $\mu_{pp}=0.4$ being smaller than the attractive party-voter force with $ \mu_{vp}=0.5$. In the middle subfigure of Figure~\ref{fig:repulsiveconvergence} we show the effect of breaking condition \eqref{eq:Prop1} in preventing convergence to unanimous consensus. The initial conditions and all parameters are kept the same as above except now $\mu_{pp}= \mu_{vp}=0.5$, violating condition \eqref{eq:Prop1}. Here, voters still converge to a single cluster under their attractive forces since the interaction radii cover all the voters initially, i.e. condition \eqref{eq:Prop2} is satisfied. However, the parties do not converge, but instead move slowly to 4 different stationary positions that lead to a net-zero force on the parties due to the attractiveness of the voters equalling the repulsiveness of the parties. Finally, unanimous consensus can also be prohibited by reducing the interaction radii, i.e. violating the second condition \eqref{eq:Prop2}. This is shown in the bottom subfigure of Figure~\ref{fig:repulsiveconvergence}. Here condition \eqref{eq:Prop1} is satisfied, however, the interaction radii are all $0.1$ violating condition \eqref{eq:Prop2}. This leads to the formation of distinct non-interacting clusters and a lack of unanimous consensus. Note that to maintain a simple interpretation of the convex hull, we have produce these figures with reflective boundary conditions as opposed to periodic shown elsewhere.

Note that the result of Proposition~\ref{prop:E} is independent of the dimension $d$. Also note that in \cite{motsch2014heterophilious} it is shown that the classical Hegselmann--Krause model \eqref{eq:HK0} satisfies $\Omega(t_2) \subset \Omega(t_1)$. This is not the case in our model because it is possible for parties near the boundary of the convex hull to be repelled beyond the boundary before being attracted back to consensus. $E(t)$ can be thought of as the area containing the convex hull with a buffer region by which the voters and parties are bounded.  Indeed, \eqref{eq:consensusrepulsion} is in practice a weak condition for the emergence of consensus, and consensus can occur with smaller interaction radii $R$ so long as the support of the strongest force covers a sufficiently large portion of the initial distribution. In an aside, we note that in the middle subfigure of Figure~\ref{fig:repulsiveconvergence}, $Q_{vv}=1$ throughout the whole time period, which is not representative of the behaviour of the system. Alternatively, $\hat{D}$ increases from $0$ to $~0.45$ before stabilising and remaining there - which correctly detects that the parties have not converged to uniform consensus (not shown).

\begin{figure}[!htpb]
    \centering
     \includegraphics[width=0.6\textwidth, height=0.3\textwidth]{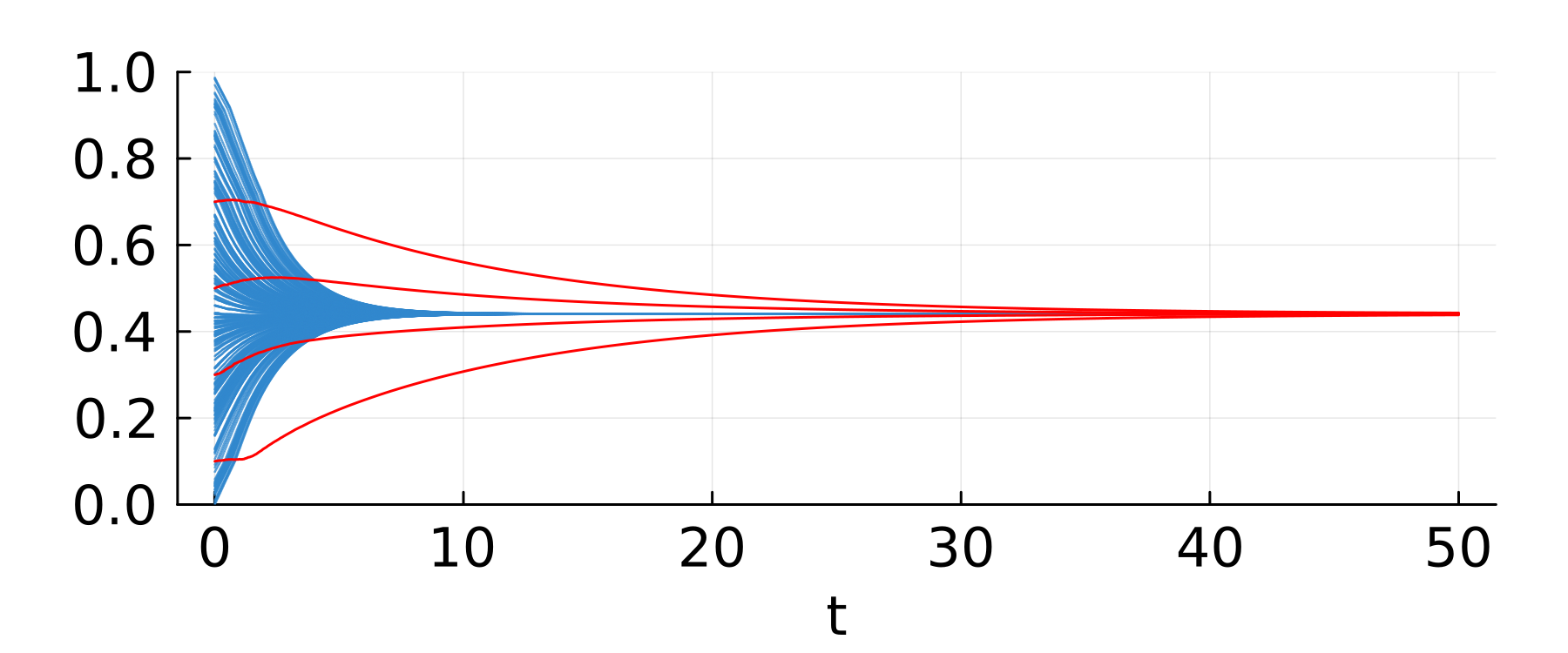}
     \includegraphics[width=0.6\textwidth, height=0.3\textwidth]{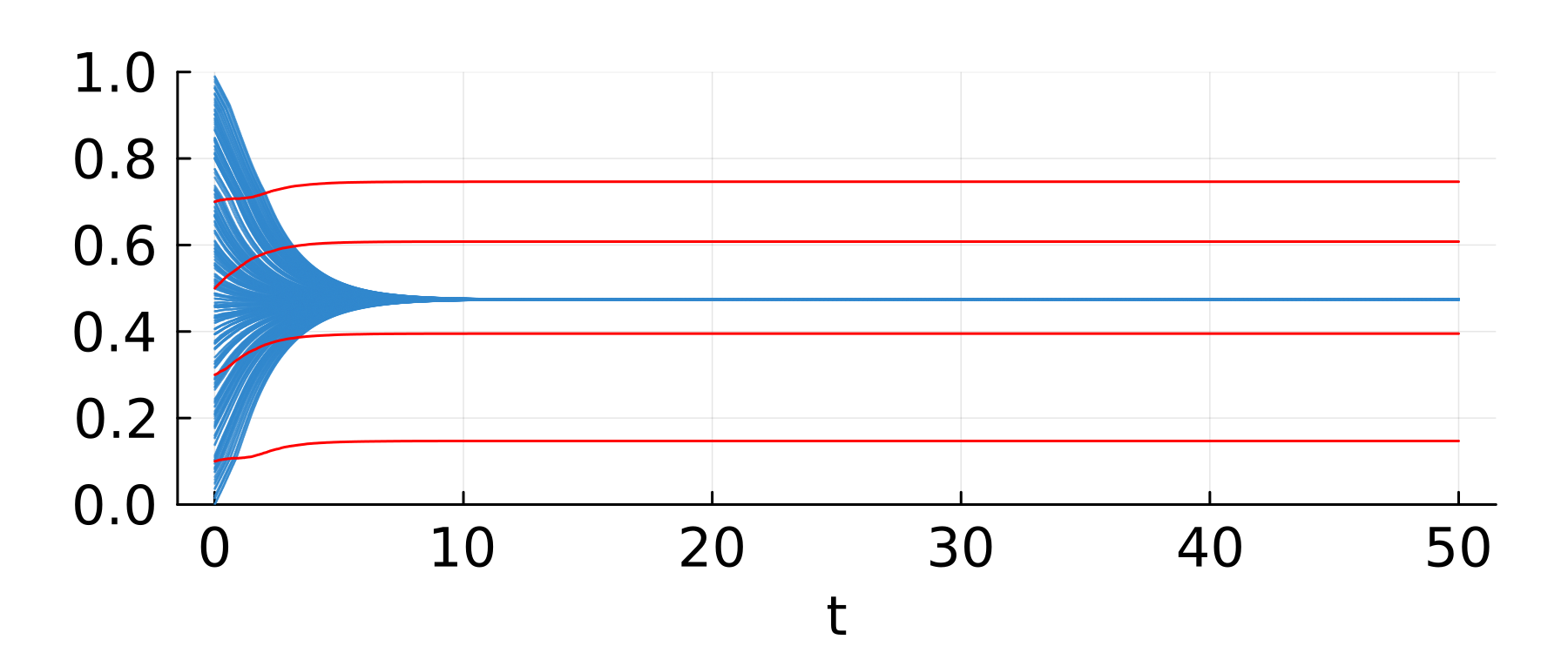}
      \includegraphics[width=0.6\textwidth, height=0.3\textwidth]{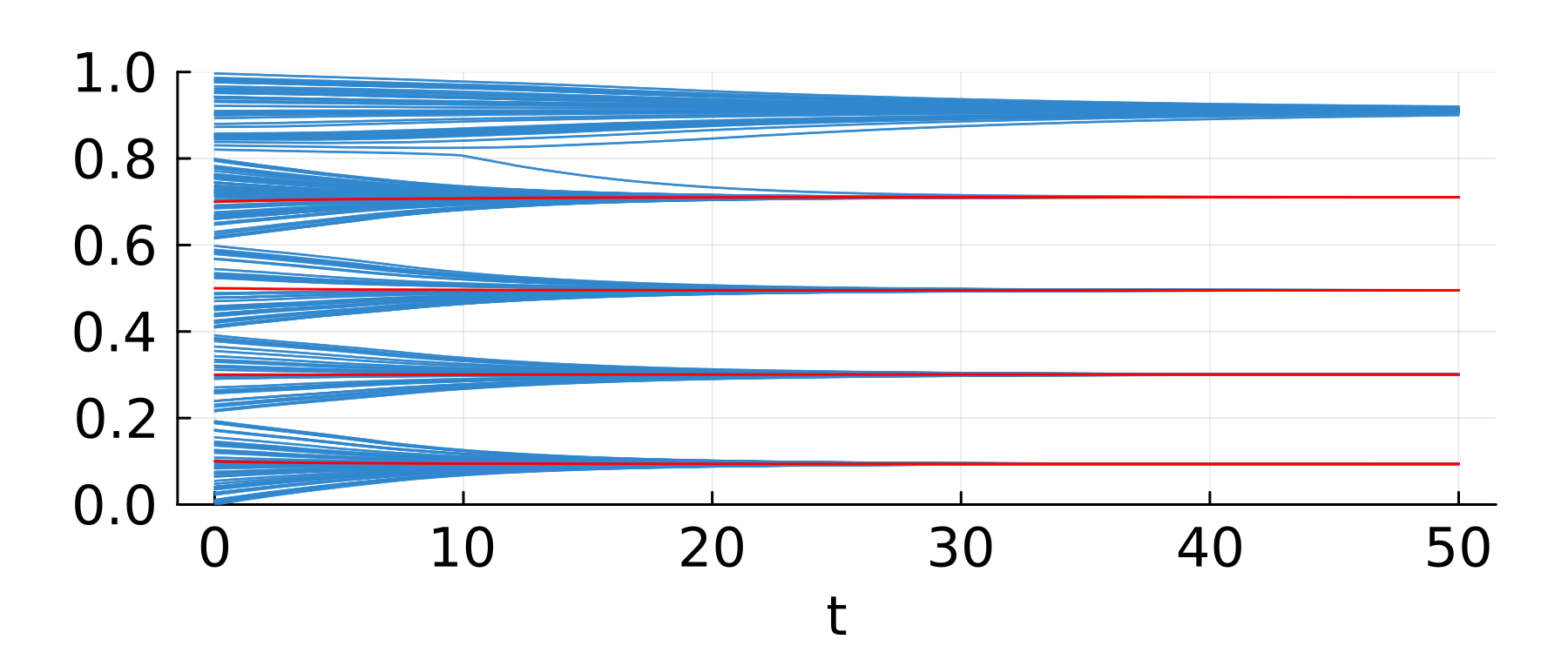}
    \caption{Evolution of the noise-less modified Hegselmann--Krause model (\ref{eq:HKa})--(\ref{eq:HKb}) in a $d=1$-dimensional opinion space with $\sigma_v=\sigma_p=0$ for $100$ voters (blue) and $4$ parties (red). Initially voters are distributed uniformly across $[0,1]$ and parties are initially at $p_1(0)=0.1$, $p_2(0)=0.3$, $p_3(0)=0.5$ and $p_4(0)=0.7$. Parameters are $\mu_{vv}=\mu_{pv}=0.3$ and $\mu_{vp}=0.5$. Here we use reflective boundary conditions. Top: unanimous consensus when both conditions \eqref{eq:Prop1} and \eqref{eq:Prop2} are satisfied with $\mu_{pp}=0.4$ and $R_{vv}=R_{pv}=R_{vp}=R_{pp}=0.6$. Middle: no unanimous consensus when condition \eqref{eq:Prop1} is not satisfied with $\mu_{pp}=0.5$ and condition \eqref{eq:Prop2} is satisfied with $R_{vv}=R_{pv}=R_{vp}=R_{pp}=0.6$. Bottom: no unanimous consensus when condition \eqref{eq:Prop2} is not satisfied with $R_{vv}=R_{pv}=R_{vp}=R_{pp}=0.1$, and condition \eqref{eq:Prop2} is satisfied with $\mu_{pp}=0.4$.}
    \label{fig:repulsiveconvergence}
\end{figure}
%


\section{The mean-field limit and phase transitions}
\label{sec:meanfield}

The original noisy Hegselmann--Krause model undergoes a phase transition as the noise strength increases \cite{wang2017noisy} whereby instead of the formation of clusters, voters approach a uniform distribution. In this section, we show that the inclusion of parties does not affect the presence of the phase transition, and we derive an analytical approximation of the critical noise strength, $\sigma_c$. This transition from an ordered state of voters to a uniform distribution of voters occurs when $\sigma_v>\sigma_c$. While large noise is unrealistic -- as voters would change their opinions in a random walk -- this exercise allows us to study the effect of parties on consensus formation. Perhaps unsurprisingly, we show that parties enhance the formation of voter clusters.

To study the stability of the consensus state when noise is added and to study the phase transition discussed in Section~\ref{sec:consensus}, we consider the mean-field limit $N_v \to \infty$  and derive an equation for the voter density $\rho(v,t)$ \cite{Toscani06,AlbiEtAl17,PareschiEtAl19,Zanella23}. To determine the phase transition we perform a linear stability analysis on the mean-field limit equation. As stated above, for convenience we consider periodic boundary conditions with $v_i={\rm{mod}}(v_i,1)$. 

The classical Hegselmann--Krause model (\ref{eq:HK0}), which does not contain interactions with political parties, allows for a straight forward derivation of the evolution equation for the voter density $\rho(v,t)$. The voter density is the limit for $N_v \to \infty$ of the empirical measure 
\begin{align}
\rho^{(N_v)}(t,dv)= \frac{1}{N_v}\sum_{j=1}^{N_v} \delta_{v_j(t)}(dv).
\end{align}
The derivation of the corresponding equation for the evolution of the limiting density $\rho(v,t)$ is achieved by the Bogoliubov-Born-Green-Kirkwood-Yvon (BBGKY) hierarchy, and hinges on the exchangeability of voters. For the modified Hegselmann--Krause model (\ref{eq:HKa})-(\ref{eq:HKb}), which includes party dynamics, the mean-field limit $N_v\to\infty$ can still be taken, however, one typically only has a small number of political parties and hence the limit $N_p\to\infty$ is unrealistic.

To derive an equation for the voter density associated with our modified Hegselmann--Krause model we make the following assumption about the dynamics of the political parties $p_\alpha$. It is reasonable to assume that voters are much more willing to change their opinion on certain issues than parties. Political parties have a higher inertia than voters and therefore change in opinion space more slowly than individual voters. Indeed, \cite{cameron20222022} showed that as of the 2022 Australian Federal election only $37\%$ of voters have supported the same party at every election, which suggests that individual voters are quite willing to change their opinions. Political parties, on the other hand, attempt to reinforce the issues they are traditionally seen as strong in, as well as engage in ideological politics \citep{ames2001electoral}. Hence, we assume that the party dynamics is slow compared to the voter dynamics and we assume  $\mu_{vp},\mu_{pp}\ll \mu_{vv}, \mu_{pv}$ and $\sigma_p \ll \sigma_v$. This allows us to treat for some time $0\le t<T$ the parties as constant parameters in the voter dynamics (\ref{eq:HKa}) leading to an evolution equation for voters with frozen parties with  
\begin{align}
    dv_i =  \mu_{vv}\sum_j \phi\left(\frac{||v_j-v_i||}{R_{vv}}\right)(v_j - v_i)\, dt + \mu_{pv}F_{pv}(v_i;p)\,dt + \sigma_v dW^i(t),
    \label{eq:HKap}
\end{align}
where $F_{pv}(v_i;p) = \sum_\beta \phi{\left(\frac{||p_\beta-v_i||}{R_{pv}}\right)}(p_\beta-v_i)$ is the effect of the parties on the voter, and $p=\begin{pmatrix}p_1 & \hdots & p_{N_p}\end{pmatrix} \in \mathbb{R}^{dN_p}$ is frozen. The voter dynamics (\ref{eq:HKap})  allows for an application of the BBGKY hierarchy to derive the following equation for the one-voter density $\rho(v,t;p)$, conditioned on the party parameters $p$, 
\begin{align}
    \partial_t \rho(v,t;p) &= -\nabla \cdot [\rho(v,t;p)\left(\mu_{vv} K \star\rho + \mu_{pv}F_{pv}(v;p)  \right)] + \frac{\sigma^2_v}{2}\Delta \rho(v,t;p) 
\label{eq:mckean}
\end{align}
with initial data
\begin{align}    
\rho(v,0;p) &= \rho_0(v;p)
\end{align}    
and interaction kernel 
\begin{align}
K\star\rho =\int \phi\left(\frac{||w-v||}{R_{vv}}\right)(w-v)\rho(w)dw.
\end{align}
To study the phase transition from global consensus to a near-uniform distribution of voters upon increasing the noise strength $\sigma_v$ we follow the pipeline used by \citet{garnier2017consensus} and \citet{wang2017noisy} for the classical Hegselmann--Krause model \eqref{eq:HK0}. 

Equation~\eqref{eq:mckean} for the stationary-party Hegselmann--Krause model allows for the stationary uniform voter distribution $\rho_0(v;p) = 1$ for $v\in[0,1]^d$ provided that the parties are equidistantly distributed in opinion space. In the case that they are not equidistantly spaced, $\rho_0$ still may serve as a good approximation for sufficiently small values of $\mu_{pv}$.

To identify the critical noise strength $\sigma_c$ below which consensus occurs, we perform a linear stability analysis around the stationary uniform state. We consider an expansion in small $\mu_{pv}$
\begin{align}
\rho(v,t) = 1 + \mu_{pv} \rho_1(v,t) + \mathcal{O}(\mu_{pv}^2).\label{eq:small_mupv}
\end{align}
We note that a small $\mu_{pv}$ may not be realistic when parties exert strong leadership effects over voters. Substituting \eqref{eq:small_mupv} into \eqref{eq:mckean} yields
\begin{align}
    \frac{\partial \rho_{1}}{\partial t}=-\nabla \cdot \left[\left(  \mu_{vv}K \star\rho_{1} + \rho_{1}(v,t)\mu_{pv}F_{pv}(v;p) \right)\right]
       + \frac{\sigma^2_v}{2}\Delta \rho_{1}(v,t), 
\label{eq:linearized}
\end{align}
where we used $K \star \rho_0=0$. Employing periodic boundary conditions, we can solve this linear equation for $\rho_1$ by a Fourier ansatz $\rho_{1}(v,t) = \sum_{\boldsymbol{k}} T_{\boldsymbol{k}}(t)e^{2\pi i \boldsymbol{k} \cdot v}$ for $\boldsymbol{k} \in \mathbb{Z}^d$ with $k:=||\boldsymbol{k}||_2\neq 0$ to ensure that $\rho$ is a probability density. The convolution term can be simplified using
\begin{align}
    -\nabla \cdot (K\star e^{2\pi i \boldsymbol{k} \cdot v}) 
    &= -\nabla \cdot \int_{||\mathrm{z}||\le R_{vv}} \mathrm{z} e^{2\pi i \boldsymbol{k} \cdot (\mathrm{z}+v)}d\mathrm{z}\notag\\
    &= sR_{vv}e^{2\pi i \boldsymbol{k}\cdot v}\int_{||\mathrm{z}||\le 1}z\left(\sin(sz) - i\cos(sz)\right)d\mathrm{z}, 
\label{eq:convolutionintegral}
\end{align}
where we introduced the scalar $z = \frac{\boldsymbol{k} \cdot z}{k}$ and the scaled wave number modulus $s = s(k) = 2\pi R_{vv}k$. Restricting to a single Fourier mode, the linearized  equation \eqref{eq:linearized} is written as
\begin{align}
T_{\boldsymbol{k}}^\prime(t) = A(k;v) T_{\boldsymbol{k}}(t),
\end{align}
with
\begin{align}        
      A(k;v)  &= -\mu_{pv}\left[\nabla \cdot F_{pv}(v;p) + 2\pi i \boldsymbol{k} \cdot F_{pv}(v;p)\right] - \frac{s^2}{2R_{vv}^2} \sigma_v^2 \notag \\
  &\;  + \mu_{vv} sR_{vv}\int_{||\mathrm{z}||\le 1}z \left(\sin(sz) - i\cos(sz)\right)d\mathrm{z} .
\label{eq:A}
\end{align}
Note that since $A(k;v)$ depends on $v$ the differential equation for $T$ has non-constant coefficients. We define the largest possible growth rate as 
\begin{align}
\gamma(k) = \max_{v} \text{Re} A(k;v).
\end{align}
We now set out to determine $\gamma(k)$, which will allows us to determine the critical noise strength $\sigma_c$ such that $\gamma(k)\le 0$ for all $k$ for all $\sigma_v<\sigma_c$.  
We make the additional assumption that each voter is affected by all parties, i.e. $R_{pv}$ is sufficiently large. For convenience, we will estimate the growth rate in the rescaled variable $s$. The maximal growth rate, expressed in terms of $s=s(k)$, is evaluated as 
\begin{align}
    \gamma (s) &= \max_v \text{Re}A(v) \notag\\
               &= \max_v  \left[ -\mu_{pv}[\nabla \cdot F_{pv}(v;p)] - \frac{\sigma_v^2}{2R_{vv}^2}s^2 + \mu_{vv} sR_{vv} \int_{-1}^{1} z\sin(sz)\, dz \right] \nonumber\\
               &= \mu_{pv}d  - \frac{\sigma_v^2}{2R_{vv}^2}s^2 + \mu_{vv} sR_{vv} \int_{-1}^{1} z\sin(sz)dz,
\label{eq:gammaS}
\end{align}
where we used that $||-\nabla \cdot F_{pv}(v;p)||\le d$ with equality when voters are affected by all $N_p$ parties at all times, which we assumed to be true. We first present results for a $1$-dimensional opinion space and then for higher-dimensional opinion spaces.

\subsection{One-dimensional case \texorpdfstring{$d=1$}{d=1}}
For $d=1$ the integral term in \eqref{eq:gammaS} can be explicitly evaluated as 
\begin{align}
sR_{vv} \int_{-1}^{1} z\sin(sz)dz = 2R_{vv}\left(\frac{\sin(s)}{s}-\cos(s)\right), 
\end{align}
and \eqref{eq:gammaS} becomes
\begin{equation}
    \gamma(s) = \mu_{pv} - \frac{\sigma_v^2}{2R_{vv}^2} s^2  + 2\mu_{vv}R_{vv}\left(\frac{\sin(s)}{s}-\cos(s)\right) .
\label{eq:gammak}
\end{equation}

\begin{figure}[htbp]
    \centering
     \includegraphics[width=0.8\textwidth]{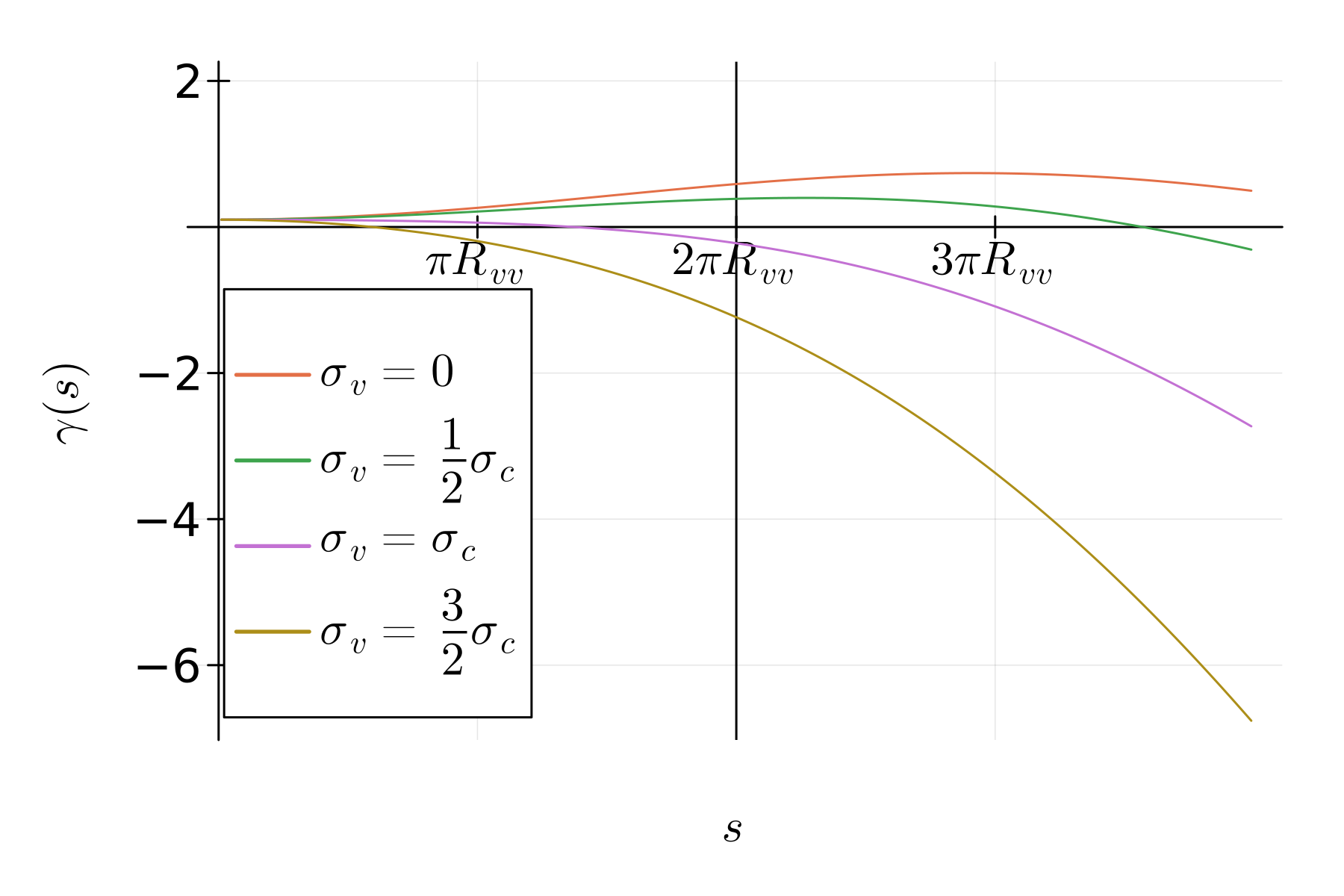}
    \caption{Growth rate $\gamma(s)$, \eqref{eq:gammak}, for different values of the noise strength $\sigma_v$. The vertical black line denotes the smallest occurring wave number $s_{\rm{min}}=2\pi R_{vv}$ for $k=1$.  The requirement for linear stability is that $\gamma(s) < 0$ for all $s\geq s_{\rm{min}}$. $R_{vv} = 0.3, \mu_{vv}=1, \mu_{pv} = 0.1$ which corresponds to $\sigma_c=0.2$.}
    \label{fig:gammakplots}
\end{figure}

The negative term quadratic in $s$ ensures that (scaled) high wave numbers $s$ are linearly stable. We hence focus on small wave number instability. Since the smallest wave number is $k = 1$, we have $s \ge 2\pi R_{vv}$, and small $s$ is achieved for small $R_{vv}$. We can find an explicit small wave number expansion of $\gamma(s)$ by performing a Taylor expansion of (\ref{eq:gammak}) in $s$, leading to
\begin{align}
    \gamma(s) 
    &= \mu_{pv} + s^2 \left(\frac{2}{3}\mu_{vv}R_{vv}-\frac{\sigma_v^2}{2R_{vv}^2}\right).
    \label{eq:gamma_Taylor}
\end{align}
The uniform state is linearly unstable if for any $s\ge 2\pi R_{vv}$ the growth rate is positive with $\gamma(s)>0$. In particular, we require small wave number instability at $s=s_{\rm{min}}=2\pi R_{vv}$. We deduce from \eqref{eq:gamma_Taylor} that small wave number instability occurs for noise strength with
\begin{align}
    \sigma_v^2 < \frac{\mu_{pv}}{2\pi^2}  + \frac{4}{3}\mu_{vv}R_{vv}^3, 
\end{align}
which implies the critical noise strength
\begin{align}
    \sigma_c^2 = \frac{\mu_{pv}}{2\pi^2}+\frac{4}{3}\mu_{vv}R_{vv}^3 .
\label{eq:smallRsigmac}
\end{align}

Note that for the standard noisy Hegselmann--Krause model (\ref{eq:HK0}) with $\mu_{pv} = 0$ and $\mu_{vv} = 1$, we recover the known critical noise strength $\sigma_c^2 = \frac{4}{3}R_{vv}^3$ \cite{wang2017noisy,garnier2017consensus}. When $\mu_{pv} >0,$ the inclusion of party dynamics shifts the phase transition to higher values of the noise strength. This can be understood heuristically as parties provide additional stability to a large group of like-minded voters. We remark that the stationary party model (\ref{eq:HKap}) is unable to recover unanimous consensus. However, there is a clear phase transition, as we will see below, from an ordered state of party-base clusters or voter consensus clusters to a disordered state of uniformly distributed voters. Figure \ref{fig:gammakplots} shows the growth rate $\gamma(s)$ for various values of the noise strength $\sigma_v$. It is seen that for small values of $\sigma_v$ the growth rate $\gamma(s)$ is positive for a range of values of $s$. For larger values of $s$, the growth rate can again increase obtaining positive values (not shown).\\

Figure~\ref{fig:improvedtransitions} shows a phase diagram obtained from a long simulation of the stationary-party model (\ref{eq:HKap}) in $(\sigma_v,R_{vv})$-space. We consider $2,000$ voters that are initially distributed uniformly on $[0,1]$ and stationary parties, $p_1=0.86, p_2 = 0.53$ and $p_3=0.2$. We simulated until time $t=500$ with $\Delta t= 0.1$. We tested for statistical equilibrium of the voter dynamics using an Augmented Dickey-Fuller (ADF) test \cite{cheung1995adftest}, testing for stationarity over the last $20\%$ of the simulation. A phase transition is clearly seen, quantified by the consensus diagnostic $\hat{D}$. To best visualize the phase transition and the departure from uniformity, which for $2,000$ voters uniformly sampled across $[0,1]$ yields $\hat{D}=0.037$, we employ a colour map with the colour transition occurring at $0.039$. The slightly larger value of $\hat D=0.039$ was chosen to allow for the detection of sufficiently large deviations from uniformity. The phase transition is well approximated by our approximation (\ref{eq:smallRsigmac}) for small values of $R_{vv}\lessapprox 0.15$, consistent with the approximation of $R_{vv}\ll 1$ we made to derive (\ref{eq:smallRsigmac}). Note that the transition in our finite-size system with $2,000$ voters is gradual rather than abrupt as suggested by the mean-field theory. The orange region of high level consensus in the phase diagram in Figure~\ref{fig:improvedtransitions} occurring for $R_{vv}>0.12$ and $\sigma > 0.03$ is reached via a transition from a state of three party-base clusters to a single swing voter cluster of smaller size akin to the voter consensus scenario discussed in Section~\ref{sec:examples} (not shown).  

Figure~\ref{fig:muphasetransition} shows the corresponding phase diagram in $(\mu_{pv},\sigma_v)$-space. The critical noise strength $\sigma_c$ exhibits a $\sqrt{\mu_{pv}}$ dependency, consistent with our approximation (\ref{eq:smallRsigmac}), illustrating the effect of the parties on the phase transition. Note that although the approximation \eqref{eq:smallRsigmac} is based on an expansion in small $\mu_{pv}$ (cf \eqref{eq:small_mupv}), it is valid up to values of $\mu_{pv}\approx 0.6$.

\begin{figure}[htbp]
    \centering
     \includegraphics[width=0.8\textwidth]{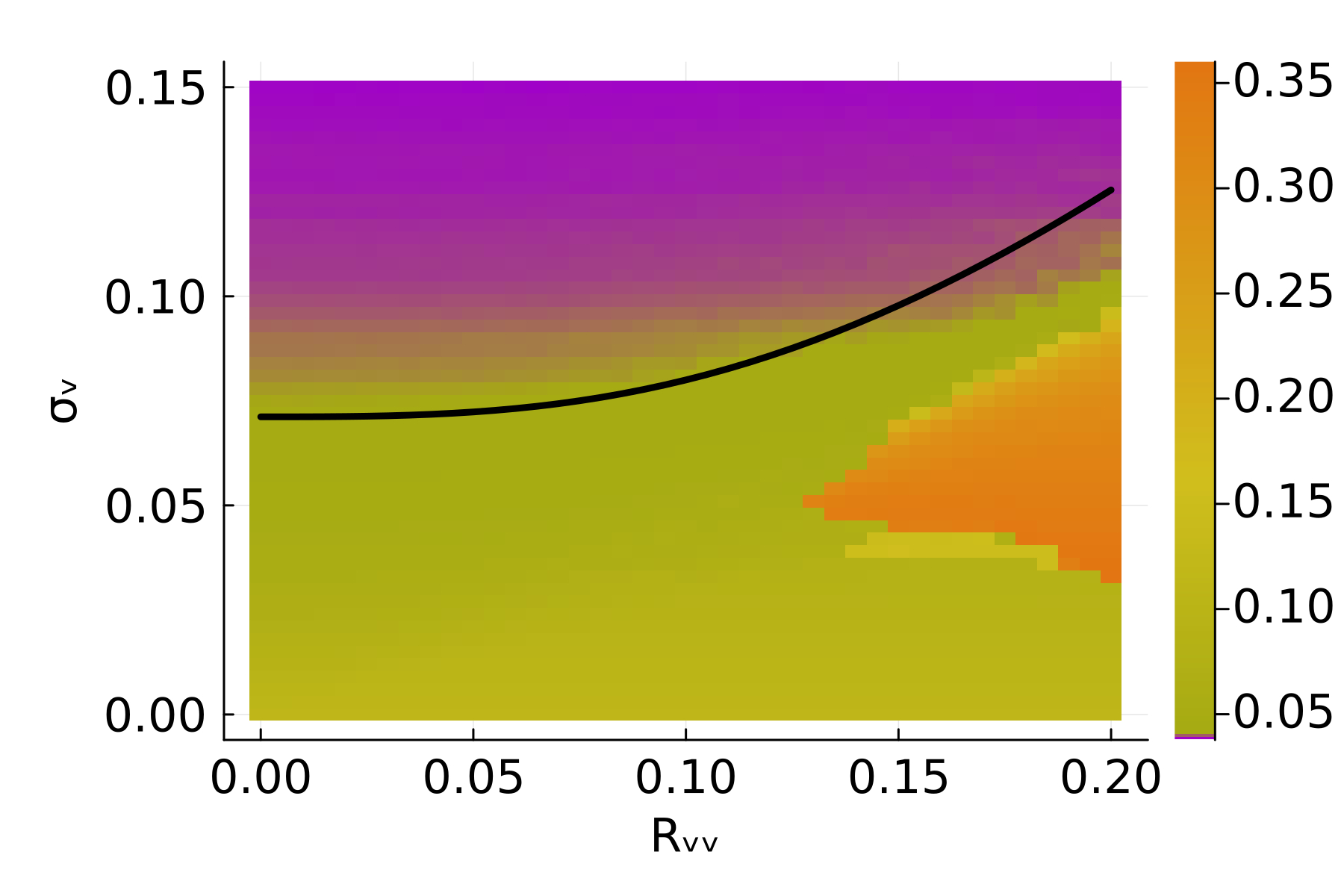}
    \caption{Phase diagram for the stationary-party-model \eqref{eq:HKap} for $2,000$ voters randomly distributed across $[0,1]$ and with three parties, $p_1=0.86, p_2 = 0.53$ and $p_3=0.2$. Parameters are $\mu_{pv} = 0.1$, $\mu_{vv}=1$ and $R_{pv}=0.5$ with  $R_{vv}$ varying from $0$ to $0.2$. Colours show the value of the consensus diagnostic $\hat D$, averaged from $t=400$ to $t=500$. The black line shows the analytical approximation \eqref{eq:smallRsigmac} of the critical curve.}
    \label{fig:improvedtransitions}
\end{figure}
\begin{figure}[htbp]
    \centering
     \includegraphics[width=0.8\textwidth]{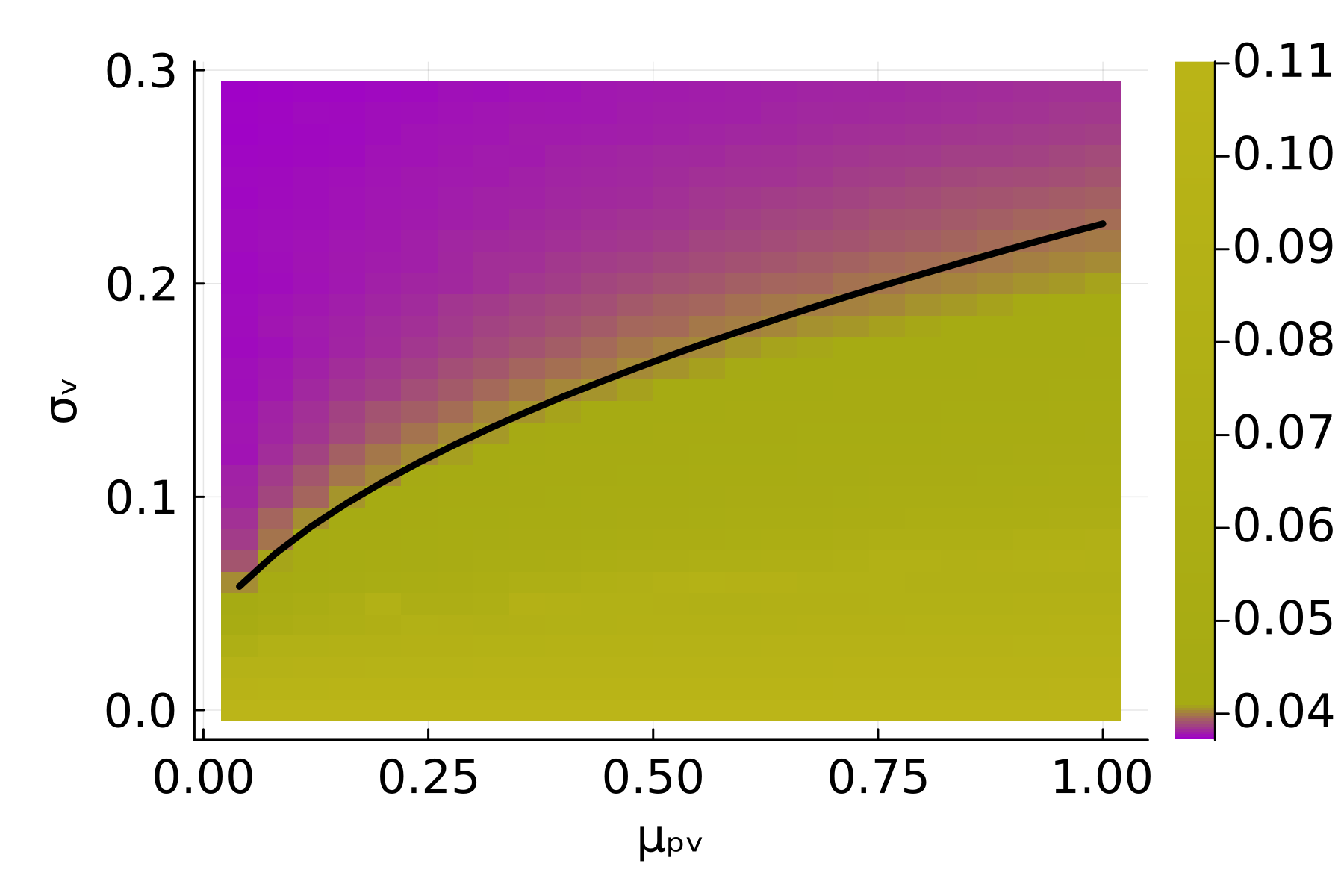}
    \caption{Phase diagram for the stationary-party-model \eqref{eq:HKap} for $2,000$ voters randomly distributed across $[0,1]$ and with three parties, $p_1=0.86, p_2 = 0.53$ and $p_3=0.2$. Parameters are $R_{vv} = 0.1$, $\mu_{vv}=1$ and $R_{pv}=0.5$ with $\mu_{pv}$ varying from 0 to 1. Colours show the value of the consensus diagnostic $\hat D$, averaged from $t=400$ to $t=500$. The black line shows the analytical approximation \eqref{eq:smallRsigmac} of the critical curve.} 
    \label{fig:muphasetransition}
\end{figure}
%


\subsection{Higher dimensional cases}
For $d\ge 2$, the maximal growth rate can be found from \eqref{eq:A} as
\begin{align}
    \gamma(s) = d \mu_{pv} + \mu_{vv}sR_{vv}\int_{||\mathrm{z}||\le 1}z\sin(sz)d\mathrm{z}- \frac{s^2}{2R_{vv}^2} \sigma_v^2 ,
\label{eq:gammakhigherdim}
\end{align}
where we assumed again that $R_{pv}>1/2$. Recall $z = \hat{\boldsymbol{k}}\cdot \mathrm{z}$ with unit vector $\hat{\boldsymbol{k}} = \boldsymbol{k}/k$.  For $s\ll 1$ we Taylor expand the integral term to obtain
\begin{align}
   sR_{vv}\int_{||z'||\le 1}z''\sin(sz'')dz'  &\approx \frac{R_{vv}}{d}s^2 \int_{||z'||\le 1}||z'||^2 dz' \notag\\
    &=  \frac{R_{vv}}{d}s^2 \frac{1}{d+2}S_{d-1},
\end{align}
where $S_{d-1} = \frac{2\pi^{d/2}}{\Gamma(d/2)}$ is the surface area of the $d$-dimensional unit sphere with $\Gamma(t)=\int_0^\infty x^{t-1}e^{-x}dx$.  
Hence, for $d\ge 2,$ we obtain 
\begin{align}
    \gamma(s) = d \mu_{pv} + \frac{2\pi^{d/2}}{d(d+2)\Gamma(d/2)}\mu_{vv}R_{vv}s^2 - \frac{\sigma_v^2}{2R_{vv}^2}s^2.
\end{align}
Applying a further approximation for $R_{vv}\ll 1$, we obtain the critical noise strength
\begin{align}
    \sigma_c^2 = d\frac{\mu_{pv}}{2\pi^2} + \frac{4 \pi^{d/2}}{d(d+2)\Gamma(d/2)}\mu_{vv}R_{vv}^3. 
\label{eq:higherdimsigma_c}
\end{align}
For the classical noisy Hegselmann--Krause model (\ref{eq:HK0}) with $\mu_{vv}=1$ and $\mu_{pv} = 0$, (\ref{eq:higherdimsigma_c}) reduces to the approximation obtained by \citet{wang2017noisy}. Note that in the classical noisy Hegselmann--Krause model (\ref{eq:HK0}) the critical noise strength $\sigma_c$ decreases as the dimension decreases. The inclusion of political parties, reflected in the linear contribution $d \mu_{pv}$, will be dominant for sufficiently large dimension $d$, consistent with our premise that parties have a stabilizing effect on the opinion dynamics of voters.

In summary, we have shown that the modified Hegselmann-Krause model displays a similar phase transition as the one found for the original noisy Hegselmann--Krause model \cite{garnier2017consensus,wang2017noisy}. We find that political parties act to enhance the formation of voter clusters, providing a stabilising force against voter noise. A remarkable finding is that this effect increases with dimension. This suggests that as the number of political topics increases, parties become more effective at stabilising the dynamics compared to the case when only a few topics dominate the political debate.


\section{Discussion}
\label{sec:disc}
We introduced and analyzed a modified Hegselmann--Krause model which describes the interactions of voters and parties in a $d$-dimensional opinion space. The model exhibits cluster formation and a phase transition from unstructured dynamics to unanimous consensus when all voters and parties collapse into the same region in opinion space. The model exhibits rich dynamical behaviour depending on the interaction radii of the voters and parties and the strength of the mutual interactions.

We established a sufficient condition for consensus in the deterministic version which states that consensus is guaranteed if the interaction radii are sufficiently large allowing for the interaction of all voters and parties and if the attractive forces dominate over the repulsive forces exerted by parties to delineate themselves from each other. We further employed mean-field theory to find the critical noise strength below which consensus occurs. Our analytical formula reflects a stabilizing effect of parties on consensus formation. 

The proposed model recreates important and complex political dynamics such as party-base clusters, swing voters, disaffected voters and transitions between those states.

The model typically exhibits clusters of voters around individual parties. These clusters form what is known as the party-base of a party. The party-base represents an important feature of politics as it is a core group of voters who consistently support a specific political party \cite{miller20202020}. The model further demonstrates the emergence of swing voters as a cluster of voters situated in opinion space between two parties. The parties then compete for the preference of the voters mediated by the repulsive force between them. The distances in opinion space between a voter from the swing cluster and each of the two parties are similar, so small changes in relative party positions can change which party is closest to a particular voter, and hence determines what party they would vote for. The model further supports clusters of disaffected voters that are not aligned with any political party. Such clusters of disaffected voters are generated when political parties evolve into regions in opinion space with large voter mass potentially leaving behind disaffected voters which do not experience any attracting force to the party if their distance is sufficiently large. This latter phenomenon is of increasing importance in modern political science \cite{e2001individual}. More extreme political scenarios, such as a sudden collapse of voters to a single party, can be found in the modified Hegselmann--Krause model. These scenarios occurred in our model in bounded domains of dimensions $d=1$ and $d=2$. In our simulations, the model approached consensus after a sufficiently long period of time. We remark that in unbounded domains and opinion spaces of dimension $d\ge 3$ recurrence of a random walker is not guaranteed \cite{doyle1984recurrence} and merging of clusters into unanimous consensus is unlikely.

While the model assumes that the evolution of voter opinion is due to the relative positions in opinion space, the reality is far more complex. Media consumption and lack of information play an important role in contributing to voters' decisions \cite{prior2013media}, which is not covered by the model. Another limitation of the model is the assumption that all parties and all voters have the same interaction radii $R_{pp}$ and $R_{vv}$, respectively, and exert the same force on the other agents. Political parties are clearly not all equal - it is conceivable that some parties exert stronger attractive forces on voters for example, due to political charisma or effective advertising campaigns. Similarly, some voters may be more open to the thoughts and opinions of others, so their interaction radii might be larger than those of other voters. In our model parties can freely meander through opinion space, either due to their own stochastic driving force or due to attraction to far away voter clusters (for a sufficiently large interaction radius $R_{vp}$). To ensure that parties remain in some bounded region in opinion space one may include memory in the model, or apply party specific boundary conditions. Inclusion of such agent-specific force strengths and interaction radii is planned for further research.

Another aspect of political dynamics not examined here is that some political issues may be more controversial than others. While we chose to examine the problem of a domain $[0,1]^d$ in this paper, it may be that some issues require more distance to be traversed than others in order for consensus to be achieved. For example, a rectangular domain $[0,L_1]\times [0,L_2],$ with $L_2\gg L_1$ would mean that convergence is easier along the smaller dimension because there is less space to traverse. Investigating whether altering the domain significantly alters the model behaviour is left for future work.

\section*{Data availability statement}
All code and data used to produce the figures are available from the GitHub repository: \url{https://github.com/PatrickhCahill/ModifiedHegselmannKrauseModel}.

\bibliographystyle{elsarticle-num-names}



\appendix
\section{Prototypical scenarios for a three party system}
Here we show that the 5 political scenarios of party-base, swing voters, political competitions, disaffected voters, and voter consensus, which we observed for the two-party system shown in Figure~\ref{fig:2party-example_behaviours}, are also observed in a system consisting of multiple parties. Figure~\ref{fig:3party-example_behaviours} shows simulations for a three-party system. 

\begin{figure}[htbp!]
    \centering
        \begin{subfigure}[htbp]{0.48\textwidth}
            \includegraphics[width=\linewidth]{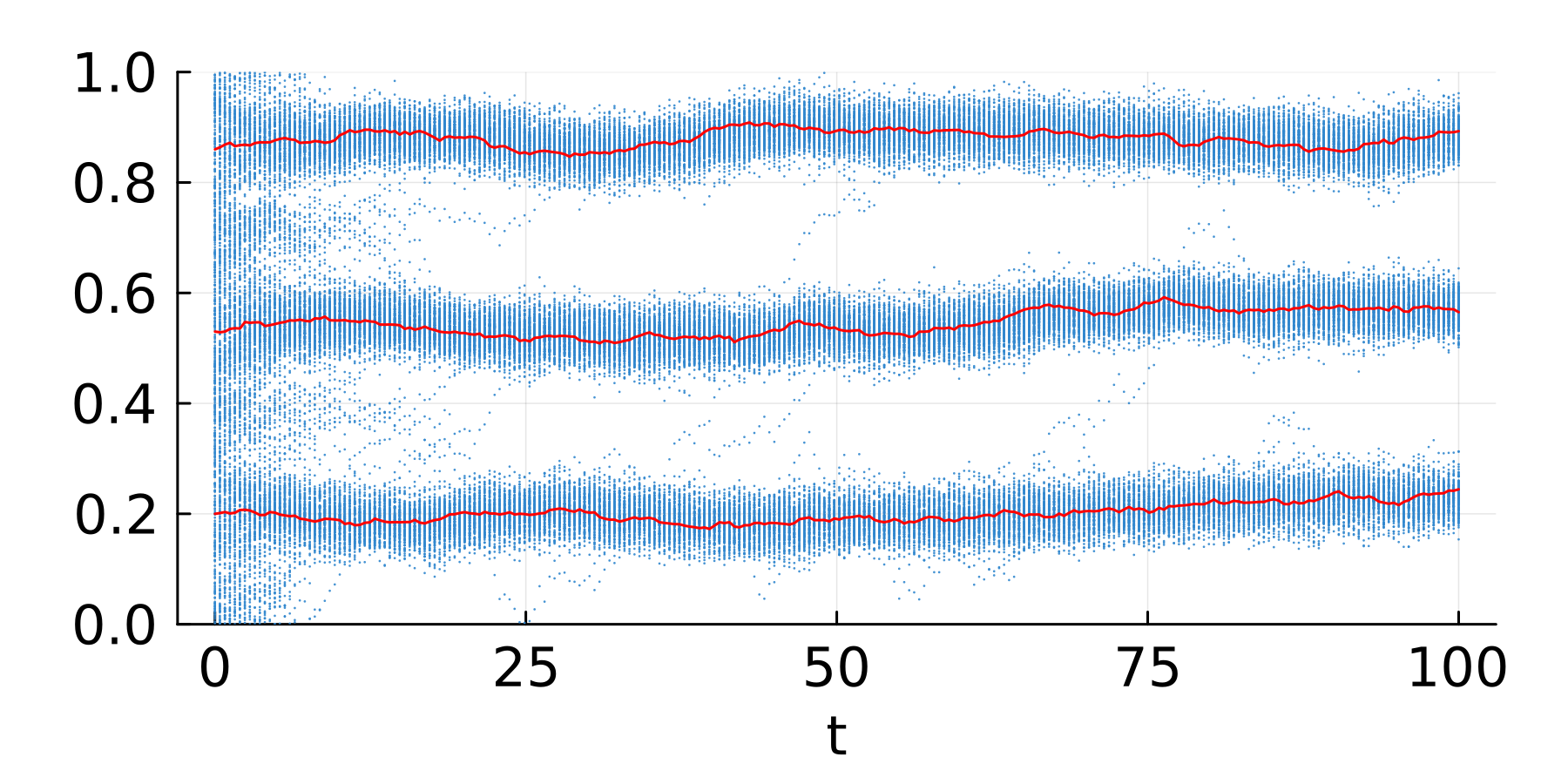}
            \caption{\textit{Party-base}}
            \label{fig:3party-political_base}
    \end{subfigure}
    \hfill
    \begin{subfigure}[htbp]{0.48\textwidth}  
        \centering
        \includegraphics[width=\textwidth]{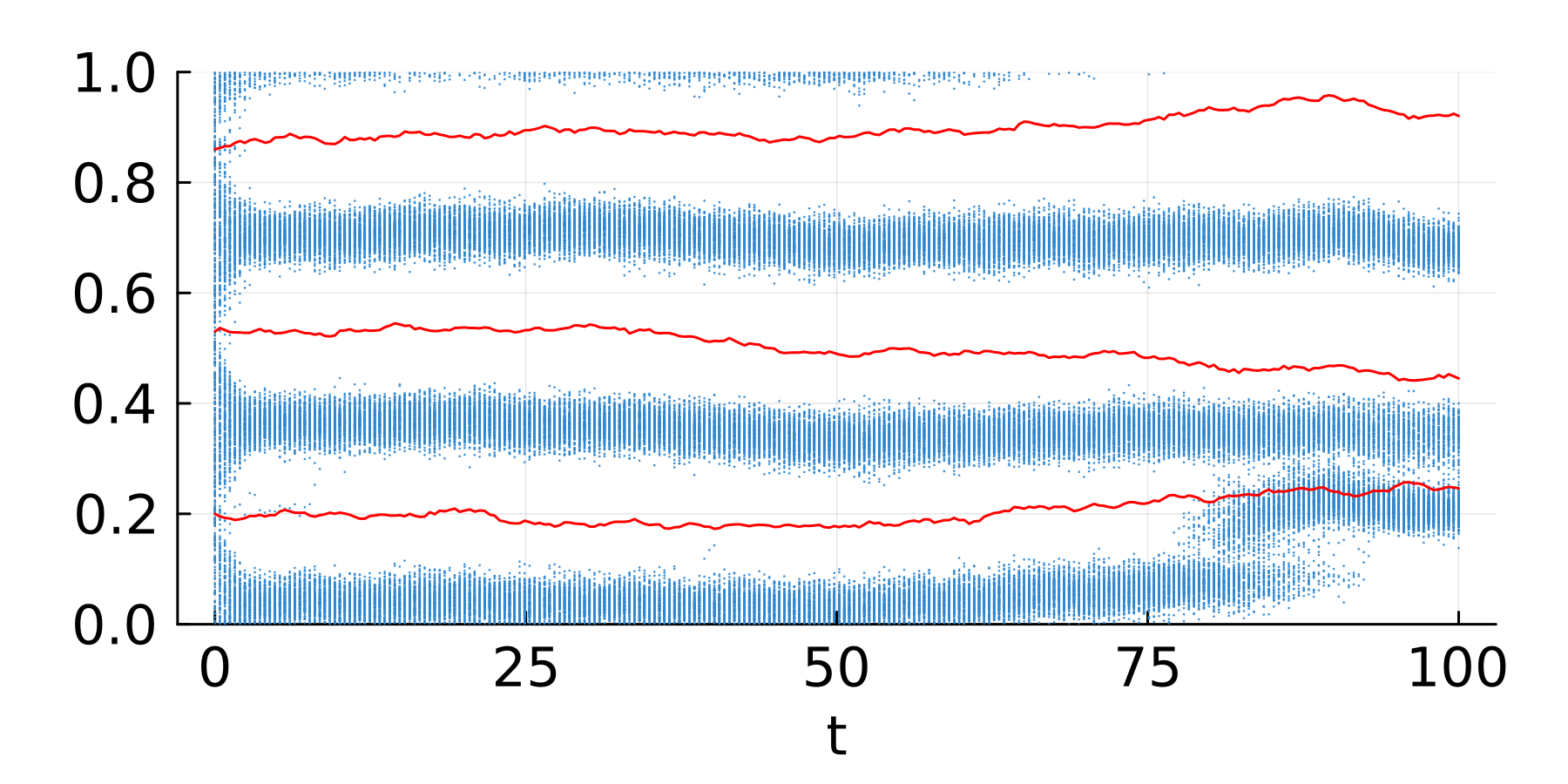}
        \caption{\textit{Swing voters}. Here $R_{pv} = 0.35$.}
        \label{fig:3party-swing_voters}
    \end{subfigure}
    \vskip\baselineskip
    \begin{subfigure}[htbp]{0.48\textwidth}
        \centering
        \includegraphics[width=\textwidth]{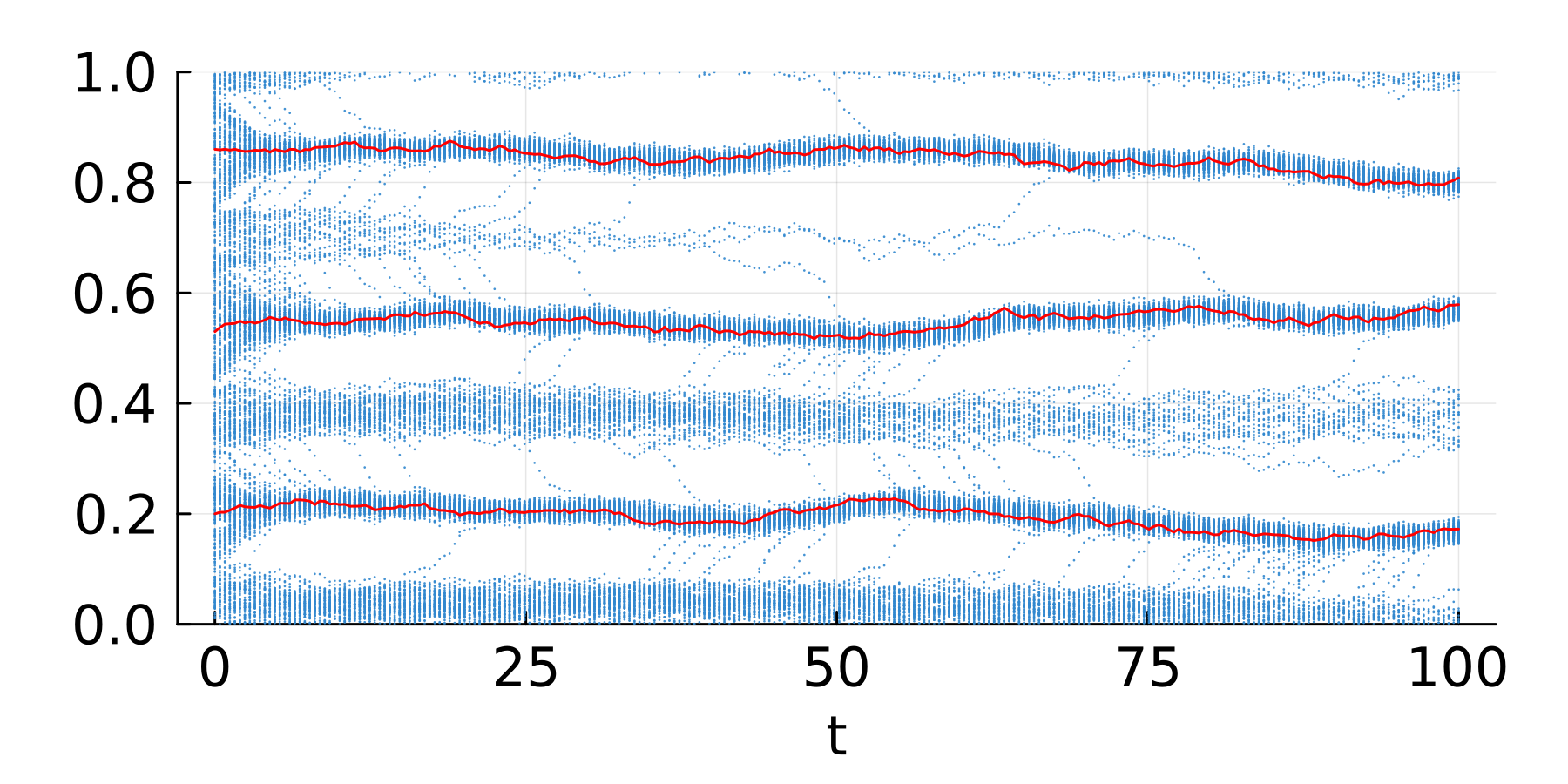}
        \caption{\textit{Political competition}. Here $\sigma_v = 0.01.$}
        \label{fig:3party-political_comp_1}
    \end{subfigure}
    \hfill
    \begin{subfigure}[htbp]{0.48\textwidth}  
        \centering
        \includegraphics[width=\textwidth]{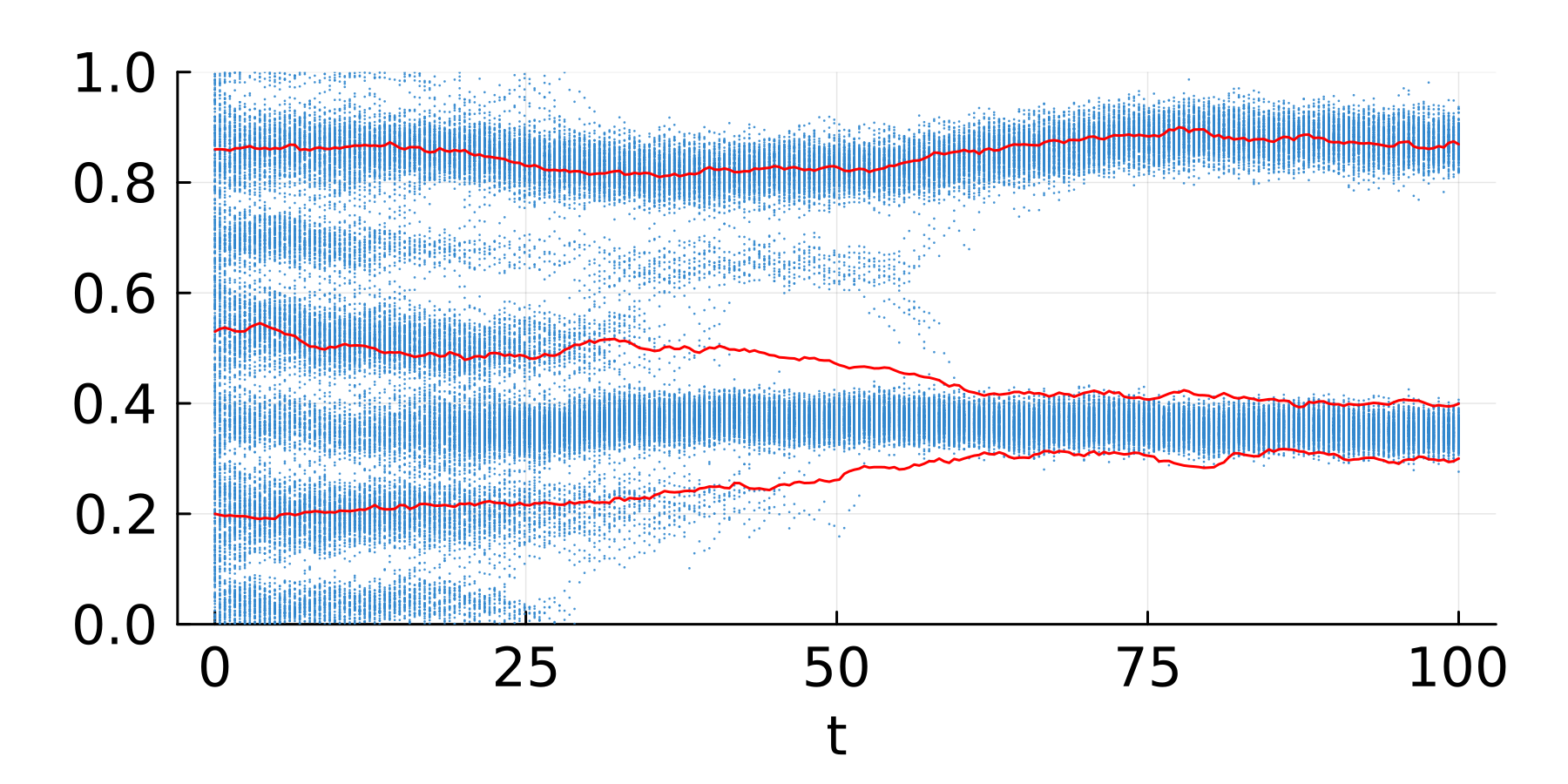}
        \caption{\textit{Political competition}: Here $R_{pv}= 0.23$.}
        \label{fig:3party-political_comp_2}
    \end{subfigure}
    \vskip\baselineskip
    \begin{subfigure}[htbp]{0.48\textwidth}   
        \centering 
        \includegraphics[width=\textwidth]{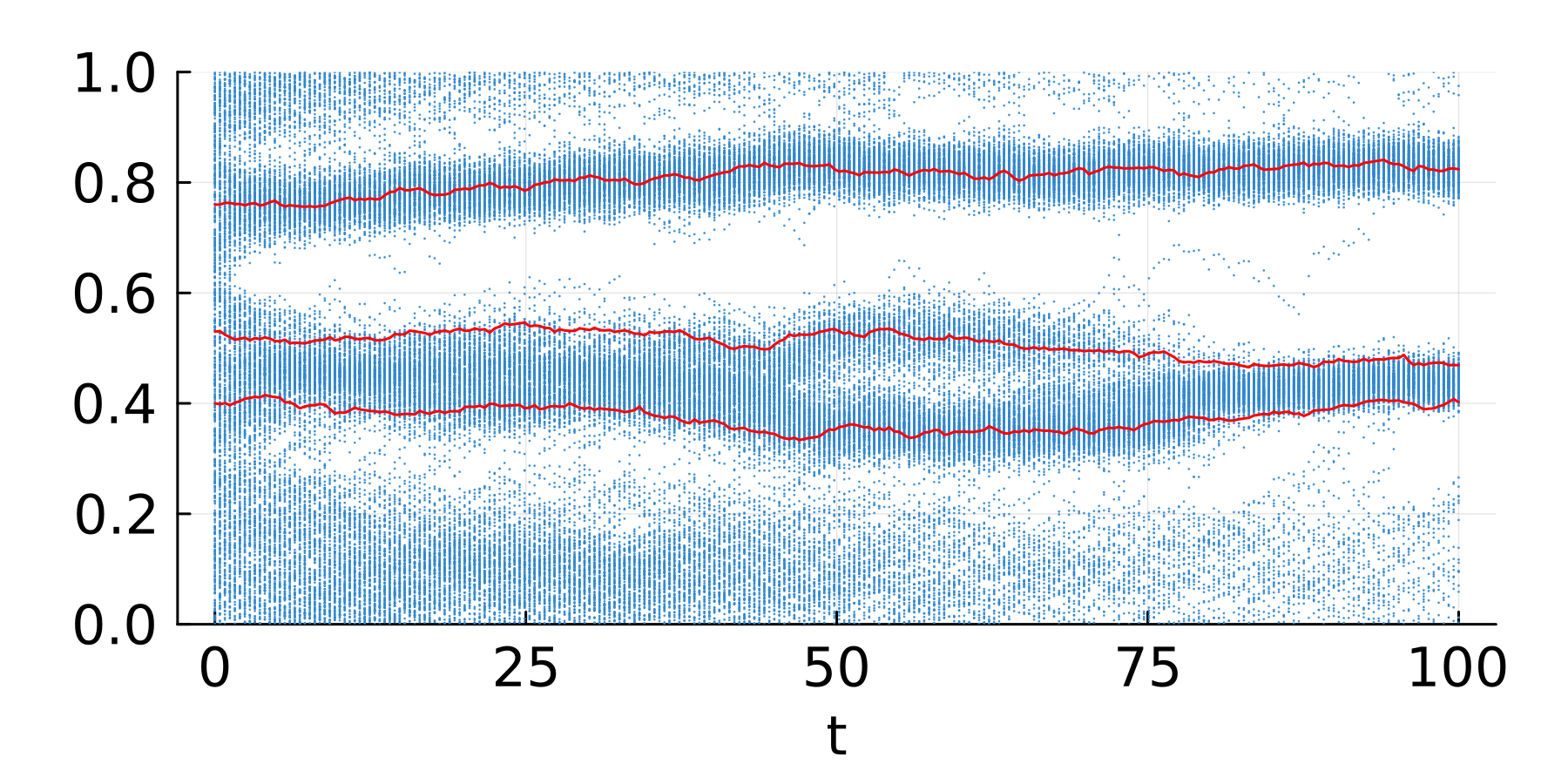}
        \caption{\textit{Disaffected voters}: Here $p_1(0)=0.4$, $p_2(0)=0.53$ and $p_3(0)=0.76$.} 
        \label{fig:3party-disaffected_voters}
    \end{subfigure}
    \hfill
    \begin{subfigure}[htbp]{0.48\textwidth}   
        \centering 
        \includegraphics[width=\textwidth]{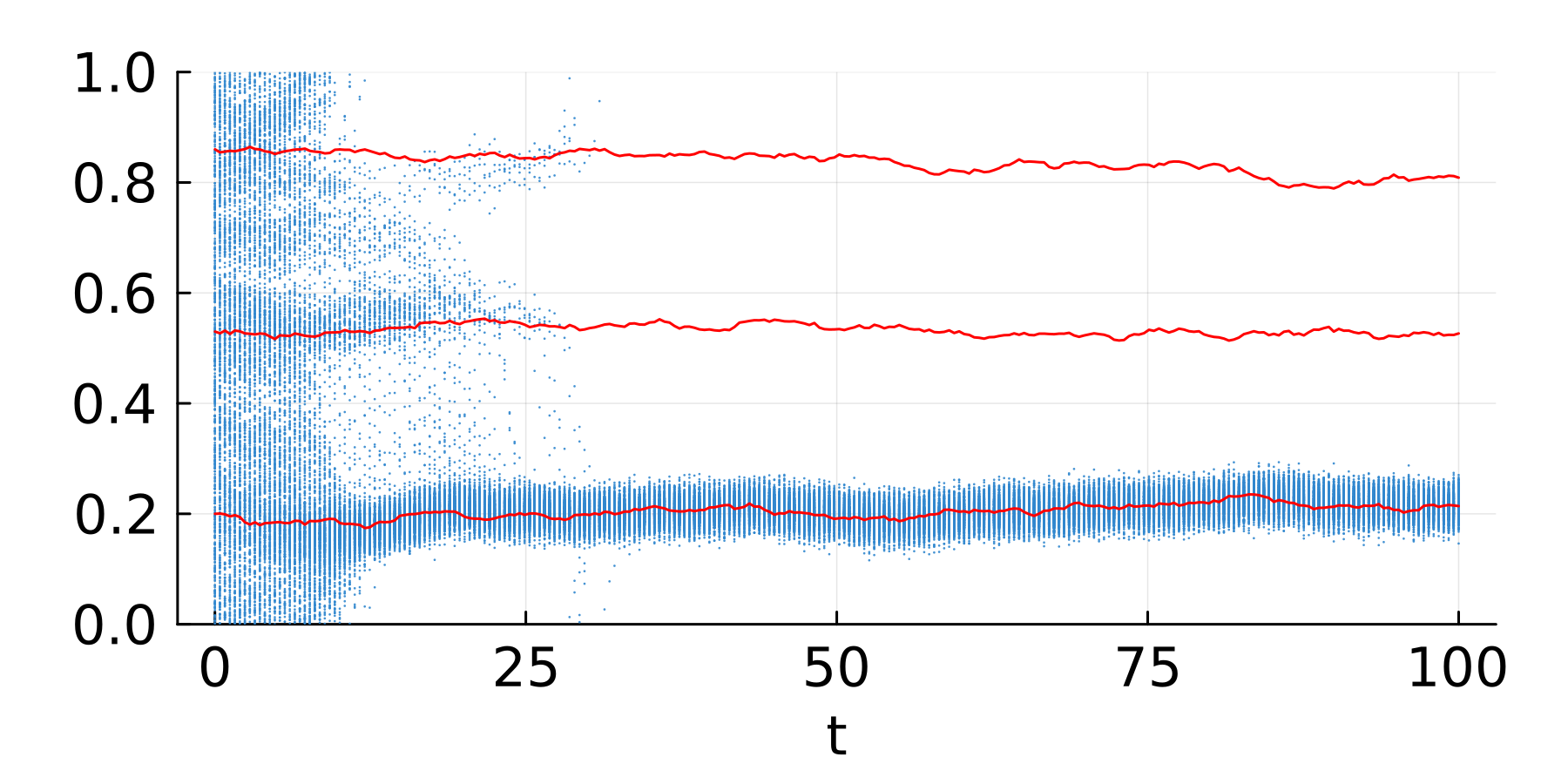}
        \caption{\textit{Voter consensus}. Here $R_{vv} = 0.3$} 
        \label{fig:3party-voter consensus}
    \end{subfigure}
    \caption{
    Prototypical political scenarios in a three-party system. At time $t=0$, $N_v = 1,000$ voters are distributed uniformly across $[0,1]$ and parties are initially at $p_1(0)=0.2$, $p_2(0)=0.53$ and $p_3(0)=0.86$  for Figures~\ref{fig:3party-swing_voters}-\ref{fig:3party-voter consensus}. The strengths of the voter forces are $\mu_{vv}=\mu_{pv}=1$. The strengths of the party forces are $\mu_{vp}=0.03$ and $\mu_{pp}=0.01$. The interaction radii are $R_{vv}=R_{pv}=R_{pp}=0.1$ and $ R_{vp}=0.2$. The noise strengths are $\sigma_v = 0.03$ and $\sigma_p = 0.005$.
    }
    \label{fig:3party-example_behaviours}
\end{figure}
\begin{figure}[htbp!]
    \centering
    \includegraphics[width=0.5\linewidth]{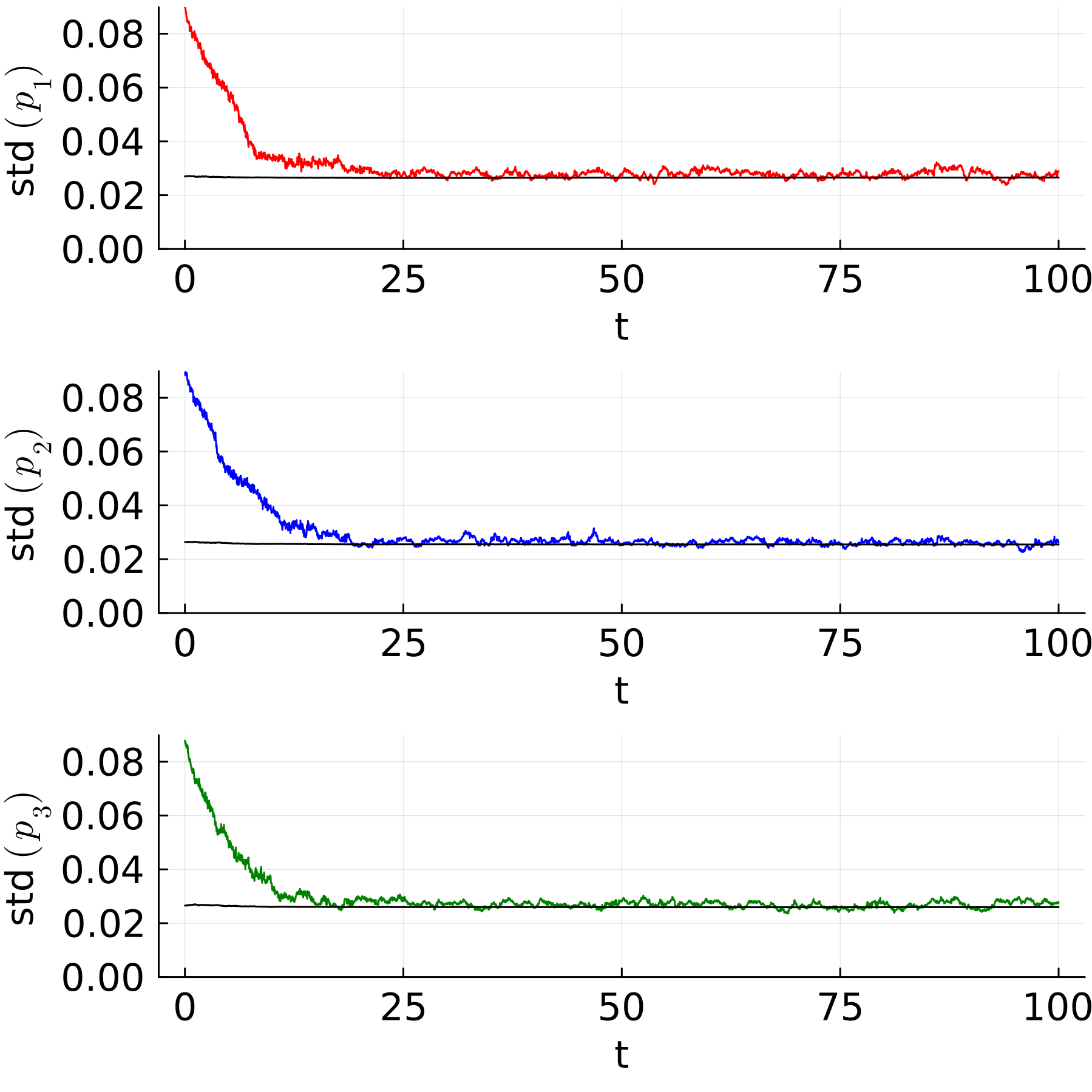}
    \caption{Standard deviations of the voter cluster centred around each of the three parties in Figure~\ref{fig:3party-political_base}. The black line denotes the analytical approximation \eqref{eq:stdOU}. Top: party $p_3$ with $p_3(0)=0.86$. Middle: party $p_2$ with $p_2(0)=0.53$. Bottom: party $p_1$ with $p_1(0)=0.2$.}
    \label{fig:3party-deltacluster_check}
\end{figure}

Figure~\ref{fig:3party-political_base} shows an example for a party base formation around each of the three parties  with $N_v = 1,000$ voters which initially are distributed uniformly across $[0,1]$ and with $3$ parties which are initially located in opinion space at $p_1(0)=0.2$, $p_2(0)=0.53$ and $p_3(0)=0.86$. The strengths of the voter forces were chosen as $\mu_{vv}=\mu_{pv}=1$ and those of the party forces as $\mu_{vp}=0.03$ and $\mu_{pp}=0.01$. The interaction radii are $R_{vv}=R_{pv}=R_{pp}=0.1$ and $R_{vp}=0.2$. The noise strengths are $\sigma_v = 0.03$ and $\sigma_p = 0.005$, which again ensures that the voter dynamics occurs much faster than the party dynamics. 

Figure~\ref{fig:3party-deltacluster_check} shows the standard deviations of voters in the three party base clusters corresponding to the party base case depicted in Figure~\ref{fig:3party-political_base}. We show results of the numerical simulation of the modified Hegselmann--Krause model \eqref{eq:HKa}-\eqref{eq:HKb} as well as the prediction of the cluster size $\delta_{\rm{cl}}$ given by \eqref{eq:deltacluster}. Since the voters form three approximately equal clusters, we expect that $N_{v}^{(c)} \approx \frac{N_v}{3}$ and $ N_{p}^{(c)}=1$ (as voters cluster around a single isolated party), implying $\delta_{\rm{cl}}=4\, \rm{std}{\rm OU}=0.104$. Note that here the cluster is covered by the interaction radii with $\delta_{\rm{cl}}<2R_{vv}=2R_{pv}=0.2$. To numerically estimate $\delta_{\rm cl}$ for the full modified Hegselmann--Krause model, we set $N_v^{(c)}$ to be the number of voters which are within a distance of $0.15$ away from the respective parties $p_\alpha$. The threshold $0.15$ is chosen because the region $[p_\alpha-0.15, p_\alpha+0.15]$ for $\alpha = 1,2,3$ contains the cluster around $p_\alpha$ since the interaction radii are sufficiently small with $R_{vv}=R_{pv}=0.1$. The region is also sufficiently small that it does not contain any other clusters. Figure~\ref{fig:3party-deltacluster_check} shows that our expression \eqref{eq:deltacluster} well approximates the observed cluster size as measured by the standard deviations.

As in the two-party system, the stable party bases break down if parties begin to compete for voters (or their interaction radius $R_{pv}$ is increased with $||p_\alpha-p_\beta|| < {\rm{max}}(R_{pv},R_{vv})$ (cf \ref{eq:condswing}), and swing voter states may occur as shown in Figure~\ref{fig:3party-swing_voters}. Here, the parameters are as for the party base scenario but for a larger party interaction radius $R_{pv}=0.35$. The size of swing voter clusters $\delta_{\rm{cl}}$ is also given by \eqref{eq:deltacluster} with $N^{(c)}_p=2$ because the party above and below are both within the interaction radius of the voter cluster. We find $\delta_{\rm{cl}}=0.085$ which approximates well the observed cluster size. The state of swing voters can be disturbed by the party dynamics as shown in Figure~\ref{fig:3party-swing_voters} around $t\approx 80$. The swing voter cluster between parties $p_1$ and $p_3$ is broken up after the middle party $p_2$ drifted sufficiently towards the swing voter cluster which now interacts with all three parties, eventually leading to the formation of a party base cluster around $p_1$. Note the coexistence of party-base dynamics and swing voters in Figure~\ref{fig:3party-swing_voters}. It is also possible for party base clusters to break up and form swing voter clusters when parties move closer to each other. An example of this is shown in Figure~\ref{fig:3party-disaffected_voters} around $t\approx 70$ when parties $p_1$ and $p_2$ move closer together and their respective party bases merge to form a single swing voter cluster. 

The three-party system also supports political competition through their party dynamics similar to the cases described in Section~\ref{sec:examples}. Figure~\ref{fig:3party-political_comp_1} shows clusters of swing voters that are slowly entrained by one of their two closest parties. We again observe transitory coexistence of party bases and swing voter clusters. Here $\sigma_v = 0.01$ ensures the voter behaviour is slower than in Figure~\ref{fig:3party-political_base}. Swing voters decide to join a particular party base either by their individual stochastic slow exploration of the opinion space or by parties moving towards them on a faster time scale. The latter scenario may be viewed as a form of party competition to attract more voters. 

Analogously to Figure~\ref{fig:2party-political_comp_2}, in Figure~\ref{fig:3party-political_comp_2}  the voters around the bottom two parties $p_1$ and $p_2$ alternate between party base and swing voter behaviour. A party base around the top party $p_3$. There is further competition between the top and middle parties $p_3$ and $p_2$ over a swing voter cluster that dissolves around $t\approx 60.$ 

Figure~\ref{fig:3party-disaffected_voters} shows an example with disaffected voters. Here the same parameters are used as for the party base clusters in Figure~\ref{fig:3party-political_base} but with different initial conditions for the parties allowing for unoccupied political space. As opposed to the example of disaffected voters in the two-party case in Section~\ref{sec:examples}, the disaffected voters do not form a localized cluster here of size $\delta{\rm{cl}}$ since $R_{vv}$ is too small.  

Lastly, voter consensus behaviour is depicted in Figure~\ref{fig:3party-voter consensus} where a single party base cluster forms around party $p_1$.

\end{document}